\documentclass[review]{elsarticle}

\usepackage{lineno}
\usepackage{hyperref}
\modulolinenumbers[5]

\journal{Journal}

\biboptions{authoryear}

\usepackage{graphicx}
\usepackage{geometry}
\usepackage{amsmath} 
\usepackage{amssymb}
\usepackage{amsthm}
\theoremstyle{plain}
\newtheorem{theorem}{Theorem}
\newtheorem{corollary}[theorem]{Corollary}
\newtheorem{lemma}[theorem]{Lemma}
\newtheorem{proposition}[theorem]{Proposition}

\newtheorem*{theorem*}{Theorem}
\newtheorem*{corollary*}{Corollary}
\newtheorem*{lemma*}{Lemma}
\newtheorem*{proposition*}{Proposition}

\usepackage{bm}
\usepackage{bbm}

\usepackage{subcaption}% Support for small, `sub' figures and tables

\usepackage{mathtools}

\usepackage{multirow}

\usepackage[flushleft]{threeparttable}

\usepackage{tabularx}
\newcolumntype{Y}{>{\centering\arraybackslash}X}
\newcolumntype{Z}{>{\raggedleft\arraybackslash}X}

\newcommand{\T}{\mathsf{T}} % for transposition
\date{August 10, 2020}

\begin{document}
%%%%%%%%%%%%%%%%%%%%%%
%\thispagestyle{empty} 
%\tableofcontents
%\newpage
%\setcounter{page}{1}
%%%%%%%%%%%%%%%%%%%%%%

\begin{frontmatter}
\title{Centralizing-Unitizing Standardized High-Dimensional Directional Statistics and Its Applications in Finance}


\author[mymainaddress]{Yijian Chuan}
\ead{ychuan@pku.edu.cn}

\author[mymainaddress,mysecondaryaddress]{Lan Wu\corref{mycorrespondingauthor}}
\cortext[mycorrespondingauthor]{Corresponding author}
\ead{lwu@pku.edu.cn}

\address[mymainaddress]{School of Mathematical Sciences, Peking University, Beijing, China}
\address[mysecondaryaddress]{Key Laboratory of Mathematical Economics and Quantitative Finance, Peking University,	Beijing, China}
\begin{abstract}
Cross-sectional ``Information Coefficient'' (IC) is a widely and deeply accepted measure in portfolio management. The paper gives an insight into IC in view of high-dimensional directional statistics: IC is a linear operator on the components of a centralizing-unitizing standardized random vector of next-period cross-sectional returns. Our primary research first clearly defines IC with the high-dimensional directional statistics, discussing its first two moments. We derive the closed-form expressions of the directional statistics' covariance matrix and IC's variance in a homoscedastic condition. Also, we solve the optimization of IC's maximum expectation and minimum variance. Simulation intuitively characterizes the standardized directional statistics and IC's p.d.f.. The empirical analysis of the Chinese stock market uncovers interesting facts about the standardized vectors of cross-sectional returns and helps obtain the time series of the measure in the real market. The paper discovers a potential application of directional statistics in finance, proves explicit results of the projected normal distribution, and reveals IC's nature.
\end{abstract}

\end{frontmatter}
%\linenumbers

\section{Introduction}

%

%1)
Directional statistics concerns a population of unit multivariates in $ L^{2} $ norm, or, equivalently, on a unit hypersphere.
In some application scenarios of high-dimensional statistics, people usually standardize the multivariate data by \textit{centralizing and unitizing} the data taken from general $ \mathbb{R}^{n} $ like 
$ (\bm{z}-\widehat{\mu}) / \|\bm{z}-\widehat{\mu}\| $.
Such centralizing-unitizing standardized multivariate data is a particular direction statistic orthogonal to the vector whose elements are all $1$.
The dimension is often higher than $ 3 $ and even more than $ 100 $, particularly in finance.
There is little literature about the probabilistic characteristics and statistical inference of direction statistics on a high-dimensional sphere, let alone such centralizing-unitizing standardized multivariate concerned in the paper.

%2)
This paper's original idea is motivated by questions raised in the industry: ``For a running strategy, what kind of $ \mathrm{IC} $ time series can be considered an invalid strategy? 
Is there a feasible statistical inference method of $ \mathrm{IC} $ data?" 
$ \mathrm{IC} $ is the abbreviation of Information Coefficient of the predicted and actual return values.
The first and fundamental thing is to define $ \mathrm{IC} $ in a probabilistic theoretical viewpoint and to reveal the probabilistic and statistical properties.
In a probabilistic view, we define $ \mathrm{IC} $ as a linear operator on a centralizing-unitizing standardized multivariate.
The $ \mathrm{IC} $ measures the similarity between the series of the predicted and actual values, like a metric function.
In active portfolio management, $ \mathrm{IC} $ can be regarded as the weighted sum of the standardization of next-period cross-sectional returns.
In the rest of this paper, we denote $ \mathrm{IC} $ by $ T_{\bm{\chi}(\bm{Z})}(\bm{\theta}) $, where $ \bm{Z} $ is a random vector on $ \mathbb{R}^{n} $, $ \bm{\chi}(\bm{Z}) $ is the standardization of $ \bm{Z} $ on a hypersphere, and $ \bm{\theta} $ is the predicted values and the weights.
The precise definition is in \eqref{equ:T} of Section~\ref{sec:models}.
The time series of $ \mathrm{IC} $ is an interesting topic in academia and industry.
\citet{coggin1983} thought that it originated from sample errors, and used ``meta-analysis'' to correct it.
After that, \citet{qian2004active} re-distinguished the definition of raw $ \mathrm{IC} $  and risk-adjusted $ \mathrm{IC} $ , regarding $ \mathrm{IC} $'s volatility as ``strategy risk''.
\citet{ye2008ICt} also investigated the impact of $ \mathrm{IC} $'s variation on the performance of investment specifically.
\citet{ding2017flam} systematically developed a stationary econometric model of $ \mathrm{IC} $  for mean-variance portfolios.
The dynamic modeling of realized $ \mathrm{IC} $  is an essential and lasting issue in financial investment research.

%3) Literature about DS, particularly high-dimensional. the meaning of Z to X.
Since a pioneer work of \citet{fisher1953}, directional statistics has flourished to a certain degree.
\citet{mardia1972statistics}, \citet{mardia2000directional}, and \citet{ley2017modern, ley2018applied} summarized and developed the main works.
The current development of directional statistics helps us understand $ \mathrm{IC} $'s nature, and provides some important methods and tools.
However, due to $ \mathrm{IC} $'s high-dimensional essence, it is not easy to directly apply the existing results and methods.
A natural idea is first to explore the directional statistical properties of a standardized high-dimensional normal distribution.
Even in a normal distribution background, it is also a challenging problem.
One of the most cutting-edge research concerned in the paper is \citet{presnell2008MRL_PN}, which obtained the closed-form expressions of the mean direction ($ \mathrm{MD} $) and mean resultant length ($ \mathrm{MRL} $) of the projected normal distributed $ \frac{\bm{Z}}{\|\bm{Z}\|} $ in the homoscedastic condition of $ \bm{Z} \sim N \left ( \cdot, \sigma^{2}I \right ) $.

% 4) main works
% 4.1) models and theoretical results
In this paper, we set up models and develop their probabilistic properties.
First, we clarify the standardized random vector and its linear combination. 
We define the centralizing and unitizing standardization as a mapping $ \bm{\chi} $ and generate the standardized random vector $ \bm{\chi}(\bm{Z}) $ on a hypersphere from the pre-standardized $ \bm{Z} $ on $ \mathbb{R}^{n} $, which is different from the projected random vector $ \frac{\bm{Z}}{\|\bm{Z}\|} $.
Given the standardized weight $ \bm{\theta} $, we define the linear combination of the standardized random vector's components $  T_{\bm{\chi}(\bm{Z})}(\bm{\theta}) $, standing for $\mathrm{IC} $ in finance.
Second, we develop their probabilistic properties. 
We propose a representation theorem to connect the standardized random vector $ \bm{\chi}(\bm{Z}) $ with a projected one $ \frac{\bm{Z}^{\prime}}{\|\bm{Z}^{\prime}\|} $.
In a homoscedastic condition,
we express $ \bm{\chi}(\bm{Z}) $'s $ \mathrm{MD} $/$ \mathrm{MRL} $ and covariance matrix explicitly.
We then derive $ T_{\bm{\chi}(\bm{Z})}(\bm{\theta}) $'s expectation and variance from $ \bm{\chi}(\bm{Z}) $'s first two moments and their closed-form expressions in the homoscedastic condition.
Third, we discuss the optimization problems of $ T_{\bm{\chi}(\bm{Z})}(\bm{\theta}) $'s expectation and variance:
$ \mathbb{E} T_{\bm{\chi}(\bm{Z})}(\bm{\theta}) $'s maximum is $ \bm{\chi}(\bm{Z}) $'s $ \mathrm{MD} $,
$ \mathrm{var} \left( T_{\bm{\chi}(\bm{Z})}(\bm{\theta}) \right) $'s minimum is $ \mathrm{cov}\left (\bm{\chi}(\bm{Z}) \right ) $'s second smallest eigenvalue, and the two optimizations share the same solution in the homoscedastic condition.

% 4.2) simulation
We carry out the numerical simulation, aiming to supplement the probabilistic properties.
Due to the standardized distribution's inherent complication, we can hardly get an intuitive sense of the theory.
We illustrate the impact of the parameters of $\bm{\mu} $ and $ \Sigma$ on the $ \mathrm{MD} $ intuitively in $ 3 $-dim cases and simulate the p.d.f.'s of $ T_{\bm{\chi}(\bm{Z})}(\bm{\theta}) $ with real market parameters, from which we analyze the impact of the dimension and heteroscedasticity, meeting the previous conclusion in \cite{coggin1983} and \cite{grinold1994alpha}.
We also compare the approximate $ \mathrm{MRL} $ in theory and its true value in simulation.

% 4.3) empirical study
We implement empirical studies based on the Chinese stock market to explore the high-dimensional statistics in the real market.
For one thing, we assume the returns to be i.i.d., and clarify their statistical properties.
Specifically, we first explore the descriptive statistics of the sample of the standardized vectors $ \left \{ \bm{x}_{t} \right \} $ in a traditional approach in directional statistics, such as the sample $ \mathrm{MD} $/$ \mathrm{MRL} $ and the scatter matrix.
The high- and low-value components of the sample $ \mathrm{MD} $ vary in different patterns, indicating the distinct behaviors between the high- and low-return stocks.
The scatter matrix's eigenvalues are uneven in the sense that the sum of the $ 20 $ largest eigenvalues accounts for more than $ 50\% $ of the sum of the all $ 185 $ ones,
implying the standardized returns vary along with few specific directions. 
Then, we reduce the dimension of the sample $ \left \{\bm{x}_{t} \right \} $, by taking their inner product with a vector $ \widehat{\bm{\iota}} $ such as the sample $ \mathrm{MD} $, and generate the $ 1 $-dim sample $ \left\{\widehat{\bm{\iota}}^{\T}\bm{x}_{t}\right\} $.
We illustrate the descriptive statistics, histograms, and box plots of $ \left\{\widehat{\bm{\iota}}^{\T}\bm{x}_{t}\right\} $.
The market's centralization varies in different time windows, 
and $\left \{ \bm{x}_{t}\right \} $ may be multimodal distributed.
The sample $ \mathrm{MRL} $ of $ \left \{ \bm{x}_{t} \right \} $ and the sample standard deviation of $ \left \{ \bm{\iota}^{\T}\bm{x}_{t} \right \} $ result in opposing conclusions about the variance, which could be interpreted as their different perspectives.
What's more, we analyze the connection between pre- and post-standardized $ \left \{ \bm{z}_{t} \right \} $ and $ \left \{ \bm{x}_{t} \right \} $: The sample correlation coefficient between $ x_{it} $ and $ x_{jt} $ is significantly lower than that between $ z_{it} $ and $ z_{jt} $ on average, 
which, in finance, means that the standardization eliminates the \textit{beta} part of the returns.
For another, we analyze their time-series properties.
In detail, we illustrate the time series of the representative $ z_{it} $ and $ x_{it} $, discovering similar volatility clustering phenomena.
The standardized $ x_{it} $ has more information than the pre-standardized $ z_{it} $.
The time series of the estimate of $ \mathrm{IC} $ is compared with the cross-sectional standard deviation, revealing different market information.

%5)
The rest of the paper is organized as follows.
Section~\ref{sec:models} defines the models and deduces the probabilistic properties.
Section~\ref{sec:simulation} applies simulation to supplement and corroborate the theory.
Section~\ref{sec:empirical} implements empirical studies in the Chinese stock market.
Section~\ref{sec:conclusion} concludes.

\section{Models and Theoretical Results}\label{sec:models}

Given the probability space $ (\Omega, \mathcal{F}, \mathbb{P}) $,
directional statistics concerns the random vector $ \bm{X} \in \mathcal{S}^{n-1} := \left\{ \bm{x}\in \mathbb{R}^{n} \big| \|\bm{x}\| =1 \right\} $, where $ \|\bm{x}\| := \sqrt{\sum_{i=1}^{n} x_{i}^{2}} $.
We introduce two basic numerical characteristics:
the mean direction ($ \mathrm{MD} $) and the mean resultant length ($ \mathrm{MRL} $) of $ \bm{X} $:
\begin{align*}
\mathrm{MD}(\bm{X})
:=&
\frac{\mathbb{E}\bm{X}}{\|\mathbb{E}\bm{X}\|}
,
\\
\mathrm{MRL}(\bm{X})
:=&
\|\mathbb{E}\bm{X}\|
,
\end{align*}
where the $\mathrm{MD} $ plays the role of the expectation and the $\mathrm{MRL} $ the variance %%
(see \citealp[\S 9.2, p. 163-164]{mardia2000directional} and \citealp[\S 1.3, p. 10-12]{ley2017modern} for more details).
For simplicity, we assume that, in the paper, all the $ \mathrm{MRL} $'s are strictly greater than $ 0 $, so the $ \mathrm{MD} $'s are well-defined.
One of the most common approaches to set up a random vector on $ \mathcal{S}^{n-1} $ relies on the following projection function:
\begin{equation}
\begin{aligned}
\mathbb{R}^{n}\setminus\{\bm{0}\}
\rightarrow&\ 
\mathbb{R}^{n}
,
\\
\bm{z} 
\mapsto&\ 
\frac{\bm{z}}{\|\bm{z}\|}
.
\end{aligned}
\label{equ:DP}
\end{equation}
Then, 
for a random vector $ \bm{Z} \in \mathbb{R}^{n} $,
we have
$ \frac{\bm{Z}}{\|\bm{Z}\|} \in \mathcal{S}^{n-1}$, and call it has a projected distribution.
Particularly, if $ \bm{Z} $ has a multivariate normal distribution,
then $ \frac{\bm{Z}}{\|\bm{Z}\|} $ has a projected normal distribution (PND, see \citealp[\S 1.5, p. 12 and \S 3.5.6, p. 46]{mardia2000directional}).

\subsection{Model Setup and Basic Probabilistic Properties}
In some studies concerning high-dimensional random vectors taking values in $ \mathbb{R}^{n} $, 
people standardize the data cross-sectionally by centralization and unitization, say $ \frac{z - \mu(z)}{\sigma(z)} $.
The standardized random vectors are a particular kind of random vectors in directional statistics.
It is fundamentally different from the function \eqref{equ:DP}
because the \textit{centralization} leads to the natural singularity of its covariance matrix.
This subsection 
discusses the connection between pre-standardized high-dimensional random vectors in $ \mathbb{R}^{n} $ and the standardized direction statistics in $\mathcal{S}^{n-1}$,
and explore the properties of the linear combination of the components of the directional statistics.

\subsubsection{The Standardized Random Vector in Directional Statistics:
	\texorpdfstring{$ \bm{\chi}(\bm{Z}) $}{$ chi(Z) $}
}\label{subsec:chi_Z}

We define the standardization function,
introduce the particular random vector, 
and then clarify its basic properties with $ 2 $-dim cases.

To begin with, we define the \textit{centralizing and unitizing} standardization function $ \bm{\chi} $.
Define the $ n $-dim centering matrix $ P $ as
\begin{align}
P
:=&\ 
I - \frac{1}{n}\bm{1}\bm{1}^{\T}
,
\label{equ:P}
\end{align} 
where $ I $ is an $ n $-dim identity matrix and $ \bm{1} := (1, 1, \cdots, 1)^{\T} \in \mathbb{R}^{n} $.
Given $ \bm{z} \in \mathbb{R}^{n} $, 
$ P\bm{z} = \bm{z} - \overline{z} \bm{1} $ is similar to $ \bm{z} - \mu(\bm{z}) $ subtracting the mean from every component of a vector,
and
$ \frac{1}{\sqrt{n}}\|P\bm{z}\| = \sqrt{\frac{1}{n}\sum_{j=1}^{n} \left(z_{j} - \overline{z}\right)^{2}} $ to $ \sigma(\bm{z}) $,
where $ \overline{z} := \frac{1}{n}\sum_{i=1}^{n}z_{i} $.
Thus, 
$ \frac{P\bm{z}}{\|P\bm{z}\|} $ is the result of the \textit{centralizing and unitizing} standardization process of $ \bm{z} $.
For convenience and conciseness, 
we define the standardization process formally as a multivariate function:
\begin{equation}
\begin{aligned}
\bm{\chi}: 
\mathbb{R}^{n} \setminus \langle \bm{1} \rangle \rightarrow&\ 
\mathbb{R}^{n},
\\
\bm{z}
\mapsto&\ 
\frac{P\bm{z}}{\|P\bm{z}\|},
\end{aligned}
\label{equ:chi}
\end{equation}
where $ \langle \bm{1} \rangle := \left\{ \bm{x}\in\mathbb{R}^{n} \big| \bm{x} = a\bm{1}, a\in \mathbb{R} \right\} $ is excluded in the domain for well-definition. 

Then, we introduce the standardized random vector $ \bm{\chi}(\bm{Z}) $.
Give a random vector $ \bm{Z} \in \mathbb{R}^{n} $, such that its expectation $ \mathbb{E}\bm{Z} =: \bm{\mu} $, and its covariance matrix $ \mathrm{cov}(\bm{Z}) =: \Sigma $ ($ \Sigma $ is positive definite).
We are interested in the random vector
\begin{align}
\bm{\chi}(\bm{Z}) = \frac{P\bm{Z}}{\|P\bm{Z}\|},
\label{equ:chi_Z}
\end{align}
which is the standardized random vector from $ \bm{Z} $ with function \eqref{equ:chi}.

Last, we clarify the basic properties of $ \bm{\chi}(\bm{Z}) $ in the following five points.
First,
the support of the random vector $ \bm{\chi}(\bm{Z}) $ is 
$ \mathcal{S}^{n-1} \cap \langle\bm{1}\rangle^{\bot} $,
where $ \langle\bm{1}\rangle^{\bot} := \left\{
\bm{x} \in \mathbb{R}^{n} | \bm{x}^{\T}\bm{1} = 0
\right\} $,
which is
the intersection of the unit sphere centered at the origin and the hyperplane with normal $ \bm{1} $ passing through the origin.
Thus, 
$ \bm{\chi}(\bm{Z}) $ is a particular random vector in directional statistics
subjecting to $ \bm{1}^{\T} \bm{\chi}(\bm{Z}) = 0 $.

To illustrate the generality and particularity of $ \bm{\chi}(\bm{Z}) $ intuitively,
we show the supports of $ \bm{\chi}(\bm{Z}) $ of $ n =2 $ and $ 3 $ in Figure~\ref{fig:X_marginal_dist}.
In detail,
subfigure~\ref{subfig:X_marginal_dist._2} portrays the $ n = 2 $ case: The dashed blue circle is $ \mathcal{S}^{1} $, and the purple dotted line is the subspace $ \langle\bm{1}\rangle^{\bot} \subset \mathbb{R}^{2} $. 
The support of $ \bm{\chi}(\bm{Z}) $ is the intersection of them
---
a two-point set composed of the orange dot and the green dot:
$
\left\{
\left({\sqrt{2}}/{2}, -{\sqrt{2}}/{2}\right)^{\T}
,
\left(-{\sqrt{2}}/{2}, {\sqrt{2}}/{2}\right)^{\T}
\right\}
$.
Subfigure~\ref{subfig:X_marginal_dist._3} illustrates the $ n = 3 $ case: the light-colored sphere is $ \mathcal{S}^{2} $, and the black plain is the subspace $ \langle\bm{1}\rangle^{\bot} \subset \mathbb{R}^{3} $. 
The support of $ \bm{\chi}(\bm{Z}) $ is the intersection of them
---
the blue circle.

\begin{figure}[!h]
	\centering
	\subcaptionbox{$ n = 2 $.\label{subfig:X_marginal_dist._2}}
	{\includegraphics[width=0.48\textwidth]{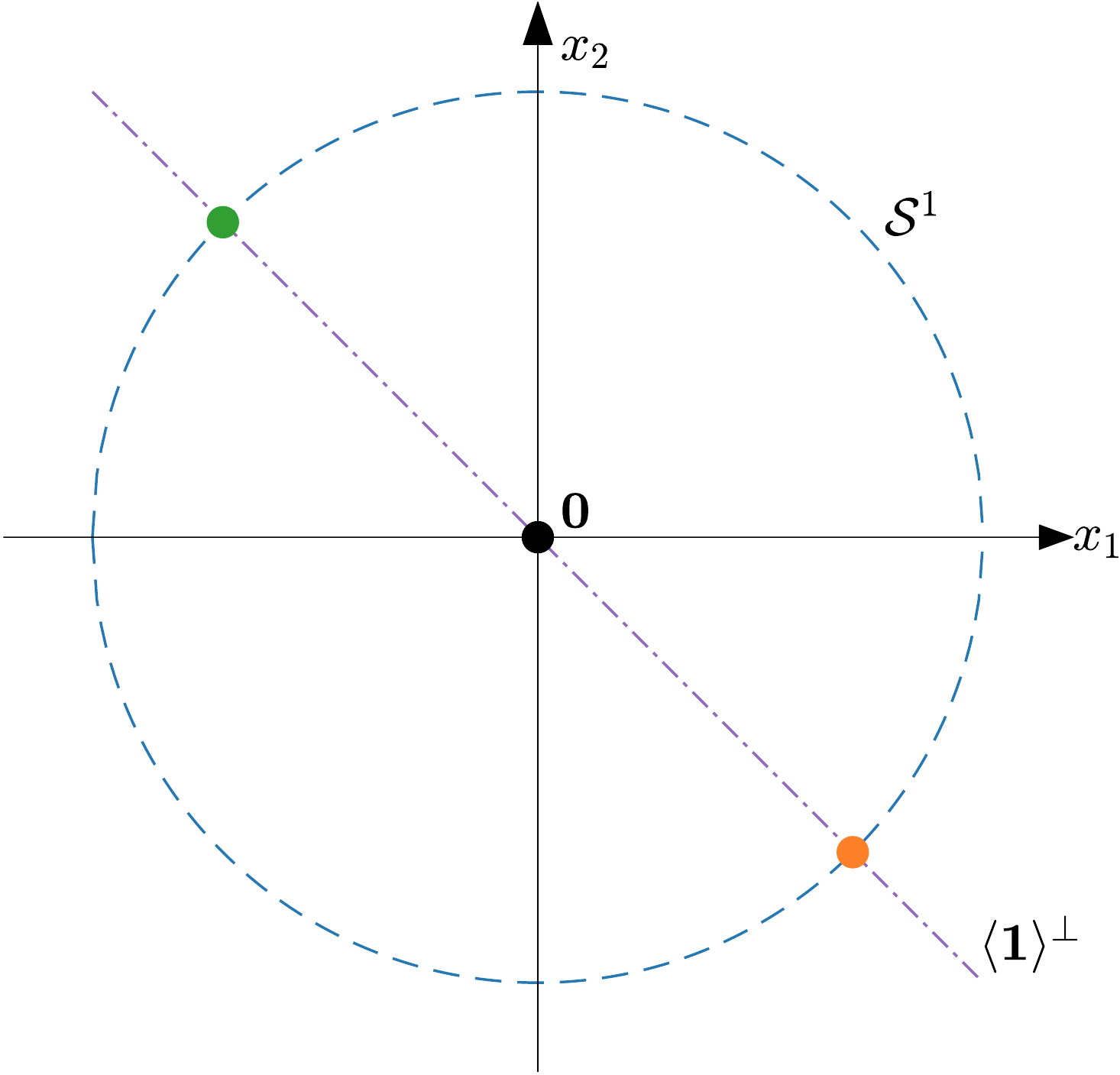}}
	\subcaptionbox{$ n = 3 $.\label{subfig:X_marginal_dist._3}}
	{\includegraphics[width=0.48\textwidth]{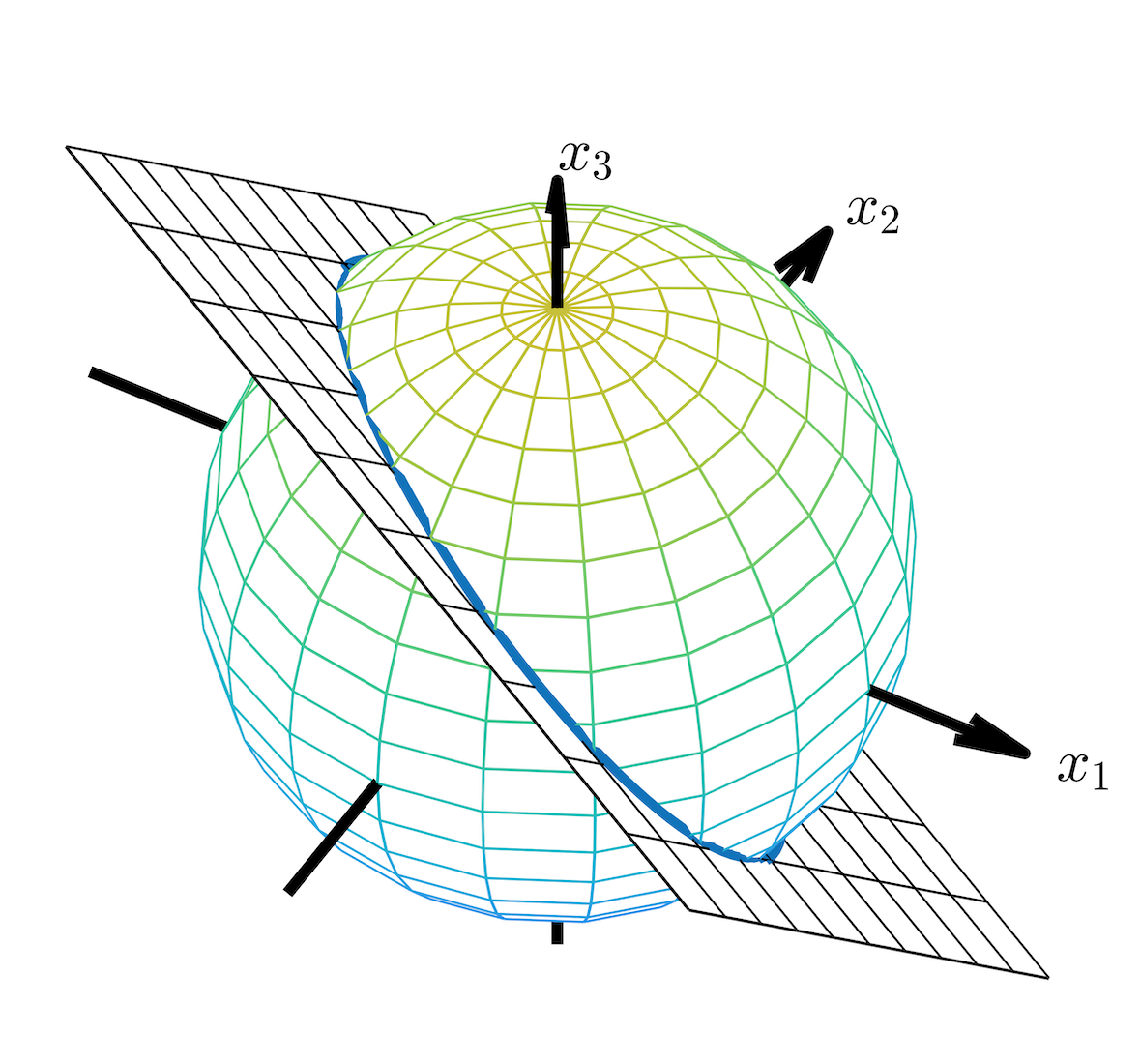}}
	\caption{The support of $ \bm{\chi}(\bm{Z}) $ of $ n = 2 $ and $ 3 $.}
	\label{fig:X_marginal_dist}
\end{figure}

Second, in high-dimensional cases, notably when $ n>10 $, the support area diminishes dramatically.
In detail, 
the surface area of the support
$ \mathcal{S}^{n-1} \cap \langle\bm{1}\rangle^{\bot} $ 
is 
$ f(n) := \frac{2\pi^{\frac{n-2}{2}}}{\Gamma\left(\frac{n-2}{2}\right)} $,%%
where $ f(10) \approx 32, f(50) \approx 6.6\times 10^{-11}, f(100) \approx 3.7 \times 10^{-37}, f(300) \approx 9.3\times 10^{-185} $.
In other words, the surface area of the support decreases exponentially, as the dimension $ n $ increases.
In finance, the dimension of  $ \bm{\chi}(\bm{Z}) $ is generally higher than $ 50 $. 

Third, the random vector $ \bm{\chi}(\bm{Z}) $ is a compound operator of $ \bm{Z} $.
In detail,
there are two projections implied in $ \bm{\chi}(\bm{Z}) $.
One is the \textit{linear} projection $ P $, which projects $ \bm{Z} $ from $ \mathbb{R}^{n} $ to the subspace $ \langle\bm{1}\rangle^{\bot} $.
$ P $ leads to the degeneration of the random vectors $ P
\bm{Z} $
---
the covariance matrix of $ P\bm{Z} $ must be singular.%%
The other is the \textit{directional} projection, which projects $ P\bm{Z} $ radically on the sphere $ \mathcal{S}^{n-1} $.
It is also inherently difficult because there is some ``singular nature'' on $ \mathcal{S}^{n-1} $ \citep[p. 166]{mardia2000directional}.

Fourth, 
there is no ready-made result about $ \bm{\chi}(\bm{Z}) $.
Although the transformation is pervasive in academia and industry, the results of the transformed variables are rare.
In the existing literature, models were usually set up on the sphere directly, particularly on cases like the paleomagnetism and wildfire orientation.

Fifth, $ \bm{\chi}(\bm{Z}) $'s basic properties in the $ 2 $-dim cases are given in Proposition~\ref{prop:X_2dim},
helping us have a fundamental and concrete understanding of $ \bm{\chi}(\bm{Z}) $.
\begin{proposition}\label{prop:X_2dim}
	Given a $ 2 $-dim random vector $ \bm{Z} \in \mathbb{R}^{2} $, 
	we have
	\begin{align}
	\bm{\chi}(\bm{Z})
	=&
	\mathrm{sign}
	\{Z_{1} - Z_{2}\}
	\cdot
	\left(\frac{\sqrt{2}}{2}, -\frac{\sqrt{2}}{2}\right)^{\T}
	,
	\label{equ:X_2dim}
	\end{align}
	where $ \mathrm{sign}\{x\} = \mathbbm{1}\{x>0\} - \mathbbm{1}\{x<0\}$. 
	The distribution of $ \bm{\chi}(\bm{Z}) $ is:
	\begin{center}
		\begin{tabularx}{0.5\textwidth}{c|Y|Y}
			\hline
			$ \bm{\chi}(\bm{Z}) $ & $ \left(
			{\sqrt{2}}/{2}, -{\sqrt{2}}/{2}
			\right)^{\T} $ & $ \left(
			-{\sqrt{2}}/{2}, {\sqrt{2}}/{2}
			\right)^{\T} $ \\
			\hline
			$ \mathbb{P} $ & $ \mathbb{P}\{Z_{1} > Z_{2}\} $ & $ 1-\mathbb{P}\{Z_{1} > Z_{2}\} $ \\
			\hline
		\end{tabularx}
	\end{center}
	The $\mathrm{MRL} $ and $\mathrm{MD} $ of $ \bm{\chi}(\bm{Z}) $ are
	\begin{align}
	\mathrm{MRL} \left(\big.\bm{\chi}(\bm{Z})\right) 
	=&\ 
	2
	\cdot
	\left |
	\big.
	\mathbb{P}\left\{
	Z_{1} > Z_{2}
	\right\}
	-
	0.5
	\right |
	,
	\label{equ:MRL_2d}
	\\
	\mathrm{MD} \left(\big.\bm{\chi}(\bm{Z})\right) 
	=&
	\begin{cases}
	\left(
	\frac{\sqrt{2}}{2}, -\frac{\sqrt{2}}{2}
	\right)^{\T},
	&
	\text{if } \mathbb{P}\{Z_{1} > Z_{2}\} > 0.5,
	\\
	\left(
	-\frac{\sqrt{2}}{2}, \frac{\sqrt{2}}{2}
	\right)^{\T},
	&
	\text{if } \mathbb{P}\{Z_{1} > Z_{2}\} < 0.5.
	\end{cases}
	\end{align}
\end{proposition}

Proposition~\ref{prop:X_2dim} is a general $ 2 $-dim proposition in the sense of the weak assumption of the distribution of $ \bm{Z} $,
showing an intricate property of $ \bm{\chi}(\bm{Z}) $: 
$ \bm{\chi}(\bm{Z}) $ is affected by only a part of the information in $ \bm{Z} $
---
In the $ 2 $-dim case, it is $ \mathrm{sign}\{Z_{1}-Z_{2}\} $ 
that determines $ \bm{\chi}(\bm{Z}) $ and its distribution and numerical characteristics,
irrelevant to the other information in $ \bm{Z} $.
Besides,
\eqref{equ:MRL_2d} shows that, when $ \mathbb{P}\{Z_{1} > Z_{2}\} = 0.5 $,
the $ \mathrm{MRL} = 0 $.
In directional statistics, it corresponds to the case where $ \bm{X} \in \mathcal{S}^{1} $ is symmetric about the origin.
One of the possible sufficient assumption is that $ \{Z_{i}\}_{i=1, 2} $ is i.i.d.. 
In the paper, we assume that $ \mathbb{P}\{Z_{1} > Z_{2}\} \neq 0.5 $.

Furthermore, we give $ \bm{\chi}(\bm{Z}) $'s properties of the $ 2 $-dim normal distributed $ \bm{Z} $ for a deeper understanding.
If $ \bm{Z} $ has a $ 2 $-dim normal distribution, say $ 
\bm{Z}
\sim
N
\left(
\begin{pmatrix}
\mu_{1} \\
\mu_{2}
\end{pmatrix}
,
\begin{pmatrix}
\sigma_{1}^{2} & \rho\sigma_{1}\sigma_{2}  \\
\rho\sigma_{1}\sigma_{2} & \sigma_{2}^{2} \\
\end{pmatrix}
\right)
,
$
then
\begin{align}
\mathrm{MRL}
\left(\big.\bm{\chi}(\bm{Z})\right)
&=
2
\left[
\Phi
\left(
\frac{\left|\mu_{1} - \mu_{2}\right|}{\sqrt{\sigma_{1}^{2} + \sigma_{2}^{2} - 2\rho\sigma_{1}\sigma_{2}}}
\right)
-
0.5
\right]
,
\label{equ:X_MRL_2dim_normal}
\\
\mathrm{MD}
\left(\big.\bm{\chi}(\bm{Z})\right)
&=
\mathrm{sign}
\left\{
\mu_{1} - \mu_{2}
\right\}
\cdot
\left(\frac{\sqrt{2}}{2}, -\frac{\sqrt{2}}{2}\right)^{\T}
,
\label{equ:X_MD_2dim_normal}
\end{align}
where $ \Phi(x) $ is the cumulative distribution function of the standard normal distribution.
For one thing,
\eqref{equ:X_MRL_2dim_normal} shows that $ \mathrm{MRL}\left(\big.\bm{\chi}(\bm{Z})\right) $ 
increases with the dispersion $ \left| \mu_{1} - \mu_{2} \right| $ ($=\sqrt{2} \Vert P \bm{\mu} \Vert $) 
and the correlation coefficient $ \rho $, 
but decreases with the volatility $ \sigma_{1}, \sigma_{2} $. 
Particularly, when there is no dispersion (i.e., $ \mu_{1} = \mu_{2} $),
we have $ \mathrm{MRL} = 0 $,
implying that $ \bm{\chi}(\bm{Z}) $'s $\mathrm{MD}$ is not well-defined.
For another, 
as implied in \eqref{equ:X_MD_2dim_normal}, $\mathrm{sign} \{\mu_{1} - \mu_{2}\} $ determines $ \mathrm{MD}\left(\big.\bm{\chi}(\bm{Z})\right) $,
because, in the 2-dim case,  $ \mathbb{P}\{Z_{1} > Z_{2}\} > 0.5 \Leftrightarrow \mathrm{sign} \{\mu_{1} - \mu_{2}\} = 1$.
The choice of the two possible vectors is based on whose expectation of the component $ Z_{i} $, say $ \mu_{i} $, $ i = 1, 2 $, is higher.

\subsubsection{The Linear Combination of the Components of 
	\texorpdfstring{$ \bm{\chi}(\bm{Z}) $}{$ chi(Z) $}: \texorpdfstring{$ T_{\bm{\chi}(\bm{Z})}(\bm{\theta}) $}{$ T_{chi(Z)}(theta) $}
}

Instead of the $ n $-dim random vector $ \bm{\chi}(\bm{Z}) $ itself, 
some practical problems emphasize the linear combination of $ \bm{\chi}(\bm{Z}) $'s components. 
This subsection first defines the linear combination, lists its properties in $ 2 $-dim cases,
and interprets them in finance.

To begin with,  
given a deterministic vector $ \bm{\theta} \in \mathbb{R}^{n} $ representing the forecasts/weights, 
we have $ \forall \bm{x}\in \mathcal{S}^{n-1} \cap \langle \bm{1} \rangle ^{\bot}, \forall a \in \mathbb{R}$, 
$ (\bm{\theta} + a\bm{1})^{\T}\bm{x} = \bm{\theta}^{\T} \bm{x} $.
Thus, we restrict the domain of $ \bm{\theta} $ to $ \mathcal{S}^{n-1} \cap \langle \bm{1} \rangle ^{\bot} $, and define function
\begin{equation}
\begin{aligned}
T:
\left(
\mathcal{S}^{n-1} \cap \langle \bm{1} \rangle ^{\bot}
\right)
^{2}
\rightarrow&
\mathbb{R},
\\
(\bm{\theta}, \bm{x})
\mapsto&
\bm{\theta}^{\T}
\bm{x}
,
\end{aligned}
\label{equ:T}
\end{equation}
and denote \eqref{equ:T} by
\begin{align*}
T_{\bm{x}}(\bm{\theta}) := \bm{\theta}^{\T}\bm{x},
\quad
\bm{\theta} \in \mathcal{S}^{n-1} \cap \langle \bm{1} \rangle ^{\bot}.
\end{align*}
For the random vector $ \bm{\chi}(\bm{Z}) \in \mathcal{S}^{n-1} \cap \langle \bm{1} \rangle ^{\bot} $ defined in \eqref{equ:chi_Z}, 
we define a random function as
\begin{align*}
T_{\bm{\chi}(\bm{Z})}(\bm{\theta})
=&\ 
\bm{\theta}^{\T}
\bm{\chi}(\bm{Z})
,
\quad
\bm{\theta} \in \mathcal{S}^{n-1} \cap \langle \bm{1} \rangle ^{\bot}
.
\end{align*}

Then,
we list the expression, distribution, and numerical characteristics of the $ 2 $-dim $ T_{\bm{\chi}(\bm{Z})}(\bm{\theta}) $ in the following four points.
First, it is symmetric with domain and function value. The domain of $ \bm{\theta} $ in $ T_{\bm{\chi}(\bm{Z})}(\bm{\theta}) $ is symmetric
$ \mathcal{S}^{1} \cap \langle\bm{1}\rangle^{\bot} = \left\{ 
\left({\sqrt{2}}/{2}, -{\sqrt{2}}/{2} \right)^{\T}, 
\left(-{\sqrt{2}}/{2}, {\sqrt{2}}/{2} \right)^{\T} 
\right\} $, 
as in Figure~\ref{subfig:X_marginal_dist._2},
so there are only two ways to combine $ X_{1} $ and $ X_{2} $.
The values are also symmetric:
$ T_{\bm{\chi}(\bm{Z})}(-\bm{\theta}) = -T_{\bm{\chi}(\bm{Z})}(\bm{\theta}) $.
Without essential difference,
we denote $ \widetilde{\bm{\theta}} := \left({\sqrt{2}}/{2}, -{\sqrt{2}}/{2} \right)^{\T} $ and only discuss $ T_{\bm{\chi}(\bm{Z}}(\widetilde{\bm{\theta}}) $ in a certain case.
Second, 
\begin{align*}
T_{\bm{\chi}(\bm{Z})}(\bm{\theta})
=
\begin{cases}
1, & \text{if } (\theta_{1} - \theta_{2})(Z_{1}-Z_{2}) > 0,\\
-1, & \text{if } (\theta_{1} - \theta_{2})(Z_{1}-Z_{2}) < 0.
\end{cases}
\end{align*}
Thus, $ T_{\bm{\chi}(\bm{Z})}(\bm{\theta}) $ is a random function of $ \bm{\theta} $, taking value in $ \{-1, 1\} $.
Third, the distribution of $ T_{\bm{\chi}(\bm{Z})}(\widetilde{\bm{\theta}}) $ is 
\begin{center}
	\begin{tabularx}{0.7\textwidth}{c|Y|Y}
		\hline
		$ T_{\bm{\chi}(\bm{Z})} (\widetilde{\bm{\theta}}) $ & $ 1 $ & $ -1 $ \\
		\hline
		$ \mathbb{P} $ & $ \mathbb{P}\{ (Z_{1} - Z_{2}) > 0\} $ & $ 1-\mathbb{P}\{ (Z_{1} - Z_{2}) > 0\} $ \\
		\hline
	\end{tabularx}
\end{center}
which shows that
the distribution of $ T_{\bm{\chi}(\bm{Z})}(\bm{\theta}) $ is a two-point distribution determined by the $ \mathrm{sign}\left\{ Z_{1} - Z_{2} \right\} $, also irrelevant to the other exact values of $ \bm{Z} $.
Fourth, the numerical characteristics of the random variable $ T_{\bm{\chi}(\bm{Z})}(\bm{\theta}) $ are
\begin{align*}
\mathbb{E}T_{\bm{\chi}(\bm{Z})}(\bm{\theta}) 
=&\ 
\begin{cases}
2 \left(\mathbb{P}\{Z_{1} >  Z_{2}\} - 0.5 \right ), & \text{if } \bm{\theta} = 	\left(\frac{\sqrt{2}}{2}, -\frac{\sqrt{2}}{2} \right)^{\T}, \\
2 \left(0.5 - \mathbb{P}\{Z_{1} >  Z_{2}\} \right ), & \text{if } \bm{\theta} = 	\left(-\frac{\sqrt{2}}{2}, \frac{\sqrt{2}}{2} \right)^{\T},
\end{cases}
\\
\mathrm{var}\left(T_{\bm{\chi}(\bm{Z})}(\bm{\theta})\right)
=&\ 
1 - 4 \left(\mathbb{P}\{Z_{1} >  Z_{2}\} - 0.5 \right )^{2}
.
\end{align*}
For the expectation,
$ \mathbb{E} T_{\bm{\chi}(\bm{Z})}(\bm{\theta}) $ is symmetric about $ \bm{\theta} $:
$ \mathbb{E}T_{\bm{\chi}(\bm{Z})}(-\bm{\theta}) = -\mathbb{E}T_{\bm{\chi}(\bm{Z})}(\bm{\theta}) $.
If $ \theta_{1} > \theta_{2} $ and $ \left(\mathbb{P}\{Z_{1} >  Z_{2}\} - 0.5 \right) > 0.5 $, then $ \mathbb{E}T_{\bm{\chi}(\bm{Z})}(\bm{\theta}) > 0 $, and vice versa.
Its implication is entirely intuitive: If the ``order'' of $ \bm{\theta} $ is equal to the ``order'' of $ \bm{Z} $, then the expectation is positive.
For the variance, 
$ \mathrm{var}\left ( T_{\bm{\chi}(\bm{Z})}(\bm{\theta}) \right ) $ is irrelevant to the choice of $ \bm{\theta} $
but only determined by 
$ \mathbb{P}\{Z_{1} >  Z_{2}\} $.

In finance, 
$ T_{\bm{\chi}(\bm{Z})}(\bm{\theta}) $ is interpreted as the Information Coefficient (IC),
a popular and useful performance measure of investment strategies.
IC is
\begin{align}
 \mathrm{IC} 
= 
\frac
{\sum_{i=1}^{n}(\widehat{r}_{i, t} - \overline{\widehat{r}_{t}})
	(r_{i, t+1} - \overline{r_{t+1}})}
{\sqrt{\sum_{i=1}^{n}(\widehat{r}_{i, t} - \overline{\widehat{r}_{t}})^{2}
		\sum_{i=1}^{n}(r_{i, t+1} - \overline{r_{t+1}})^{2}}
} 
=
\bm{\chi}(\widehat{\bm{r}}_{t})^{\T}
\bm{\chi}(\bm{r}_{t+1})
,
\end{align}
where $ r_{i, t+1} $ is the return of stock $ i $ at the next-period time $ (t+1) $, and $ \widehat{r}_{i, t} $ is the forecasts of the return at the current-period time $ t $.
IC is the same as $ T_{\bm{\chi}(\bm{Z})}({\bm{\theta}}) $,
where $ \bm{\theta} $ is the standardized forecasts $ 
\bm{\chi}\left (\widehat{\bm{r}}_{t}\right )
$ 
and $ \bm{Z} $ is the next-period returns $ \bm{r}_{t+1} $.
The randomness in IC is similar to that in $ T_{\bm{\chi}(\bm{Z})}({\bm{\theta}}) $ at time $ t $.
From this point of view, IC is the weighted sum of $\bm{\chi}\left (\bm{r}_{t+1}\right )$
rather than the correlation coefficient of two random variables.

The research on $ T_{\bm{\chi}(\bm{Z})}(\bm{\theta}) $ is challenging.
First, 
$ T_{\bm{\chi}(\bm{Z})}(\bm{\theta}) $'s value is in $ [-1,1] $,
free from the support of the pre-standardized random vector $ \bm{Z} $.
We have found little research on the linear combination of the components of a random vector on $ \mathcal{S}^{n-1} $,
let alone that of the particular random vector $ \bm{\chi}(\bm{Z}) $.
Second, the probabilistic property of $ T_{\bm{\chi}(\bm{Z})}(\bm{\theta}) $ for given  $ \bm{\theta} $ is nontrivial.
As is clarified above, the research on $ \bm{\chi}(\bm{Z}) $ is difficult, so investigating $ T_{\bm{\chi}(\bm{Z})}(\bm{\theta}) $ from $ \bm{\chi}(\bm{Z}) $ is also hard.
Meanwhile, there need to develop tools and techniques to investigate $ T_{\bm{\chi}(\bm{Z})}(\bm{\theta}) $.
 
\subsection{The Probabilistic Properties of \texorpdfstring{$\bm{\chi}(\bm{Z})$ and $ T_{\bm{\chi}(\bm{Z})}(\bm{\theta}) $}{$chi(Z)$ and $ T_{chi(Z)}(theta) $}}

$ \bm{\chi}(\bm{Z}) $ and $ T_{\bm{\chi}(\bm{Z})}(\bm{\theta}) $ are the main research topics in the paper.
This subsection discusses their specific probabilistic properties.
We first extend the $ 2 $-dim case studies on $ \bm{\chi}(\bm{Z}) $ in \S~\ref{subsec:chi_Z} to a more general condition,
exploring the distribution, $\mathrm{MD}$, and $\mathrm{MRL} $ of the random vector $ \bm{\chi}(\bm{Z}) $.
Based on $ \bm{\chi}(\bm{Z}) $, we introduce the distribution and numerical characteristics of $ T_{\bm{\chi}(\bm{Z})}(\bm{\theta}) $.

\subsubsection{The Properties of \texorpdfstring{$ \bm{\chi}(\bm{Z}) $}{$ chi(Z) $}}

The properties of $ \bm{\chi}(\bm{Z}) $ of a general $ n $ ($ n>2 $) dimension are given below.
Specifically, we transform the particular $ \bm{\chi}(\bm{Z}) $ to a general one in directional statistics.
Further, we give closed-form expressions of $ \bm{\chi}(\bm{Z}) $'s numerical characteristics in a relatively general condition.

When $ \bm{Z} $ is a general $ n $-dim random vector,
we can transfer $ \bm{\chi}(\bm{Z}) $ into a particular projected random vector.
Thus, we give a representation theorem,
which transforms $ \bm{\chi}(\bm{Z}) $ into an $ (n-1) $-dim component-uncorrelated random vector $ \bm{\xi}_{n-1} $ (i.e., $ \mathrm{cov}(\bm{\xi}_{n-1}) $ is diagonal).\footnote{
In this paper, $ \bm{x}_{n-1} $ is an $ (n-1) $-dim column vector.
If $ \bm{x} $ exists, then $ \bm{x}_{n-1} $ is the first $ (n-1) $ components of $ \bm{x} $,
i.e., $ \bm{x}_{n-1} = \begin{psmallmatrix}
I_{n-1} & \bm{0}_{n-1}
\end{psmallmatrix} \bm{x} $.
$ \Sigma_{n-1} $ is an $ (n-1)\times (n-1) $ matrix.
If $ \Sigma $ exists, then we use $ \Sigma_{n-1} $ to represent the $ (n-1) $-th order leading principal submatrix of $ \Sigma $,
i.e., 
$ \Sigma_{n-1} = \begin{psmallmatrix}
I_{n-1} & \bm{0}_{n-1}
\end{psmallmatrix} 
\Sigma 
\begin{psmallmatrix}
I_{n-1} & \bm{0}_{n-1}
\end{psmallmatrix}^{\T}
$.
}

\begin{theorem}
	\label{thm:chi_Z_UZ}
	Given an $ n $-dim random vector $ \bm{Z} \in \mathbb{R}^{n} $,
	whose covariance matrix $ \Sigma $ is positive definitive,
	then there exists an $ n $-dim orthogonal matrix $ U $ such that
	\begin{align}
	\bm{\chi}(\bm{Z})
	=&\ 
	U
	\frac{
		\bm{\xi}
	}{
		\left\|
		\bm{\xi}
		\right\|
	}
	=
	U
	\begin{pmatrix}
	{\bm{\xi}_{n-1}}/{\|\bm{\xi}_{n-1}\|}\\
	0
	\end{pmatrix}
	,
	\label{equ:chi_Z_Uxi}
	\end{align}
	and
	the covariance matrix 
	$ \mathrm{cov}\left (\bm{\xi}_{n-1}\right ) =: \Lambda_{n-1} 
	$ 
	is diagonal,
	where
	\begin{align}
	\bm{\xi} 
	:=&\
	\begin{pmatrix}
	I_{n-1} & \\
	& 0
	\end{pmatrix}
	U^{\T}\bm{Z} 
	,
	\label{equ:xi}
	\\
	\bm{\xi}_{n-1} 
	:=&\
	\begin{pmatrix}
	I_{n-1} & \bm{0}_{n-1}
	\end{pmatrix}
	\bm{\xi}
	\label{equ:xi_n-1}
	.
	\\
	U :=&\ VW.
	\label{equ:U}
	\end{align}
	$ V $ is an $ n $-dim orthogonal matrix defined as
	\begin{align}
	V_{ij}
	:=&\ 
	\begin{cases}
	\frac{n-j}{\sqrt{(n-j)(n-j+1)}}, & i = j \neq n; \\
	-\frac{1}{\sqrt{(n-j)(n-j+1)}}, & i > j \neq n; \\
	\frac{1}{\sqrt{n}}, & j = n; \\
	0, & \text{else,}
	\end{cases}
	\label{equ:V_short}
	\end{align}
	and
	$ 
	W 
	:= 
	\begin{psmallmatrix}
	W_{n-1} & \\
	& 1
	\end{psmallmatrix}
	$.
	$ W_{n-1} $ is an $ (n-1) $-dim orthogonal matrix such that
	$
	W_{n-1}^{\T}
	\left[
	\begin{psmallmatrix}
	I_{n-1} & \bm{0}_{n-1}
	\end{psmallmatrix}
	V^{\T}\Sigma V
	\begin{psmallmatrix}
	I_{n-1} \\ \bm{0}_{n-1}^{\T}
	\end{psmallmatrix}
	\right]
	W_{n-1}
	$ is a diagonal matrix.\footnote{
		In other words, $ W_{n-1} $ is the orthogonal matrix such that the covariance matrix of $ \bm{\xi}_{n-1} $ is a diagonal matrix $ \Lambda_{n-1} $.
		$ W $ is the $ n $-dim orthogonal matrix such that the covariance matrix of $ \bm{\xi} $ is a diagonal matrix $ \begin{psmallmatrix}
		\Lambda_{n-1} & \\
		& 0
		\end{psmallmatrix}
		=:
		\Lambda
		$.
	}
\end{theorem}

Theorem~\ref{thm:chi_Z_UZ} 
transforms
$ \bm{\chi}(\bm{Z}) $ 
into a standardized multivariate distribution, say the projected distribution.
One of the direct applications of Theorem~\ref{thm:chi_Z_UZ} is the concise expression of numerical characteristics of $ \bm{\chi}(\bm{Z}) $: Corollary~\ref{cor:numerical_characteristics_xi}.
\begin{corollary}\label{cor:numerical_characteristics_xi}
	The numerical characteristics of $ \bm{\chi}(\bm{Z}) $ are
	\begin{align}
	\mathrm{MD}\left(\big.\bm{\chi}(\bm{Z})\right)
	=&\ 
	U
	\left[
	\big.
	\mathrm{MD}
	\left(
	{
		\bm{\xi}
	}/{
		\left\|
		\bm{\xi}
		\right\|
	}
	\right)
	\right]
	=
	U
	\begin{pmatrix}
	\mathrm{MD}
	\left(
	{
		\bm{\xi}_{n-1}
	}/{
		\left\|
		\bm{\xi}_{n-1}
		\right\|
	}
	\right)
	\\
	0
	\end{pmatrix}
	,
	\label{equ:MD_chi_Z}
	\\
	\mathrm{MRL}\left(\big.\bm{\chi}(\bm{Z})\right)
	=&\ 
	\mathrm{MRL}
	\left(
	{
		\bm{\xi}
	}/{
		\left\|
		\bm{\xi}
		\right\|
	}
	\right)
	=
	\mathrm{MRL}
	\left(
	{
		\bm{\xi}_{n-1}
	}/{
		\left\|
		\bm{\xi}_{n-1}
		\right\|
	}
	\right)
	,
	\label{equ:MRL_chi_Z}
	\end{align}
	where
	$ U $,
	$ \bm{\xi} $,
	and $ \bm{\xi}_{n-1} $
	are in Theorem~\ref{thm:chi_Z_UZ}.
\end{corollary}

If $ \bm{Z} \sim N(\cdot, \Sigma) $ has a multivariate normal distribution, 
then $ \bm{\xi}_{n-1} \sim N(\cdot, \Lambda_{n-1}) $ is a component-independent random vector of multivariate normal distribution.
So Theorem~\ref{thm:chi_Z_UZ} transforms $ \bm{\chi}(\bm{Z}) $
into $ \bm{\xi}_{n-1}/\|\bm{\xi}_{n-1}\| $, 
which is a random vector of projected component-independent normal distribution.
$ \bm{\xi}_{n-1} $ has some basic and standard distribution 
due to its diagonal covariance matrix.
Thus, we could obtain $ \bm{\chi}(\bm{Z}) $'s numerical characteristics from a projected component-independent normal distributed $ \bm{\xi}_{n-1}/\|\bm{\xi}_{n-1}\| $.

Unfortunately, there are very few results of the closed-form expressions of the $\mathrm{MD} $ and $\mathrm{MRL} $ of a projected multivariate normal distribution, even in the independent case.
The most cutting-edge research that we have known about the closed-form expressions of the $\mathrm{MD} $ and $\mathrm{MRL} $ of a projected normal distribution is \citet{presnell2008MRL_PN}.
It gives the closed-form expressions in a particular case, where the covariance matrix of $ \bm{\xi}_{n-1} $ is a scalar matrix $ \bm{\xi}_{n-1} \sim N(\cdot, \lambda^{2} I_{n-1}) $, 
listed as Proposition~\ref{prop:presnell} in the paper.

In Theorem~\ref{thm:MD_MRL_ndim}, we obtain closed-form expressions of the $\mathrm{MD}$ and $\mathrm{MRL}$ with a relative general dependent case,
in which the components of the pre-standardized normal-distributed random vector $ \bm{Z} $ is homogeneous.

\begin{theorem}\label{thm:MD_MRL_ndim}
	Given $ \bm{Z} \sim N(\bm{\mu}, \sigma^{2} \Xi), \ \sigma>0 $, where $ \Xi $ is a correlation coefficient matrix with off-diagonal elements equaling to $ \rho \in (-\frac{1}{n}, 1) $ , i.e.,
	\begin{align}
	\Xi
	&:=
	\begin{pmatrix}
	1 & \rho & \rho & \cdots & \rho \\
	\rho & 1 & \rho & \cdots & \rho \\
	\rho & \rho & 1 & \cdots & \rho \\
	\vdots & \vdots & \vdots & \ddots & \vdots \\
	\rho & \rho & \rho & \cdots & 1 \\
	\end{pmatrix}
	,
	\label{equ:Xi}
	\end{align}
	we have
	\begin{align}
	\mathrm{MD}\left(\big.\bm{\chi}(\bm{Z})\right)
	&=
	\bm{\chi}(\bm{\mu})
	,
	\label{equ:MD_chi_Z_Xi}
	\\
	\mathrm{MRL}\left(\big.\bm{\chi}(\bm{Z})\right)
	&=
	\varrho_{n-1}\left(\frac{\Vert P\bm{\mu}\Vert}{\sigma \sqrt{1 - \rho}}\right)
	,
	\label{equ:MRL_chi_Z_Xi}	
	\end{align}
	where $ \varrho_{n-1} $ is defined in \eqref{equ:rho}.
\end{theorem}

We give some interpretations of Theorem~\ref{thm:MD_MRL_ndim} here.
First, $ \sigma^{2} \Xi $ is a covariance matrix widely used in statistical modeling, 
such as $n$ responds' covariance matrix with the single-factor model.
$ \sigma^{2} \Xi $ entails the special independent conditions: If $ \rho = 0 $, then $ \sigma^{2} \Xi = \sigma^{2} I_{n} $.
Second, under this specific situation, 
\eqref{equ:MD_chi_Z_Xi} implies an exchangeability of operators $ \mathbb{E} $ and $ \bm{\chi} $, that
\begin{align}
\mathbb{E}
\big.
\bm{\chi}
(\bm{Z})
\propto
\bm{\chi}
\left(
\mathbb{E}
\bm{Z}
\right)
=
\bm{\chi}
(\bm{\mu})
.
\end{align}
$ \mathrm{MD}\left(\big.\bm{\chi}(\bm{Z})\right) 
$ only depends on $ \bm{\mu} $, irrelevant to the correlation coefficient between $ \bm{Z} $'s components in the case of homogeneous distribution.
Third, $ \mathrm{MRL}\left(\big.\bm{\chi}(\bm{Z})\right) $ increases with the dispersion of $ \bm\mu $ (i.e., $ \|P\bm{\mu}\| $) and the correlation coefficient $ \rho $, while decreases with volatility $ \sigma $.
The properties are similar to that in the $ 2 $-dim case.

\begin{proposition}[Eq(6) in \cite{presnell2008MRL_PN}]\label{prop:presnell}
	Given an $ n $-dim random vector $ \bm{\xi} \sim N(\bm{\nu}, \lambda^{2}I), \bm{\nu}\neq\bm{0}, \lambda>0 $, for the projected normal distribution $ \frac{\bm{\xi}}{\Vert \bm{\xi} \Vert} $, we have
	\begin{align}
	\mathrm{MD}\left(\frac{\bm{\xi}}{\Vert \bm{\xi} \Vert}\right) =&\ \frac{\bm{\nu}}{{\Vert \bm{\nu} \Vert}},
	\\
	\mathrm{MRL}\left(\frac{\bm{\xi}}{\Vert \bm{\xi} \Vert}\right) =&\ \varrho_{n}\left(\frac{{\Vert \bm{\nu} \Vert}}{\lambda}\right)
	,
	\end{align}
	where 
	\begin{align}
	\varrho_{n}\left(x\right) :=
	\frac{\Gamma \left(\frac{n+1}{2}\right)}{\sqrt{2} \Gamma\left(\frac{n+2}{2}\right)} 
	x
	M
	\left(
	\frac{1}{2},
	\frac{n+2}{2},
	-\frac{1}{2}
	x^{2}
	\right)
	\label{equ:rho}
	\end{align} and $ M(\cdot, \cdot, \cdot) $ is the confluent hypergeometric function of the first kind.\footnote{$ M $ is a solution of a confluent hypergeometric equation, which can be written as
		$
		M(a, b, z) = \sum_{n=0}^{\infty} \frac{a^{(n)}z^{n}}{b^{(n)}n!},
		$
		where $ a^{(0)} = 1, a^{(n)} = a(a+1)(a+2)\cdots(a+n-1) $.
		When $ b > a >0 $, it  can be represented as an integral:
		\begin{align*}
		\frac{\Gamma(b-a) \Gamma(a)}{\Gamma(b)} M(a, b, z)
		&=
		\int_{0}^{1}e^{zt} t^{a-1} (1-t)^{(b-a-1)} \mathrm{d} t.
		\end{align*}
		For detailed properties, one can refer to \citet[p. 503, \S 13]{AS1964handbook}.
	}
\end{proposition}

\subsubsection{The Properties of \texorpdfstring{$ T_{\bm{\chi}(\bm{Z})}({\bm{\theta}}) $}{$ T_{chi(Z)}({{theta}}) $}}

$ T_{\bm{\chi}(\bm{Z})}({\bm{\theta}}) $ is a random function. 
Given $ \bm{\theta} \in \mathcal{S}^{n-1} \cap \langle \bm{1} \rangle ^{\bot} $, 
then $ T_{\bm{\chi}(\bm{Z})}({\bm{\theta}}) $ is a random variable. 
The expectation and variance of $ T_{\bm{\chi}(\bm{Z})}({\bm{\theta}}) $ are discussed.

First, the expectation is
\begin{align}
\mathbb{E}T_{\bm{\chi}(\bm{Z})}(\bm{\theta})
=&\ 
\mathrm{MRL} \left(\big.\bm{\chi}(\bm{Z})\right)
\cdot
\bm{\theta}^{\T}
\mathrm{MD} \left(\big.\bm{\chi}(\bm{Z})\right)
.
\label{equ:E_T}
\end{align}
$ \bm{\theta}^{\T} 	\mathrm{MD} \left(\big.\bm{\chi}(\bm{Z})\right) $ spans on $ [-1, 1] $
---
If $ \bm{\theta} = 	\mathrm{MD} \left(\big.\bm{\chi}(\bm{Z})\right)  $, 
then $ \bm{\theta}^{\T} \mathrm{MD} \left(\big.\bm{\chi}(\bm{Z})\right) $ reaches the maximum $ 1 $, 
and if $ \bm{\theta} = -\mathrm{MD} \left(\big.\bm{\chi}(\bm{Z})\right) $, then the minimum $ -1 $.
Therefore, 
$ \mathrm{MRL} \left(\big.\bm{\chi}(\bm{Z})\right)  $ is a strict bound of $ \mathbb{E} T_{\bm{\chi}(\bm{Z})}(\bm{\theta}) $, i.e.,
\begin{align}
\sup_{\bm{\theta} \in \mathcal{S}^{n-1}\cap \langle \bm{1} \rangle ^{\bot}}
\left| 
\mathbb{E} T_{\bm{\chi}(\bm{Z})}(\bm{\theta}) 
\right|
=&\ 
\mathrm{MRL} \left(\big.\bm{\chi}(\bm{Z})\right) 
.
\label{equ:E_T_bound}
\end{align}

When $ \bm{Z} \sim N(\bm{\mu}, \sigma^{2} \Xi) $, i.e., the particular condition in Theorem~\ref{thm:MD_MRL_ndim},
we have
\begin{align}
\mathbb{E}T_{\bm{\chi}(\bm{Z})}(\bm{\theta})
=&\ 
\varrho_{n-1}\left(\frac{\Vert P\bm{\mu}\Vert}{\sigma \sqrt{1 - \rho}}\right)
\cdot
\bm{\theta}^{\T}
\bm{\chi}(\bm{\mu})
.
\end{align}
So when $ \bm{\theta} = \bm{\chi}(\bm{\mu}) $, $ \mathbb{E}T_{\bm{\chi}(\bm{Z})}(\bm{\theta}) $ reaches the maximum $ \varrho_{n-1}\left(\frac{\Vert P\bm{\mu}\Vert}{\sigma \sqrt{1 - \rho}}\right) $,
and when $ \bm{\theta} = -\bm{\chi}(\bm{\mu}) $, $ \mathbb{E}T_{\bm{\chi}(\bm{Z})}(\bm{\theta}) $ reaches the minimum $ -\varrho_{n-1}\left(\frac{\Vert P\bm{\mu}\Vert}{\sigma \sqrt{1 - \rho}}\right) $.

Second,
the variance is
\begin{align}
\mathrm{var}\left(T_{\bm{\chi}(\bm{Z})}(\bm{\theta})\right)
=&\ 
\bm{\theta}^{\T}
\mathrm{cov}\left(\big.\bm{\chi}(\bm{Z})\right)
\bm{\theta}
.
\label{equ:T_var}
\end{align}
By \eqref{equ:chi_Z_Uxi} in Theorem~\ref{thm:chi_Z_UZ}, we have
$
\mathrm{cov}\left(\big.\bm{\chi}(\bm{Z})\right)
=
U 
\begin{pmatrix}
\mathrm{cov}\left(\frac{\bm{\xi}_{n-1}}{\left\|\bm{\xi}_{n-1}\right\|}\right) & \\
& 0
\end{pmatrix}
U^{\T}
$,
where $ \bm{\xi}_{n-1} $ defined in \eqref{equ:xi_n-1} is a positive definite diagonal covariance matrix $ \Lambda_{n-1} $.
Consequently,
\begin{align}
\mathrm{var}\left(T_{\bm{\chi}(\bm{Z})}(\bm{\theta})\right)
=&\ 
\left(U^{\T}\bm{\theta}\right)^{\T}
\begin{pmatrix}
\mathrm{cov}\left(\frac{\bm{\xi}_{n-1}}{\left\|\bm{\xi}_{n-1}\right\|}\right)
 & \\
& 0
\end{pmatrix}
\left(U^{\T}\bm{\theta}\right)
.
\label{equ:var_T_Utheta_cov}
\end{align}
If we figure out the covariance matrix of the particular projected component-uncorrelated distribution $ \frac{\bm{\xi}_{n-1}}{\|\bm{\xi}_{n-1}\|} $, then
$ T_{\bm{\chi}(\bm{Z})}(\bm{\theta}) $'s variance can be expressed.

However, solving the covariance matrix of the projected distribution $ \mathrm{cov}\left(\frac{\bm{\xi}_{n-1}}{\left\|\bm{\xi}_{n-1}\right\|}\right) $ is not an easy work.
The following result (see \cite[Eq(9.2.12), p. 166]{mardia2000directional}) is one of the relative works for $ \bm{X} \in \mathcal{S}^{n-1} $:
\begin{align}
\mathrm{trace}\left[ \mathrm{cov}\left(\bm{X}\right) \right]
= 1 - \left[\mathrm{MRL}\left(\bm{X}\right)\right]^{2}
.
\label{equ:cov_X_trace}
\end{align}

We obtain a closed-form expression of $ \mathrm{cov}\left(\frac{\bm{\xi}_{n-1}}{\left\|\bm{\xi}_{n-1}\right\|}\right) $, where $ \bm{\xi}_{n-1} $ has a homoscedastic normal distribution.

\begin{theorem}
	\label{thm:covariance_matrix}
	Given an $ n $-dim random vector $ \bm{\xi} \sim N(\bm{\nu}, \lambda^{2}I), \lambda>0, $  $ \bm{\nu} = (\|\bm{\nu}\|, 0, 0,\cdots, 0)^{\T} $,%%
	then
	\begin{align}
	\mathrm{cov}\left(\frac{\bm{\xi}}{\|\bm{\xi}\|}\right)
	=&\ 
	\begin{pmatrix}
	f_{n}\left(\frac{\|\bm{\nu}\|}{\lambda}\right)  &  \\
	& g_{n}\left(\frac{\|\bm{\nu}\|}{\lambda}\right)\cdot I_{n-1}   \\
	\end{pmatrix}
	,
	\label{equ:cov_xi}
	\end{align}
	where
	\begin{align}
	f_{n}(x)
	:=&\ 
	1
	-
	\frac{n-1}{n}
	M\left(1, \frac{n}{2} + 1,  -\frac{x^{2}}{2}\right)
	-
	\left[\varrho_{n}(x)\right]^{2},
	\label{equ:f}
	\\
	g_{n}(x)
	:=&\ 
	\frac{1}{n}
	M
	\left(
	1,
	\frac{n}{2} + 1,
	-\frac{x^{2}}{2}
	\right),
	\label{equ:g}
	\end{align}
	in which $ \varrho_{n} $ is defined in \eqref{equ:rho}.
\end{theorem}

The proof is enlightened by \cite{presnell2008MRL_PN}.
In the proof, the diagonal of the covariance matrix is easy to derive, but the closed-form expression of each element is a nontrivial result. 
In essence, we give the expression of $E(T^{2})$ refer to \cite{presnell2008MRL_PN}.

\begin{proof}[The skeleton of the proof of Theorem~\ref{thm:covariance_matrix}]

We first show $ \mathrm{cov}\left(\frac{\bm{\xi}}{\|\bm{\xi}\|}\right) $ is a diagonal matrix like \eqref{equ:cov_xi}, and then prove its closed-form expressions.

To begin with,
$ \mathrm{cov}\left(\frac{\bm{\xi}}{\|\bm{\xi}\|}\right) $ is a diagonal matrix of the form of \eqref{equ:cov_xi} --- $ \mathrm{cov}\left(\frac{\bm{\xi}}{\|\bm{\xi}\|}\right) = \begin{psmallmatrix}
\mathrm{var}\left(\frac{\xi_{1}}{\|\bm{\xi}\|}\right)  & \\
& \mathrm{var}\left(\frac{\xi_{2}}{\|\bm{\xi}\|}\right)  I_{n-1}
\end{psmallmatrix} $.
The diagonal is deduced by symmetry and independence.

Furthermore, 
we need to prove
\begin{align}
\mathrm{var}\left(\frac{\xi_{1}}{\|\bm{\xi}\|}\right)
=&\  
\mathbb{E} \frac{\xi_{1}^{2}}{\|\bm{\xi}\|^{2}} - \left(\mathbb{E} \frac{\xi_{1}}{\|\bm{\xi}\|}\right)^{2}
,
\label{equ:var_xi_1_normed}
\\
\mathrm{var}\left(\frac{\xi_{2}}{\|\bm{\xi}\|}\right)
=&\
\frac{1}{n-1}
\left(
1 - \mathbb{E} \frac{\xi_{1}^{2}}{\|\bm{\xi}\|^{2}}
\right)
,
\label{equ:var_xi_2_normed}
\end{align}
in which $ \mathbb{E} \frac{\xi_{1}}{\|\bm{\xi}\|} = \varrho_{n}\left(\frac{\|\bm{\nu}\|}{\lambda}\right) $.

Last, 
the closed-form expression of $ \mathbb{E} \frac{\xi_{1}^{2}}{\|\bm{\xi}\|^{2}} = 1
-
\frac{n-1}{n}
M\left(1, \frac{n}{2} + 1,  -\frac{\|\bm{\nu}\|^{2}}{2\lambda^{2}}\right)
$ is solved in three steps.
First, we give the simplified integral expression of $ \mathbb{E} \frac{\xi_{1}^{2}}{\|\bm{\xi}\|^{2}} $.
\begin{align*}
\mathbb{E} \frac{\xi_{1}^{2}}{\|\bm{\xi}\|^{2}}
=&\ 
\frac{{e}^{-\frac{\|\bm{\nu}\|^{2}}{2\lambda^{2}}}}{\sqrt{\pi} 2^{\frac{n}{2}} \Gamma\left(\frac{n-1}{2}\right)}
\int_{0}^{\infty}
r^{\frac{n-2}{2}}
e^{-\frac{r}{2}}
{\rm d}r
\int_{0}^{1} 
t^{\frac{1}{2}}
(1-t)^{\frac{n-3}{2}}
\cosh\left (\frac{\|\bm{\nu}\|}{\lambda}\sqrt{rt}\right )
{\rm d}t
.
\end{align*}
Second, the integral can be simplified by modified Bessel function $ I_{\alpha}(\cdot) $.
\begin{align*}
\mathbb{E} \frac{\xi_{1}^{2}}{\|\bm{\xi}\|^{2}}
=&\ 
\frac{e^{-\frac{\|\bm{\nu}\|^{2}}{2\lambda^{2}}}}{2\left (\frac{\|\bm{\nu}\|}{\lambda}\right )^{\frac{n}{2}}}
\int_{0}^{\infty}
e^{-\frac{r}{2}}
r^{\frac{n}{4}-1}
\left[
I_{\frac{n}{2}}\left (\frac{\|\bm{\nu}\|}{\lambda} \sqrt{r}\right )
+
\frac{\|\bm{\nu}\|}{\lambda} \sqrt{r}
I_{\frac{n+2}{2}}\left (\frac{\|\bm{\nu}\|}{\lambda} \sqrt{r}\right )
\right]
\mathrm{d}r
.
\end{align*}
Third, it can be concisely expressed by confluent hypergeometric function. 
\begin{align*}
\mathbb{E} \frac{\xi_{1}^{2}}{\|\bm{\xi}\|^{2}}
=&\ 
1
-
\frac{n-1}{n}
M\left(1, \frac{n}{2} + 1,  -\frac{\|\bm{\nu}\|^{2}}{2\lambda^{2}}\right)
.
\end{align*}

Consequently, the closed-form expressions of
\eqref{equ:var_xi_1_normed} and \eqref{equ:var_xi_2_normed} are solved.
We define them as $ f_{n} $ and $ g_{n} $, and thus we get \eqref{equ:f} and \eqref{equ:g}.
\end{proof}

When $ \bm{Z} \sim N(\bm{\mu}, \sigma^{2} \Xi) $, we deduce a closed-form expression of
$ \mathrm{var}\left(T_{\bm{\chi}(\bm{Z})}(\bm{\theta})\right)$.

\begin{theorem}\label{thm:cov_X_Z}
	Given $ \bm{Z} \sim N(\bm{\mu}, \sigma^{2} \Xi), \ \sigma>0 $, where $ \Xi $ is defined in \eqref{equ:Xi},
	then
	\begin{align*}
	\mathrm{var}\left(T_{\bm{\chi}(\bm{Z})}(\bm{\theta})\right)
	=&\
	\left[
	f_{n-1}\left(\frac{\|P\bm{\mu}\|}{\sigma\sqrt{1-\rho}}\right)
	-
	g_{n-1}\left(\frac{\|P\bm{\mu}\|}{\sigma\sqrt{1-\rho}}\right)
	\right]\cdot
	\left(
	\bm{\chi}(\bm{\mu})^{\T}\bm{\theta}
	\right)^{2}
	+
	g_{n-1}\left(\frac{\|P\bm{\mu}\|}{\sigma\sqrt{1-\rho}}\right)
	,
	\end{align*}
	where $ f_{n-1} $ and $ g_{n-1} $ is defined in \eqref{equ:f} and \eqref{equ:g}.
\end{theorem}

Theorem~\ref{thm:cov_X_Z} implies that, 
given $ \bm{Z} \sim N(\bm{\mu}, \sigma^{2} \Xi) $, 
the extremums of
$  \mathrm{var}\left(T_{\bm{\chi}(\bm{Z})}(\bm{\theta})\right) $
are
$ 	f_{n-1}\left(\frac{\|P\bm{\mu}\|}{\sigma\sqrt{1-\rho}}\right)
$ 
or 
$
g_{n-1}\left(\frac{\|P\bm{\mu}\|}{\sigma\sqrt{1-\rho}}\right) 
$, 
and they are reached when  $ \bm{\theta} = \pm \bm{\chi}(\bm{\mu}) $ or
$ \bm{\theta}^{\T}\bm{\chi}(\bm{\mu}) = 0 $.
We discuss the details of the optimization of $ T_{\bm{\chi}(\bm{Z})}(\bm{\theta}) $'s expectation and variance in the next subsection.

\subsection{The Optimization of \texorpdfstring{$ T_{\bm{\chi}(\bm{Z})}(\bm{\theta}) $}{$ T_{chi(Z)}({theta}) $}}

In some applications of $ T_{\bm{\chi}(\bm{Z})}(\bm{\theta}) $, 
we concern about the maximum of the expectation or the minimum of the variance of $ T_{\bm{\chi}(\bm{Z})}(\bm{\theta}) $ subjecting to $ \bm{\theta} $.
Mathematically, it is an optimization problem of $ \mathbb{E}T_{\bm{\chi}(\bm{Z})}(\bm{\theta}) $
or $ \mathrm{var}\left (T_{\bm{\chi}(\bm{Z})}(\bm{\theta})\right ) $.
This subsection specifies the optimization and connects it with above $ \bm{\chi}(\bm{Z}) $.

\subsubsection{The Maximization of the Expectation and the Minimization of the Variance}
\label{subsec:max_e_min_v}

First, we give a general result: $ \mathrm{MD}\left(\big.\bm{\chi}(\bm{Z})\right) $ maximizes the expectation.
Proposition~\ref{prop:ET} connects the optimization problem with the $\mathrm{MD}$ and $\mathrm{MRL}$ of $ \bm{\chi}(\bm{Z}) $.

\begin{proposition} \label{prop:ET}
	Given an $ n $-dim random vector $ \bm{Z} \in \mathbb{R}^{n} $ such that $ \mathbb{E}\bm{\chi}(\bm{Z}) \neq \bm{0} $,
	then
	the solution to the optimization problem
	\begin{align*}
	\max_{\bm{\theta} \in \mathcal{S}^{n-1} \cap \langle\bm{1}\rangle^{\bot}}
	\mathbb{E}T_{\bm{\chi}(\bm{Z})}(\bm{\theta})
	\end{align*}
	uniquely exists. The unique solution $ \bm{\theta}^{*} $ is the $\mathrm{MD} $ of $ \bm{\chi}(\bm{Z}) $, i.e.,
$
	\bm{\theta}^{*}
	=
	\mathrm{MD}\left(\big.\bm{\chi}(\bm{Z})\right)
	.
$
	The value of the maximum is the $ \mathrm{MRL} $ of $ \bm{\chi}(\bm{Z}) $, i.e.,
$
	\mathbb{E}T_{\bm{\chi}(\bm{Z})}(\bm{\theta}^{*})
	=
	\mathrm{MRL}\left(\big.\bm{\chi}(\bm{Z})\right)
	.
$
\end{proposition}

The proposition is a direct corollary of \eqref{equ:E_T}.
However, in finance, $ \bm{\theta}^{*} $ and $ \mathbb{E}T_{\bm{\chi}(\bm{Z})}(\bm{\theta}^{*}) $ are the core topics, 
whose properties are of paramount importance.
Specifically,
in active portfolio management,
a strategy of a high average of IC
is equivalent to
maximizing $ \mathbb{E}T_{\bm{\chi}(\bm{Z})}(\bm{\theta}) $.
If the joint distribution of the standardized returns of assets within a given pool were known, 
the maximum of IC and the optimal portfolio should be equivalent to 
the $\mathrm{MRL}$ and $\mathrm{MD}$ of the standardized returns of assets within the pool.

Second, we give the minimization of the variance of $ T_{\bm{\chi}(\bm{Z})}(\bm{\theta}) $ in the following proposition.

\begin{proposition}\label{prop:variance_TX}
	Given an $ n $-dim random vector $ \bm{Z} \in \mathbb{R}^{n} $ whose covariance matrix is positive definitive,
	then
	the minimum value of the optimization problem
	\begin{align}
	\min_{\bm{\theta} \in \mathcal{S}^{n-1} \cap \langle\bm{1}\rangle^{\bot}}
	\mathrm{var}\left (T_{\bm{\chi}(\bm{Z})}(\bm{\theta})\right ) 
	\label{equ:min_var_T}
	\end{align}
	is the second smallest eigenvalue of $ \mathrm{cov}\left( \bm{\chi}(\bm{Z}) \right) $.
	The solutions to the optimization problem are the corresponding unitized eigenvectors.
\end{proposition}

We interpret Proposition~\ref{prop:variance_TX}.
First, the uniqueness of the solution relies on the distribution of $ \bm{Z} $.
For instance, in the case $ \bm{Z} \sim N(\bm{1}, I) $,
every $ \bm{\theta} \in \mathcal{S}^{n-1}\cap \langle\bm{1}\rangle^{\bot} $ minimizes the variance.
Second, the eigenvectors corresponding to the smallest eigenvalue are excluded from the domain,
so the minimum value to the optimization problem of \eqref{equ:min_var_T} is the second smallest eigenvalue.
To be more specific, 
different from \eqref{equ:min_var_T},
the minimum value of the optimization problem
\begin{align}
\min_{\bm{\theta} \in \mathcal{S}^{n-1}}
\mathrm{var}\left (T_{\bm{\chi}(\bm{Z})}(\bm{\theta})\right ) 
\label{equ:min_var_T_general}
\end{align}
is the smallest eigenvalue of 
$ \mathrm{cov}\left(\big.\bm{\chi}(\bm{Z}) \right) $, 
i.e., $ 0 $.
The solution is the corresponding unitized eigenvectors $ \pm \frac{1}{\sqrt{n}}\bm{1} $.
However, 
$ \pm \frac{1}{\sqrt{n}}\bm{1} \notin \mathcal{S}^{n-1} \cap \langle\bm{1}\rangle^{\bot} $.
Therefore, for \eqref{equ:min_var_T}, 
we can only take the second smallest eigenvalue
and its corresponding unitized eigenvector.

\subsubsection{In the Homoscedastic World}

When $ \bm{Z} $ is homoscedastic, i.e., $ \bm{Z} \sim N(\bm{\mu}, \sigma^{2} \Xi), \sigma>0 $, 
$ \bm{\chi}(\bm{\mu}) $ is the mean-variance solution.
\begin{corollary}
	Given $ \bm{Z} \sim N(\bm{\mu}, \sigma^{2} \Xi), \sigma>0 $, and $ \Xi $ defined as \eqref{equ:Xi},
	$ \forall \lambda \in [0, \infty] $, the solution to 
	\begin{align*}
	 \max_{\bm{\theta} \in \mathcal{S}^{n-1} \cap \langle\bm{1}\rangle^{\bot}}
	\mathbb{E}T_{\bm{\chi}(\bm{Z})}(\bm{\theta}) - \lambda \cdot \mathrm{var}\left (T_{\bm{\chi}(\bm{Z})}(\bm{\theta})\right ) 
	\end{align*}
	is irrelevant to $ \lambda $:
	\begin{align*}
	\bm{\theta}^{*}
	&=
	\bm{\chi}(\bm{\mu})
	.
	\end{align*}
	The maximum value is
	\begin{align*}
	\varrho_{n-1}\left(\frac{\Vert P \bm{\mu}\Vert}{\sigma \sqrt{1 - \rho}}\right)
	-
	\lambda \cdot
	f_{n-1}\left(\frac{\|P\bm{\mu}\|}{\sigma\sqrt{1-\rho}}\right) 
	,
	\end{align*}
	where $ \varrho_{n-1} $ is defined in \eqref{equ:rho} and $ f_{n-1} $ is defined in \eqref{equ:f}.
\end{corollary}

It implies that in the homoscedastic world, the same solution applies to two different optimizations: the expectation's maximization and the variance's minimization.
In detail,
there are two particular instances of the corollary: $ \lambda = 0 $ and $ \lambda = \infty $.
When $ \lambda = 0 $, it is the application of \eqref{equ:E_T}.
The optimization problem is 
$ 
\max_{\bm{\theta} \in \mathcal{S}^{n-1} \cap \langle\bm{1}\rangle^{\bot}}
\mathbb{E}T_{\bm{\chi}(\bm{Z})}(\bm{\theta})
$,
and the maximum is $ \varrho_{n-1}\left(\frac{\Vert P \bm{\mu}\Vert}{\sigma \sqrt{1 - \rho}}\right) $,
which is reached at $ \bm{\theta}=\bm{\chi}(\bm{\mu}) $.
When $ \lambda = \infty $, it is the corollary of Theorem~\ref{thm:cov_X_Z}.
The optimization problem is equivalent to
$ 
\min_{\bm{\theta} \in \mathcal{S}^{n-1} \cap \langle\bm{1}\rangle^{\bot}}
\mathrm{var}\left (T_{\bm{\chi}(\bm{Z})}(\bm{\theta})\right ) 
$.
The minimum is
$ 
f_{n-1}\left(\frac{\|P\bm{\mu}\|}{\sigma\sqrt{1-\rho}}\right)  
$,
which can be reached at $ \bm{\theta}=\pm \bm{\chi}(\bm{\mu}) $.

\section{Simulation}\label{sec:simulation}

In this section, we aim to get a better understanding of $ \mathrm{MD} \left(\big.\bm{\chi}(\bm{Z})\right) $ and figure out the distribution of $ T_{\bm{\chi}(\bm{Z})}(\bm{\theta}) $ given $ \bm{\theta} = \bm{\chi}(\bm{\mu}) $ or $ \mathrm{MD}\left (\big. \bm{\chi}(\bm{Z}) \right ) $.
We begin the analysis with an intuitive $ 3 $-dim background and then consider a realistic financial market environment with $ 10 $-dim and $ 50 $-dim.

\subsection{The Numerical Analysis of $ 3 $-Dim \texorpdfstring{$ \mathrm{MD}\left(\big.\bm{\chi}(\bm{Z})\right) $}{$ \mathrm{MD}\left({chi}({Z})\right) $} }
\label{subsec:simulation_w*}

In this subsection, we examine how $ \mathrm{MD}\left(\big.\bm{\chi}(\bm{Z})\right) $ varies with respect to the distribution changes on $ \bm{Z} \sim N(\bm{\mu}_{3}, \Sigma_{3}) $,
where
\begin{align*}
\bm{\mu}_{3}
=&\
10^{-4}\times
\begin{pmatrix}
4.60 \\
7.83 \\			
14.78 
\end{pmatrix}
,
&
\Sigma_{3}
=&\
10^{-4}\times
\begin{pmatrix}
3.67  & 2.26  & 0.98  \\
& 6.60  & 0.96   \\
&       & 5.72 \\
\end{pmatrix}
.
\end{align*}
We simulate the samples of $ \bm{Z} $ for changed
$ \bm{\mu} $ and $ \Sigma $ separately. 

First, we focus on the impact of the heteroscedasticity of $ \Sigma_{3} $ on $ \mathrm{MD}\left(\big.\bm{\chi}(\bm{Z})\right) $.
Figure~\ref{fig:3-stock_simulation} shows the simulated results, where subfigure (\subref{subfig:3-stock_simulation_b}) is a top view of (\subref{subfig:3-stock_simulation_a}).
In detail, 
the red point is $ \bm{\chi}\left (\bm{\mu}_{3}\right ) $, 
which, in the above homoscedastic condition, equals to the $ \mathrm{MD} $. 
The numbers are all $ \mathrm{MD}\left(\big.\bm{\chi}(\bm{Z})\right) $'s of different parameters.
The number $ k $ is the simulated $ \mathrm{MD}\left(\big.\bm{\chi}(\bm{Z})\right) $ from $ \bm{Z} \sim N(\bm{\mu}_{3}, \Sigma_{3}^{(k)}) $ 
in which $ \Sigma_{3}^{(k)} := \mathrm{diag}(k,1,1) \Sigma_{3} \mathrm{diag}(k,1,1) $,
in other words,
the standard deviation of the first component $ \sigma_{1} $ is amplified by $ k $.

\begin{figure}[!h]
	\centering
	\subcaptionbox{\label{subfig:3-stock_simulation_a}}
	{\includegraphics[width=0.48\textwidth]{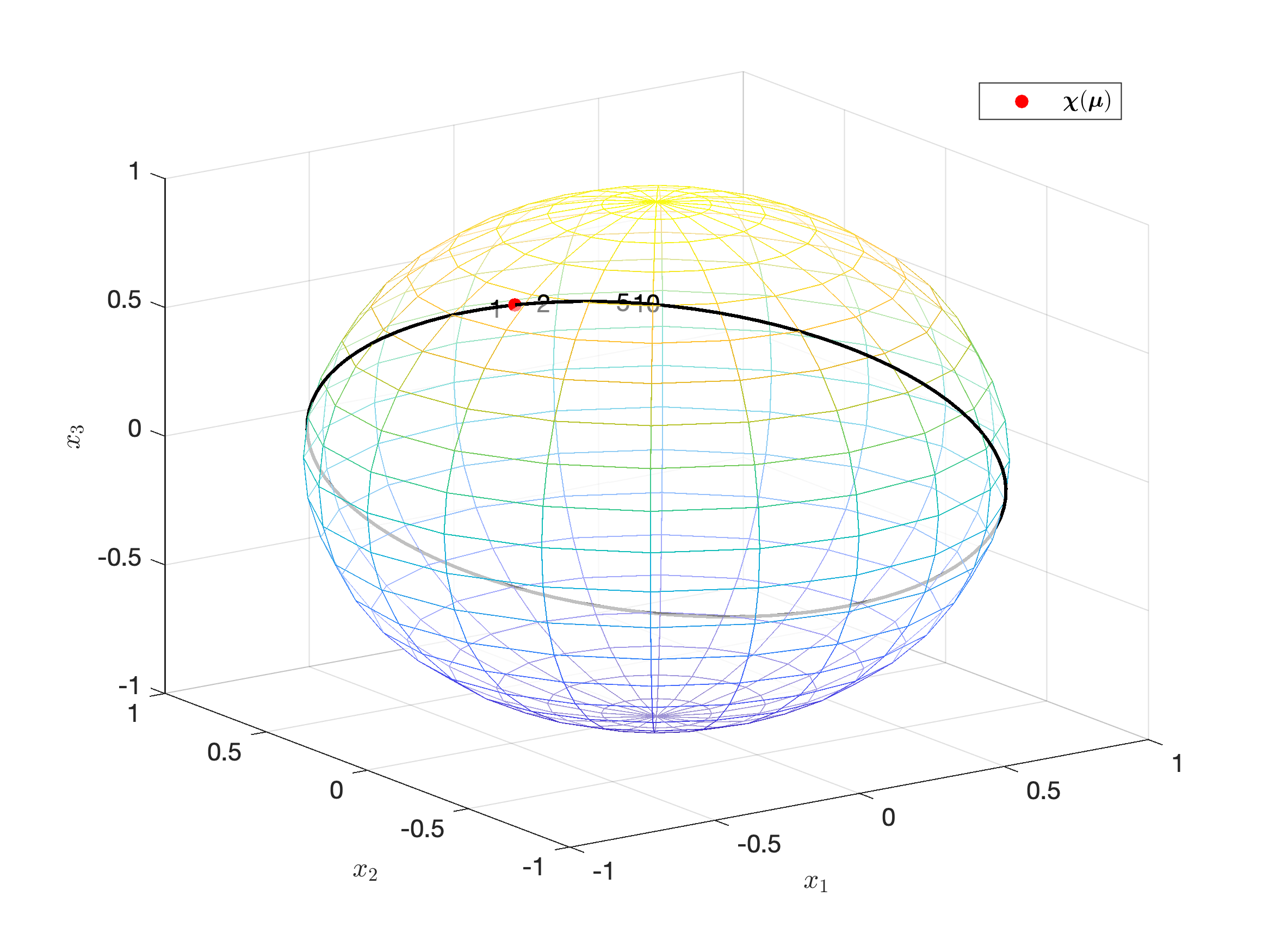}}
	\subcaptionbox{\label{subfig:3-stock_simulation_b}}
	{\includegraphics[width=0.48\textwidth]{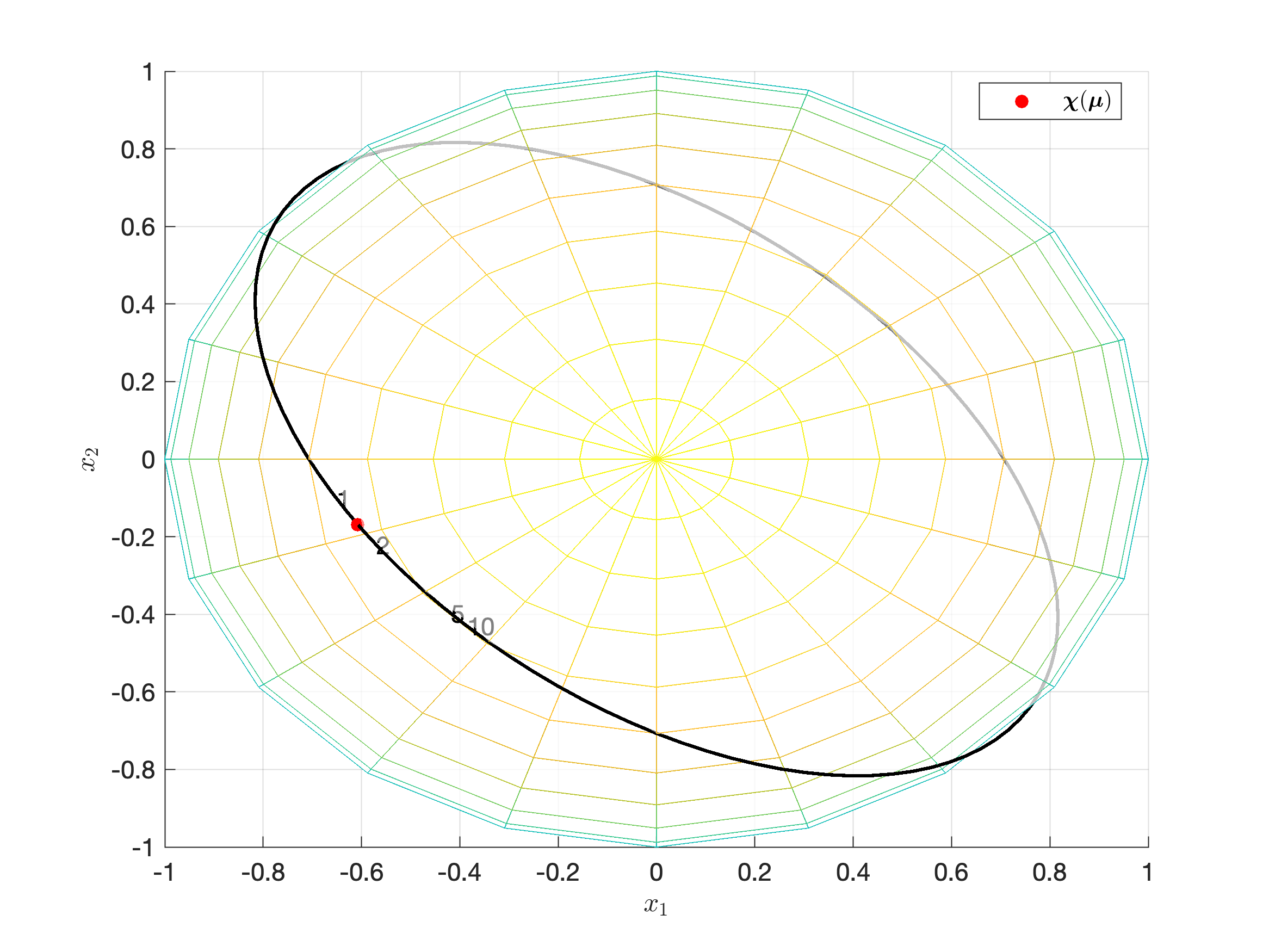}}
	\caption{3-stock simulation on the change of $ \sigma_{1} $ --- (\subref{subfig:3-stock_simulation_b}) is a top view of (\subref{subfig:3-stock_simulation_a}).}
	\label{fig:3-stock_simulation}
\end{figure}

Second, we illustrate the impact of $ \bm{\mu}_{3} $. Figure~\ref{fig:3-stock_simulation_mu} shows the simulated results.
In detail, the red numbers are the $ \bm{\chi}(\bm{\mu}) $'s and the black numbers are the $ \mathrm{MD}\left (\big. \bm{\chi}(\bm{Z}) \right ) $,
where number $ k $'s are calculated and simulated from $ \bm{Z} \sim N\left (\bm{\mu}_{3}^{(k)}, \Sigma_{3}\right ) $ in which 
$ \bm{\mu}_{3}^{(k)} := \left(
k\mu_{1}, \mu_{2}, \mu_{3}
\right)^{\T} $, $ k = 0.1, 1, 2, 5, 10 $.
In other words, we amply the first component $ \mu_{1} $ by $ k $.

\begin{figure}[!h]
	\centering
	\subcaptionbox{\label{subfig:3-stock_simulation_mu_a}}
	{\includegraphics[width=0.48\textwidth]{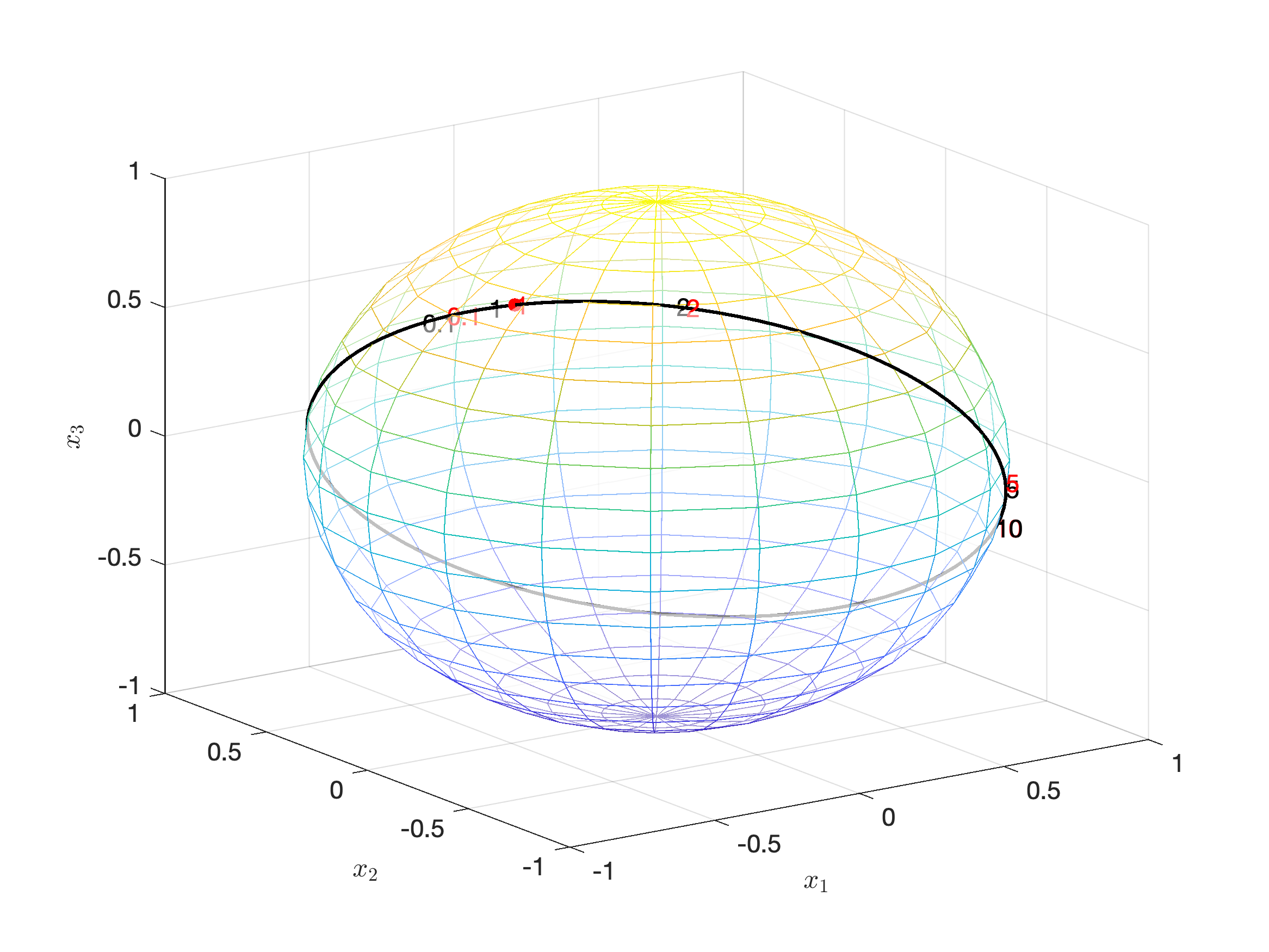}}
	\subcaptionbox{\label{subfig:3-stock_simulation_mu_b}}
	{\includegraphics[width=0.48\textwidth]{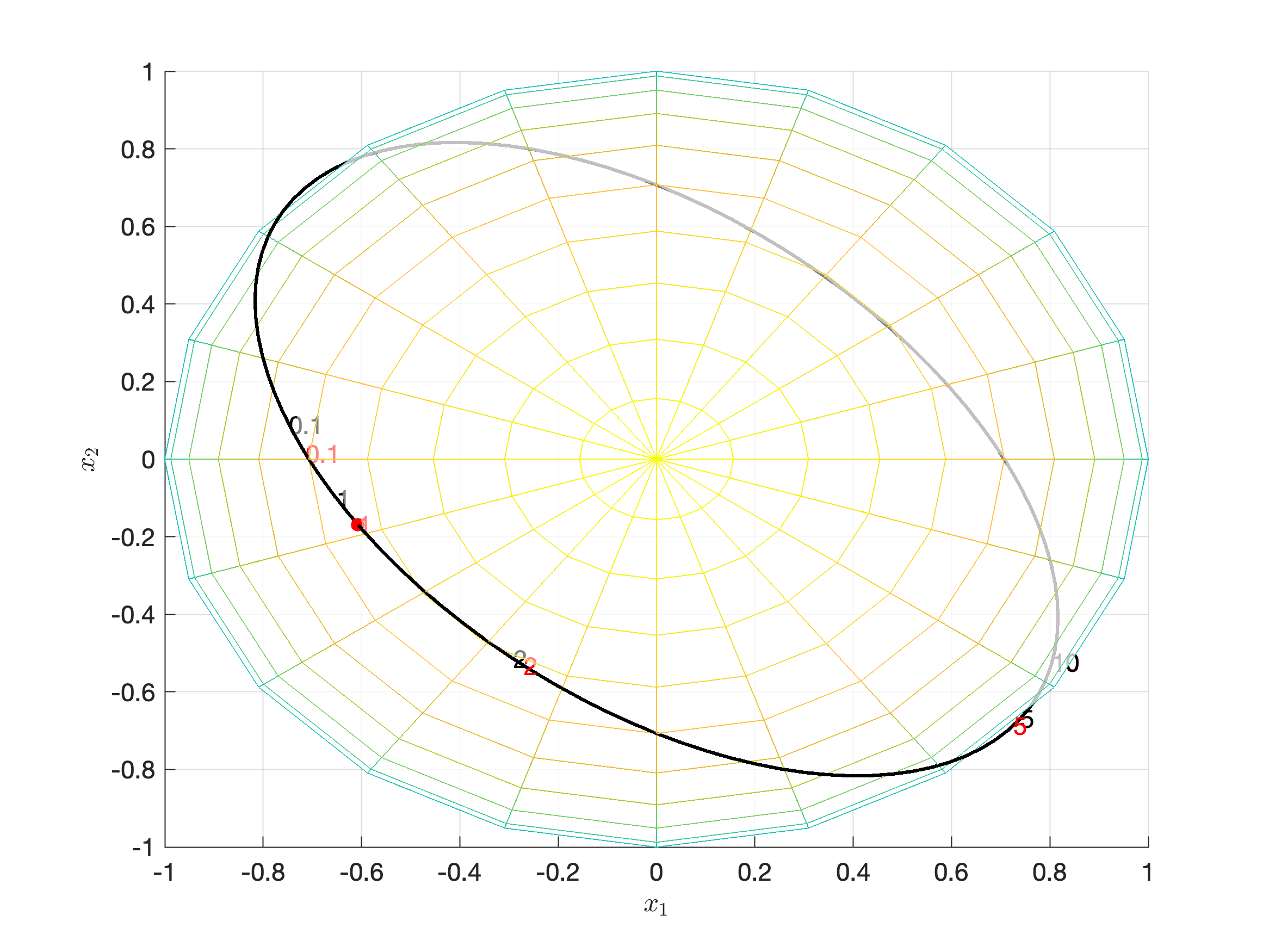}}
	\caption{3-stock simulation on the change of $ \mu_{1} $--- (\subref{subfig:3-stock_simulation_mu_b}) is a top view of (\subref{subfig:3-stock_simulation_mu_a}).}
	\label{fig:3-stock_simulation_mu}
\end{figure}

We interpret Figure~\ref{fig:3-stock_simulation} and \ref{fig:3-stock_simulation_mu}.
First, $ \bm{\chi}(\bm{\mu}_{3}) $ is different from $ \mathrm{MD}\left(\big.\bm{\chi}(\bm{Z})\right) $, as the red point is different from the number $ 1 $.
It accords that the structure of $ \Sigma $ in Theorem~\ref{thm:MD_MRL_ndim} is a sufficient condition that renders the $ \mathrm{MD} $ equal to $ \bm{\chi}(\bm{\mu}) $.
Second, the more heteroscedastic the $ \Sigma_{3} $, the more the difference between $ \bm{\chi}(\bm{\mu}) $ and $ \mathrm{MD}\left(\big.\bm{\chi}(\bm{Z})\right) $ that is represented by number from $ 2 $, $ \cdots $, to $ 10 $.
Third, 
the distance between the red and black $ 1 $ is broader than $ 2, 5, 10 $.
Thus the higher the $ \mu_{1} $, the smaller the difference between the $ \mathrm{MD}\left (\big. \bm{\chi}(\bm{Z}) \right ) $ and $ \bm{\chi}(\bm{\mu}) $.

\subsection{The Distribution of High-Dimensional \texorpdfstring{$ T_{\bm{\chi}(\bm{Z})}(\bm{\theta})| _{\bm{\theta} = \bm{\chi}(\bm{\mu}), \mathrm{MD}\left ( \bm{\chi}(\bm{Z}) \right )} $}{$  T_{chi(Z)}({theta}) | _{{theta} = \chi({mu}), {MD}({chi}({Z})) } $}}
\label{subsec:simulation_T(w)}

To get more realistic results, 
we consider high-dimensional cases
with parameters estimated from the real market
---
a part of the component stocks of CSI 300 Index of Chinese stock market. 
Specifically, we calculate the sample mean and covariance matrix of the simple daily returns from January 2017 to November 2018 as the parameters $ \bm{\mu}, \Sigma $.
In detail, 
we choose the first $ 50 $ stocks to generate $ \bm{\mu}_{50} $ and $ \Sigma_{50} $, and the first $ 10 $ stocks to generate $ \bm{\mu}_{10} $ and $ \Sigma_{10} $.
All the simulation parameters are in \ref{apx:sim_para}.

We simulate the sample distribution of $ T_{\bm{\chi}(\bm{Z})}(\bm{\theta}) $, when $ \bm{\theta} = \bm{\chi}(\bm{\mu}) $ and $ \mathrm{MD}\left ( \bm{\chi}(\bm{Z}) \right ) $, on $ \bm{Z} \sim N\left (\bm{\mu}_{10}, \Sigma_{10} \right ) $, $ N\left (\bm{\mu}_{50}, \Sigma_{50} \right ) $, 
and $ N\left (\bm{\mu}_{10}, \Sigma^{\prime}_{10} \right ) $. 
By amplifying the maximum of the standard deviance by $ 3 $, we derive the heteroscedastic $ \Sigma^{\prime}_{10} $ from $ \Sigma_{10} $.

The simulated kernel smoothing empirical p.d.f.'s of $ T_{\bm{\chi}(\bm{Z})}(\bm{\theta})| _{\bm{\theta} = \bm{\chi}(\bm{\mu}), \mathrm{MD}\left ( \bm{\chi}(\bm{Z}) \right )} $ are in Figure~\ref{fig:icfigstocks}
where the red vertical line is $ \mathbb{E}  T_{\bm{\chi}(\bm{Z})}\left (\Big.\mathrm{MD}\left (\big. \bm{\chi}(\bm{Z}) \right )\right ) $, i.e., the $ \mathrm{MRL} $ by Proposition~\ref{prop:ET}.

\begin{figure}[!h]
	\centering
	\begin{subfigure}{0.32\textwidth}
		\centering
		\includegraphics[width=\textwidth]{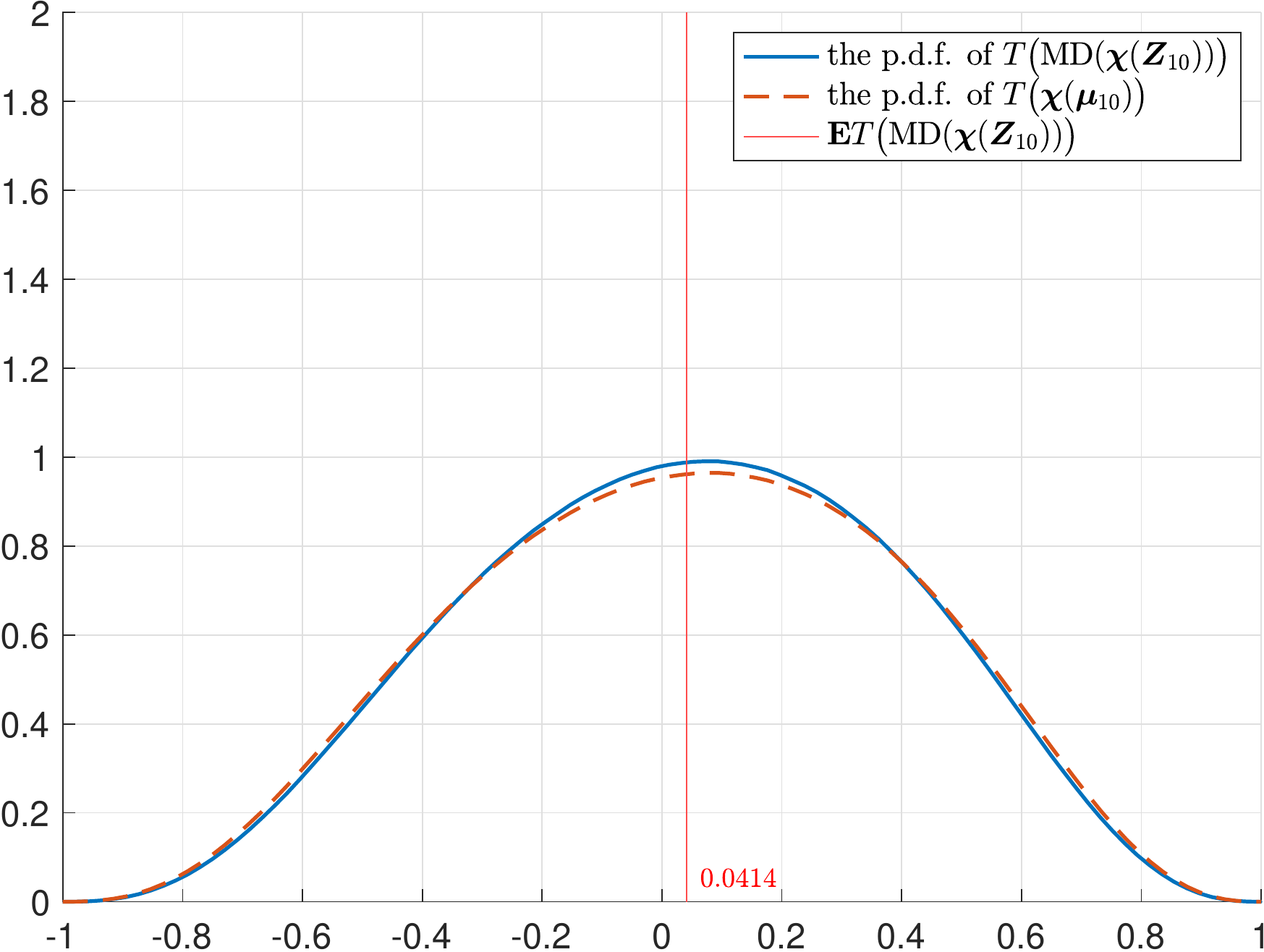}
		\caption[10 stocks.]{10 stocks.}
		\label{subfig:icfigstocks_10}
	\end{subfigure}
	\begin{subfigure}{0.32\textwidth}
		\centering
		\includegraphics[width=\textwidth]{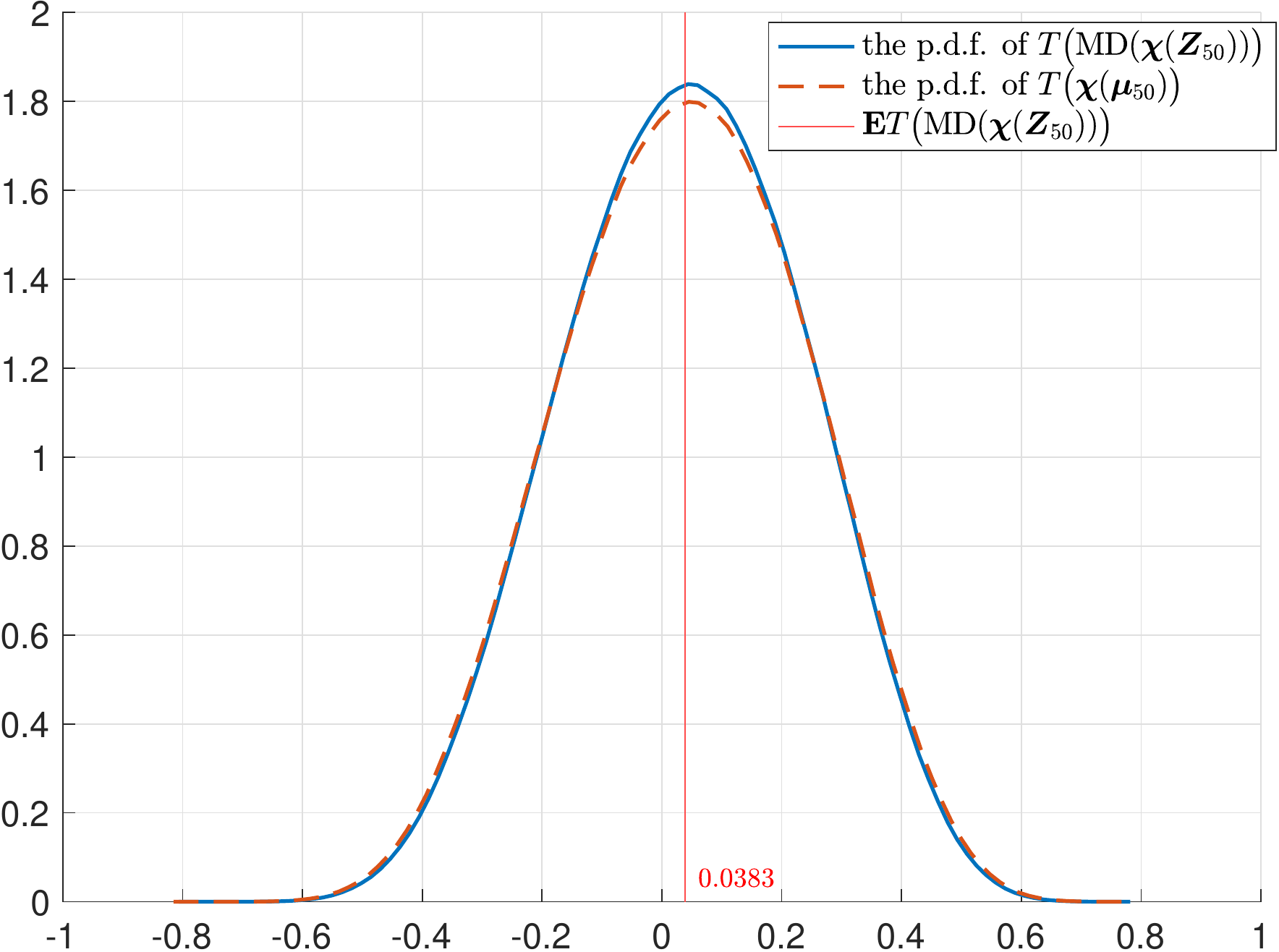}
		\caption[50 stocks.]{50 stocks.}
		\label{subfig:icfigstocks_50}
	\end{subfigure}
	\begin{subfigure}{0.32\textwidth}
		\centering
		\includegraphics[width=\textwidth]{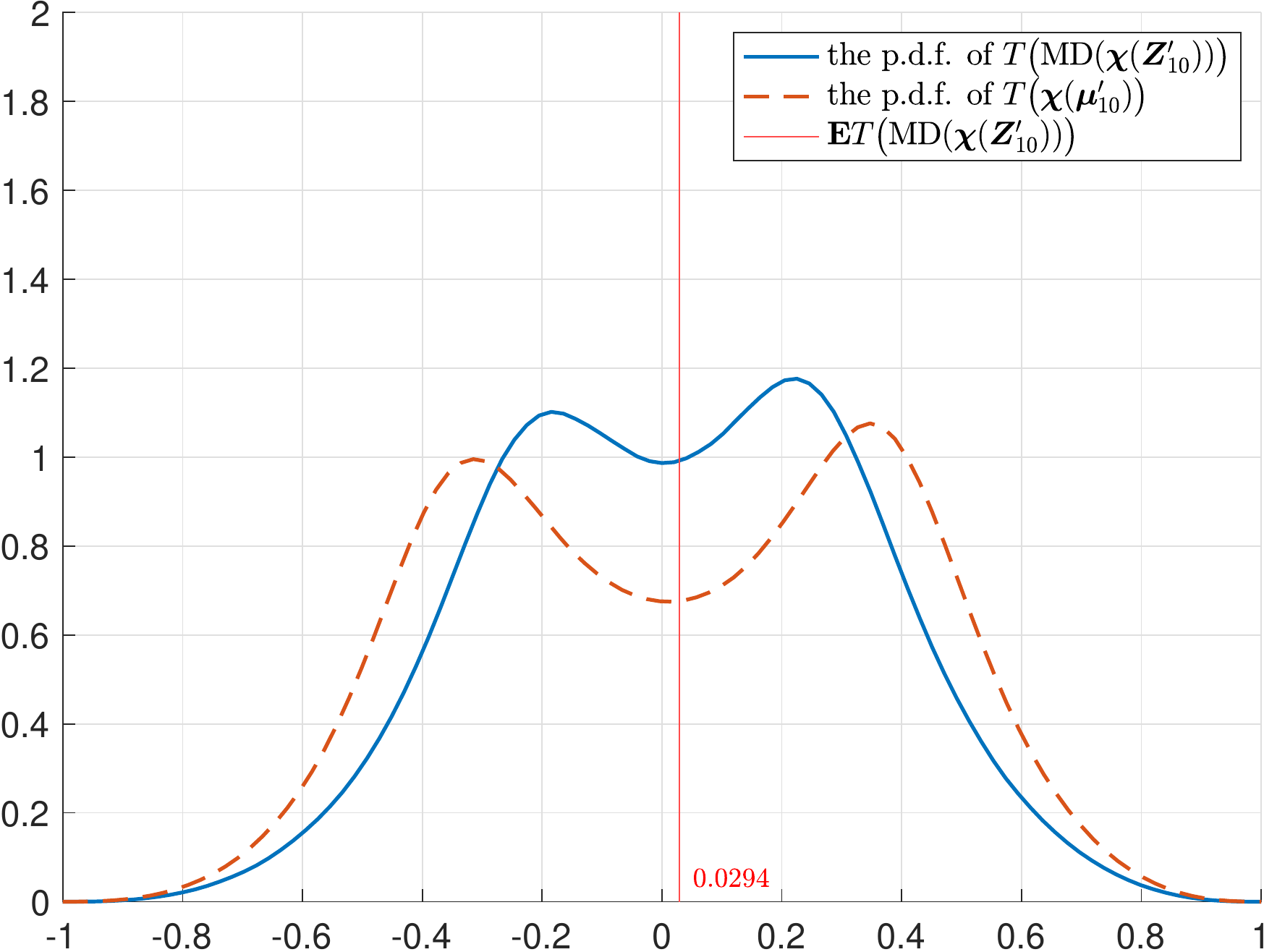}
		\caption{10 stocks with $ \Sigma'_{10} $.}
		\label{subfig:icfigstocks_10_ASA}
	\end{subfigure}
	\caption[The p.d.f.'s of $ T_{\bm{\chi}(\bm{Z})}({\theta}) $ w.r.t. different parameters.]{
		The p.d.f.'s of $ T_{\bm{\chi}(\bm{Z})}(\bm{\theta}) $ w.r.t. different parameters. 
	}
	\label{fig:icfigstocks}
\end{figure}

Some points in Figure~\ref{fig:icfigstocks} are discussed below in detail.
First, there is little difference between 
the $ \mathrm{MRL} $
in subfigure (\subref{subfig:icfigstocks_10}) and (\subref{subfig:icfigstocks_50}).
It seems that  
the $ \mathrm{MRL} $
reveals some basic facts of the financial market --- the value of $ \mathrm{IC} $ in expectation with a known future in the real market is usually around $ 0.03-0.04 $.
It follows \citet[p. 10-11]{grinold1994alpha}: ``A reasonable IC for an outstanding (top $ 5\% $) manager forecasting the returns on $ 500 $ stocks is about $ 0.06 $. If the manager is good (top quartile), $ 0.04 $ is a reasonable number.''
Second, the difference of volatility of the p.d.f.'s between the two subfigures (\subref{subfig:icfigstocks_10}) and (\subref{subfig:icfigstocks_50}) is significant, where $ T_{\bm{\chi}(\bm{Z})}(\bm{\theta}_{10})| _{\bm{\theta}_{10} = \bm{\chi}(\bm{\mu}_{10}), \mathrm{MD}(\bm{\chi}(\bm{Z}_{10}))} $ is more volatile than $ T_{\bm{\chi}(\bm{Z})}(\bm{\theta}_{50}) | _{\bm{\theta}_{50} = \bm{\chi}(\bm{\mu}_{50}), \mathrm{MD}(\bm{\chi}(\bm{Z}_{50}))}$.
It shows that the dimension $ n $ mainly impacts the variance,
rather than the expectation.
In portfolio management, the increase in the number of stocks does not impact $ \mathrm{IC} $'s mean, but significantly decreases $ \mathrm{IC} $'s variance.
Third, the difference between subfigures (\subref{subfig:icfigstocks_10}) and (\subref{subfig:icfigstocks_10_ASA}) is significant, 
suggesting that heteroscedasticity leads to the distinction of the distribution between $ T_{\bm{\chi}(\bm{Z})}(\bm{\theta}) | _{\bm{\theta} = \bm{\chi}(\bm{\mu}), \mathrm{MD}(\bm{\chi}(\bm{Z}))} $.
The heteroscedasticity renders the p.d.f. of $ \mathrm{IC} $ no longer a bell-shaped unimodal one.
Fourth, the difference between the p.d.f.'s of $ T_{\bm{\chi}(\bm{Z})}\left (\Big.\mathrm{MD}\left (\big. \bm{\chi}(\bm{Z}) \right )\right ) $ and $ T_{\bm{\chi}(\bm{Z})}(\bm{\chi}(\bm{\mu})) $ in subfigure (\subref{subfig:icfigstocks_10_ASA}) is more significant than that in other subfigures. 
It enlightens us that, in some cases, replacing  $ \mathrm{MD}\left (\big. \bm{\chi}(\bm{Z}) \right ) $ with $ \bm{\chi}(\bm{\mu}) $ may lead to fallacy.

\subsection{Comparing the Simulation and the Theoretical Results of High-Dimensional \texorpdfstring{$ \mathrm{MRL}\left (\big. \bm{\chi}(\bm{Z}) \right ) $}{$ MRL(chi(Z)) $}}
\label{subsec:simulation_ETw*}

In this subsection, 
we propose an approach to estimate $ \mathrm{MRL}\left (\big. \bm{\chi}(\bm{Z}) \right ) $ using estimated parameters of $ \bm{Z} $
and compare the approximate result with the true value by simulation.

According to Theorem~\ref{thm:MD_MRL_ndim}, $ \mathrm{MRL}\left (\big. \bm{\chi}(\bm{Z}_{10}) \right ) $ can be calculated approximately by regarding the $ \Sigma_{10} $ as the $ \sigma^{2}\Xi $.
In the $ n=10 $ simulation, we have $ \Vert P\bm{\mu}_{10} \Vert = 0.0027 $, $ \widehat{\sigma} := \sqrt{\frac{1}{10}\sum_{i=1}^{10}\sigma_{i}^{2}} = 0.0224 $, and $ \widehat{\rho} := \frac{1}{10\times 9 / 2} \sum_{1\leq i < j \leq 10} \rho_{ij} = 0.1243 $.
Therefore, $ 
\mathrm{MRL}\left (\big. \bm{\chi}(\bm{Z}_{10}) \right )
= \varrho_{9}\left(\frac{\Vert P\bm{\mu}_{10}\Vert}{\widehat{\sigma} \sqrt{1 - \widehat{\rho}}}\right) 
= \varrho_{9}(0.1288) = 0.0417 $,
while the simulated true value is $ 0.0414 $.
Furthermore, by a similar approach, we can calculate
the $ 50 $-dim cases where 
$
\mathrm{MRL}\left (\big. \bm{\chi}(\bm{Z}_{50}) \right )
\approx
0.0392
$,
while the simulated true value is $ 0.0383 $.

Thus, in some stationary scenarios, we can use Theorem~\ref{thm:MD_MRL_ndim} to approximate $ \mathrm{MRL}\left (\big. \bm{\chi}(\bm{Z}) \right ) $ quickly.

\section{Empirical Study}\label{sec:empirical}

In this section, we aim at exploring the empirical facts about the high-dimensional directional statistics of $ \bm{X} = \bm{\chi}(\bm{Z}) $, 
analyzing the connection between $ \bm{Z} $ and $ \bm{X} $ statically,
and clarifying their time-series properties.

We collect the historical daily return of the component stocks of the CSI $ 300 $ Index at the end of $ 2018 $. In detail, the daily data ranges from $ 2014 $ to $ 2018 $, including $ 1220 $ trading days in $ 5 $ years. However, not all stocks were traded throughout the time interval, so we remove those stocks not traded more than $ 10\% $ of all days, leaving $ n=185 $ stocks.
As a result, a matrix of $ 1220 \times 185 $ representing the returns is generated, called $ Z $.
The $ t $-th row of $ Z $ is a $ 185 $-dim vector representing the return of $ n = 185 $ stock returns at time $ t $, and we denote the vector as $ \bm{z}_{t} $.
With the pre-standardized $ \bm{z}_{t} $ at hand, we can calculate $ \bm{x}_{t} := \bm{\chi}(\bm{z}_{t}) $ easily.

The section is composed of two parts.
First, we carry out the exploratory data analysis
by assuming the sample points to be i.i.d..
Second, we show the time-series properties.

\subsection{Exploratory Data Analysis}

In this subsection, we explore the data $ \left \{ \bm{x}_{t} \right \}_{1\leq t\leq 1220} $, to have a basic understanding of $ \bm{X} = \bm{\chi}(\bm{Z}) $.
We assume the $ 185 $-dim time-series sample $ \left \{ \bm{x}_{t} \right \}_{1\leq t\leq 1220} $ to be i.i.d..

The analysis is composed of two parts.
One of them is the high-dimensional exploration of $ \bm{x}_{t} $.
We calculate the $ \mathrm{MD} $ and $ \mathrm{MRL} $ of the $ 185 $-dim sample $ \bm{x}_{t} $ 
and generate a $ 1 $-dim sample of 
$ \widehat{\bm{\iota}}^{\T}\bm{x}_{t} $ given an estimated $ \widehat{\bm{\iota}} $,
discussing their characteristics and interpretations.
The other one is the analysis of the connection between $ \bm{z}_{t} $ and $ \bm{x}_{t} $. 
We list the descriptive statistics of the selected $ z_{it} $ and $ x_{it} $.
The covariance matrix of $ \mathrm{cov}(\bm{Z}) $ and $ \mathrm{cov}(\bm{X}) $ are estimated and compared.

\subsubsection{Exploration of High-Dimensional Directional Statistics of \texorpdfstring{$ \bm{x}_{t} $}{$ {x}_{t} $}}

\paragraph{We show the sample $ \mathrm{MD} $ and $ \mathrm{MRL} $ of $ \bm{x}_{t} $ first}

The sample $ \mathrm{MRL} $
is the norm of the mean vector $ \frac{1}{t_{2}-t_{1}} \sum_{t\in [t_{1}, t_{2})} \bm{x}_{t} $. 
The sample $ \mathrm{MD} $
is estimated as $ \frac{1}{\|\sum_{t\in [t_{1}, t_{2})} \bm{x}_{t}\|}\sum_{t\in [t_{1}, t_{2})} \bm{x}_{t} $, which is the unitization of the mean vector.

We show
the $ 185 $-dim sample $ \mathrm{MD} $'s as bar charts in the upper part of Figure~\ref{fig:bar_MD_hist},
where the x-axis shows the values of every components.
In detail, the upper subfigure on the left, (\subref{subfig:bar_5year_hist}) $ \mathrm{MD}_{5y} $, is the sample $ \mathrm{MD} $ of the whole $ 5 $ years,
and we sort its order of components by their values from smallest to biggest.\footnote{We give specific details about the meanings of the bar graph. 
	For instance, the sorted $ \mathrm{MD}_{5y} $'s first component is $ -0.1412 $, so the top bar in (\subref{subfig:bar_5year_hist}) is toward left with a length of $ 0.1412 $. 
	The second component is $ -0.1318 $, so the second top bar in (\subref{subfig:bar_5year_hist}) is toward left with a length of $ 0.1318 $. And forth on.
	The last component is $ 0.2517 $, so the bottom bar in (\subref{subfig:bar_5year_hist}) is toward right with a length of $ 0.2517 $.}
The $ 5 $ upper subfigures on the right, (\subref{subfig:bar_2014_hist} to \subref{subfig:bar_2018_hist}), are the sample $ \mathrm{MD} $'s of $ 2014 $ to $ 2018 $, 
and their components are in the same order of $ \mathrm{MD}_{5y} $.
The red dotted line in each one lies between the $ 107 $th and $ 108 $th components of the column vector $ \mathrm{MD} $,
the watershed between positive and negative components in $ \mathrm{MD}_{5y} $.
The downer part of Figure~\ref{fig:bar_MD_hist} is the histograms of the components of the sample $ \mathrm{MD} $ as a $185$ sample,
and they share the same x-axis with the upper part.

\begin{figure}[htpb]
	\centering
	{\includegraphics[width=0.48\textwidth]{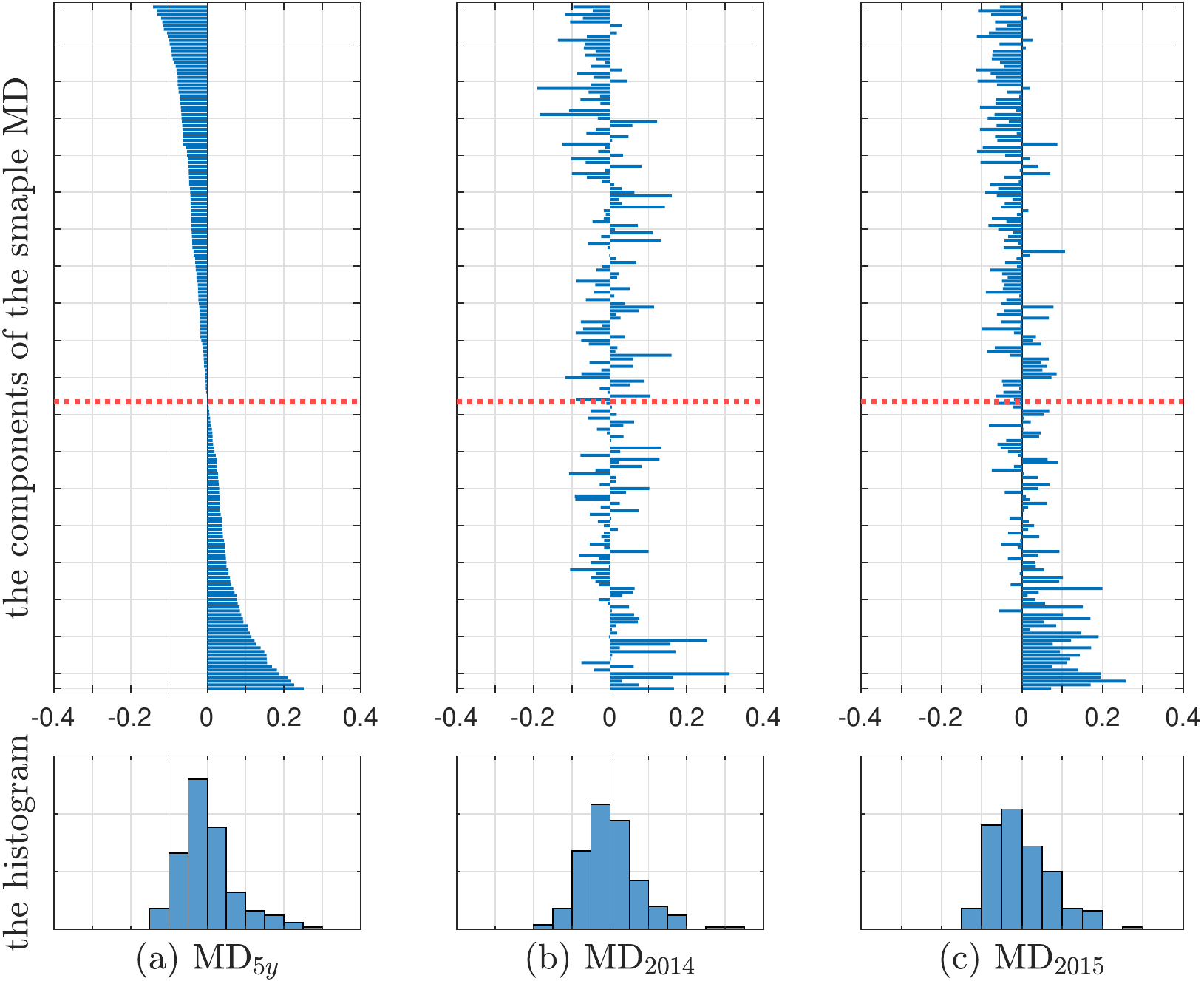}\phantomsubcaption\label{subfig:bar_5year_hist}\phantomsubcaption\label{subfig:bar_2014_hist}\phantomsubcaption\label{subfig:bar_2015_hist}}
	{\includegraphics[width=0.48\textwidth]{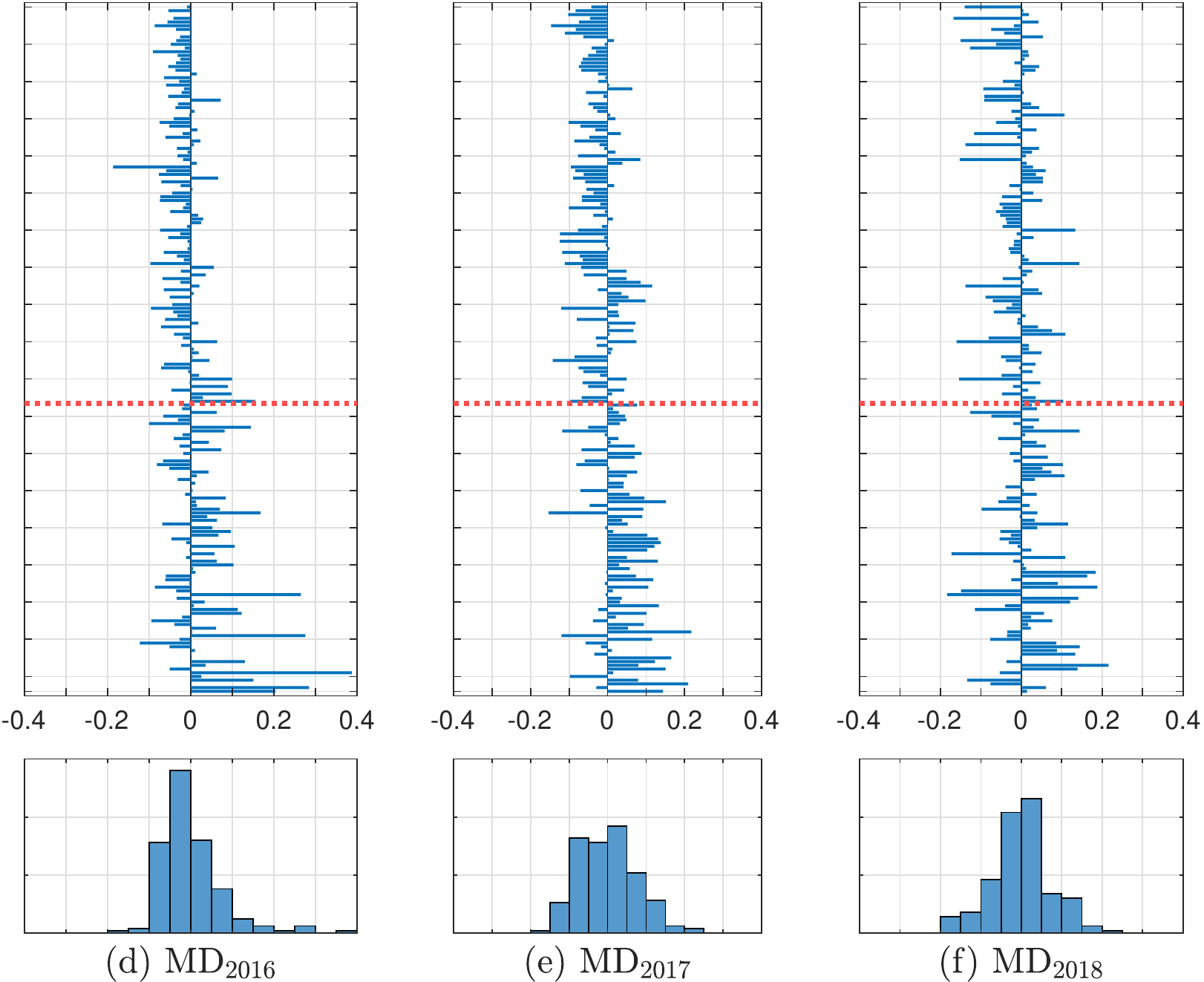}\phantomsubcaption\label{subfig:bar_2016_hist}\phantomsubcaption\label{subfig:bar_2017_hist}\phantomsubcaption\label{subfig:bar_2018_hist}}
	\caption{The sample $ \mathrm{MD} $'s and their components' histograms of different time-window $ \bm{x}_{t} $'s.}
	\label{fig:bar_MD_hist}
\end{figure}

We draw some intuitive conclusions from Figure~\ref{fig:bar_MD_hist}.
First, the six sample $ \mathrm{MD} $'s are different from each other.
Each bar chart has different shapes from others, and we can hardly put them in categories.
We calculate the angle among the $ 6 $ vectors, 
where the smallest is about $ 43^{\circ} $ and the biggest is about $ 106^{\circ} $.
Second, 
the first $ 107 $ components above the red dotted line are different from 
the last $ 78 $ ones below, 
in different time window samples.
It means that the behaviors of
the top-ranked and the bottom-ranked stock returns are different, and we need to treat them separately. 
Compared with $ \mathrm{MD}_{5y} $, the most matched year is $ 2015 $
---
in fact, $ \mathrm{MD}_{2015} $ has the smallest angle with $ \mathrm{MD}_{5y} $, i.e., $ 43^{\circ} $.
The components are more stable than the others every year.
Third, the $ 6 $ histograms are different.
In order to clarify the difference clearly, we need to describe their characteristics.

\begin{table}[htpb]
	\centering
	\caption{Descriptive statistics of the components of the sample $ \mathrm{MD} $'s of different time-window $ \bm{x}_{t} $.}
	\begin{tabularx}{0.8\textwidth}{c||c|Y|Y|Y|Y|Y}
		\hline\hline
		time window & $ 2014$-$ 18 $ & $ 2014 $   & $ 2015 $    & $ 2016 $   & $ 2017 $    & $ 2018 $   \\ \hline\hline 
		median & $ -0.0111 $ & $ -0.0068 $ & $ -0.0126 $ & $ -0.0134 $ & $ -0.0035 $ & $ 0.0053 $ \\ \hline
		skewness & $ 0.9082 $ & $ 0.7680 $ & $ 0.8777 $ & $ 1.8529 $ & $ 0.3477 $ & $ 0.0300 $ \\ \hline
		kurtosis & $ 3.9725 $ & $ 4.7518 $ & $ 3.4797 $ & $ 8.9126 $ & $ 2.7173 $ & $ 3.3746 $ \\ \hline
		positive counts & $ 78 $    & $ 86 $    & $ 79 $    & $ 73 $    & $ 91 $    & $ 94 $ \\ \hline
		negative counts & $ 107 $   & $ 99 $    & $ 106 $   & $ 112 $   & $ 94 $    & $ 91 $ \\ \hline\hline
	\end{tabularx}
	\label{tab:descriptive_statistics_of_components_of_MD}
\end{table}
Table~\ref{tab:descriptive_statistics_of_components_of_MD} lists the descriptive statistics of the components of the sample $ \mathrm{MD} $'s.
It reveals the characteristics of the sample $ \mathrm{MD} $'s from a different viewpoint.
First, 
from $ 2014 $ to $ 2018 $, 
the median changes from negative to positive, 
and the skewness varies from positive to almost $ 0 $.
The number of positive numbers gradually increases to a half with the negative decreases.
All the phenomena can be interpreted as the results of the development of the market.
Second, the samples of $ 2015 $ and $ 2016 $ can be regarded as special ones. 
The median is more negative than others, 
and the skewness and kurtosis are very large.
On the contrary, the samples $ 2014 $, $ 2017 $, and $ 2018 $ are general ones.
The numbers of their positive and negative components are very close.
Their median and skewness are close to $ 0 $.

We list the sample $ \mathrm{MRL} $'s in Table~\ref{tab:sample_MRL}.
The $ 5 $-year sample $ \mathrm{MRL} $ is smaller than other $ 1 $-year ones because of its large sample size and the bias of the statistics of small samples.
For a $ 1 $-year time window,
the sample $ \mathrm{MRL} $ is around $0.06 \sim 0.08$ and less than $0.1$. 
Combined with the information in Table~\ref{tab:descriptive_statistics_of_components_of_MD},
$ 2016 $ is the most different year from others in terms of the cross-section.
\begin{table}[htbp]
	\centering
	\caption{The sample $ \mathrm{MRL} $ of different time-window $ \bm{x}_{t} $'s.}
	\begin{tabularx}{0.8\textwidth}{c||Y|Y|Y|Y|Y|Y}
		\hline\hline
		time window    & $ 2014$-$ 18 $ & $ 2014 $  & $ 2015 $  & $ 2016 $  & $ 2017 $  & $ 2018 $ \\ \hline\hline
		sample size  & 1220  & 245   & 244   & 244   & 244   & 243 \\ \hline
		$\mathrm{MRL}$     & 0.0331 & 0.0658 & 0.0849 & 0.0592 & 0.0804 & 0.0601 \\ \hline\hline
	\end{tabularx}%
	\label{tab:sample_MRL}%
\end{table}%

The scatter matrix $ \overline{\bm{T}} $ about the origin measures the dispersion,
where $ \overline{\bm{T}}_{[t_{1}, t_{2})} := \frac{1}{t_{2}-t_{1}}\sum_{t \in [t_{1}, t_{2})} \bm{x}_{t}\bm{x}_{t}^{\T} $.
We sort eigenvalues of $ \overline{\bm{T}} $ of different time windows samples and show them in Figure~\ref{fig:T_eig}.
If we denote the eigenvalues of $ \overline{\bm{T}} $ from large to small as $ t_{1}, t_{2}, \cdots, t_{i}, \cdots, t_{185} $, and mark one of their corresponding unit eigenvectors as $ \bm{t}_{i} $,
then the larger the $ t_{i} $, the higher the volatility along $ \pm \bm{t}_{i} $.
\begin{figure}[htpb]
	\centering
	\includegraphics[width=0.48\linewidth]{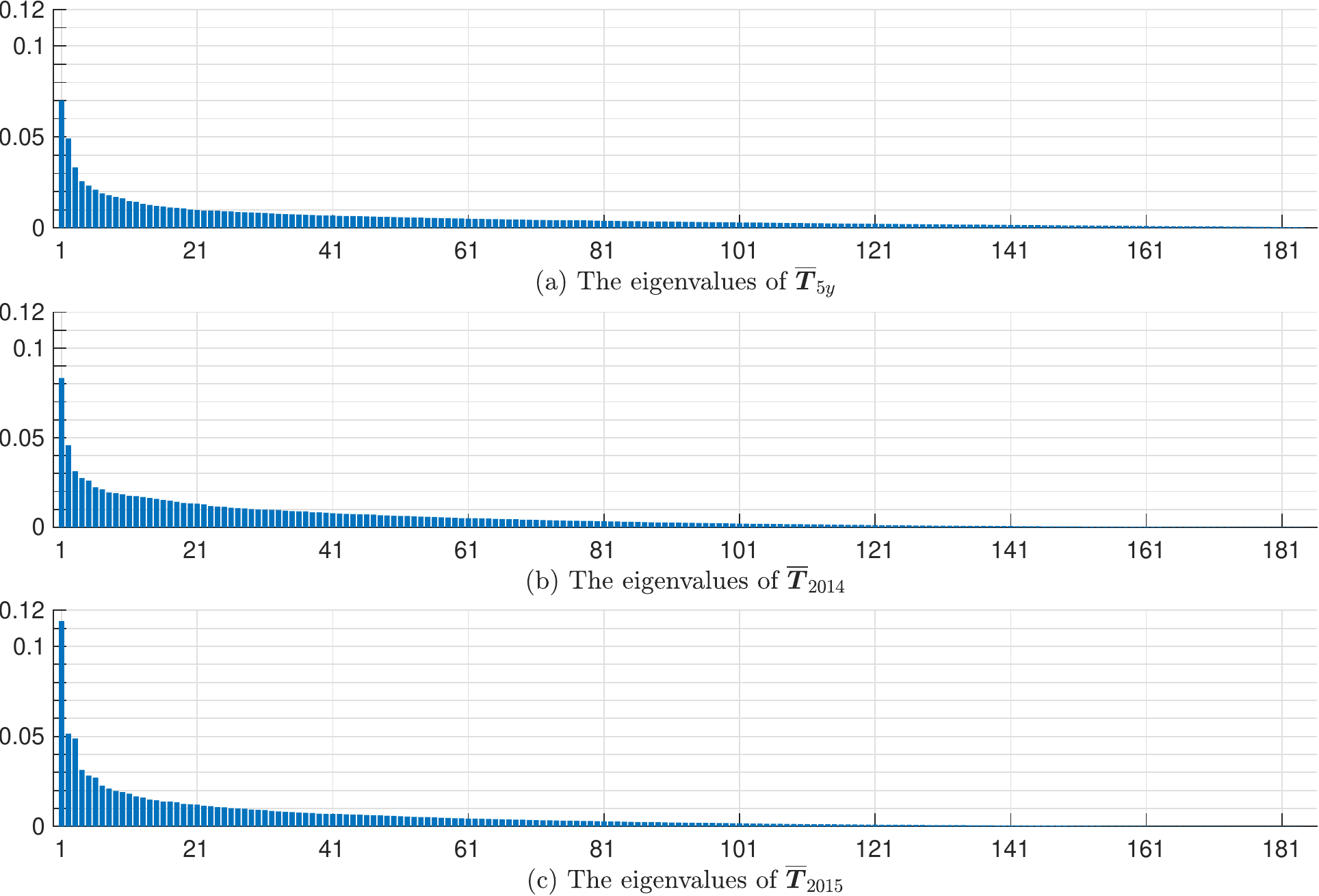}
	\includegraphics[width=0.48\linewidth]{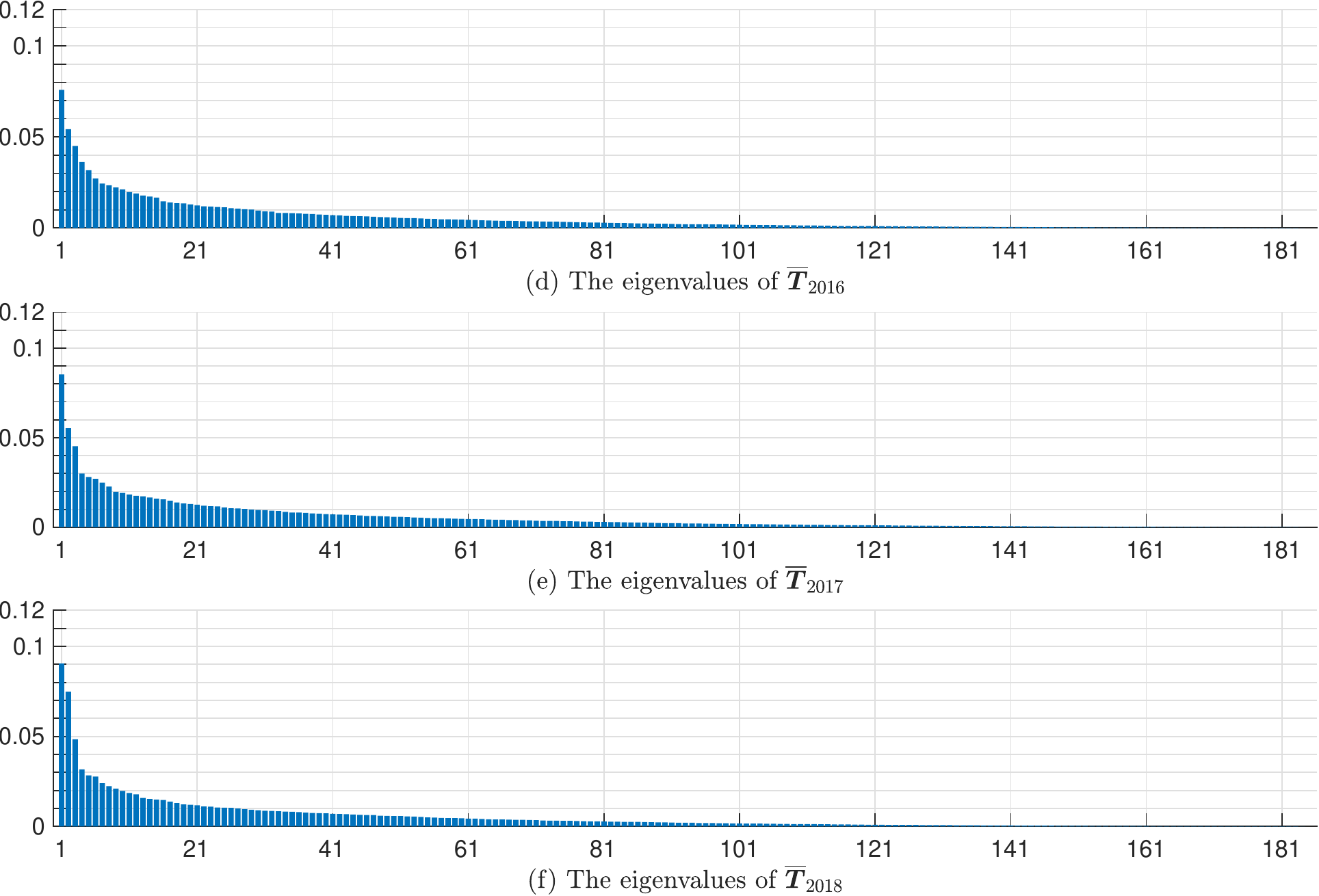}
	\caption{The sorted eigenvalues of the scatter matrix $ \overline{\bm{T}} $ of different time-window $ \bm{x}_{t} $'s.}
	\label{fig:T_eig}
\end{figure}
We interpret the figure.
First, the eigenvalues in each subfigure are all not evenly distributed.
It can be interpreted as that
$ \bm{x}_{t} $ stabilizes in some direction but fluctuates in some others.
Furthermore, the sum of the $ 20 $ largest eigenvalues accounts for more than $ 50\% $ of the sum of all $ 185 $ ones.
The sum of the smallest $ 20 $ ones in each subfigure is about $ 1\% $.
These phenomena indicate that the principal component analysis method may be applicable in the directional statistics of stock returns.
Second, subfigure (c) of $ 2015 $ has the largest maximum eigenvalue, $ t_{1} = 0.1140 $,
which is the only one larger than $ 0.1 $.
Thus, in $ 2015 $, along $ \pm\bm{t}_{1} $, there is the highest volatility.
Meanwhile, in $ 2015 $, the Chinese stock market experienced tremendous ups and downs.
Third, the minimum eigenvalue in each subfigure is $ 0 $, 
and the eigenvectors corresponding to $ 0 $ are $ k\bm{1} ( k\neq 0 )$, which is implied in the definition of the standardization $ \bm{\chi} $.

\paragraph{We then discuss the properties of  $ \widehat{\bm{\iota}}^{\T}\bm{x}_{t}$}

``Non-parametric techniques are almost non-existent on higher-dimensional spheres.'' \citet[\S 10.1, p. 193]{mardia2000directional}.
$ 185 $ is a quite large dimension in statistics.
We can hardly make any further inferences from this high-dimension sample directly. 
Therefore, we take the inner product between a given vector $ \bm{\iota} $ and the $ 185 $-dim sample 
$ \left\{ \bm{x}_{t} \right\}_{t\in [t_1, t_2)} $ to generate a new $1$-dim sample $ \left\{ \bm{\iota}^{\T}\bm{x}_{t} \right\}_{t\in [t_1, t_2)} $.
One of the most natural selections of $ \bm{\iota} $ is the sample $ \mathrm{MD} $ of $ \bm{x}_{t} $,
so we denote the generated sample as $ \left\{ \widehat{\bm{\iota}}^{\T}\bm{x}_{t} \right\}_{t\in [t_1, t_2)} $, where $ \widehat{\bm{\iota}} :=  \frac{1}{\|\sum_{t\in [t_{1}, t_{2})} \bm{x}_{t}\|}\sum_{t\in [t_{1}, t_{2})} \bm{x}_{t} $.

Table~\ref{tab:descriptive_statistics_of_muTx} shows $ \left \{ \widehat{\bm{\iota}}^{\T}\bm{x}_{t} \right \} $'s descriptive statistics.
Figure~\ref{fig:descrpitve_statistics_of_muTx} illustrates $ \left \{ \widehat{\bm{\iota}}^{\T}\bm{x}_{t} \right \} $'s empirical distribution,
where (\subref{subfig:hist_5years}) is the histogram of the
$ 2014$-$2018 $ time window, 
and (\subref{subfig:boxplot_6windows}) is the box plots.
In (\subref{subfig:hist_5years}), we use the orange dashed line to show the p.d.f. of the normal distribution with the same sample mean and standard deviation.

\begin{table}[htpb]
	\centering
	\caption{Descriptive statistics of different time-window $ \left \{ \widehat{\bm{\iota}}^{\T}\bm{x}_{t} \right \} $'s.}
	\begin{tabularx}{0.8\textwidth}{c||Y|Y|Y|Y|Y|Y}
		\hline\hline
		time window & $ 2014$-$ 18 $ & $ 2014 $   & $ 2015 $    & $ 2016 $   & $ 2017 $    & $ 2018 $   \\ \hline \hline
		  {mean}    & $ 0.0331 $  & $ 0.0658 $ & $ 0.0849 $  & $ 0.0592 $ & $ 0.0804 $  & $ 0.0601 $ \\ \hline
		   {standard deviation}    & $ 0.1786 $  & $ 0.1549 $ & $ 0.2823 $  & $ 0.1404 $ & $ 0.1809 $  & $ 0.1479 $ \\ \hline
		{skewness}  & $ -0.1050 $ & $ 0.0658 $ & $ -0.3734 $ & $ 0.2436 $ & $ -0.3272 $ & $ 0.0306 $ \\ \hline
		{kurtosis}  & $ 2.8450 $  & $ 2.9246 $ & $ 2.5179 $  & $ 2.9896 $ & $ 3.0134 $  & $ 2.5283 $ \\ \hline\hline
	\end{tabularx}
	\label{tab:descriptive_statistics_of_muTx}
\end{table}

\begin{figure}[htpb]
	\centering
	\subcaptionbox{The histogram of $ 5 $-year $ \left \{ \widehat{\bm{\iota}}^{\T}\bm{x}_{t} \right \} $.\label{subfig:hist_5years}}
	{\includegraphics[width=0.48\textwidth]{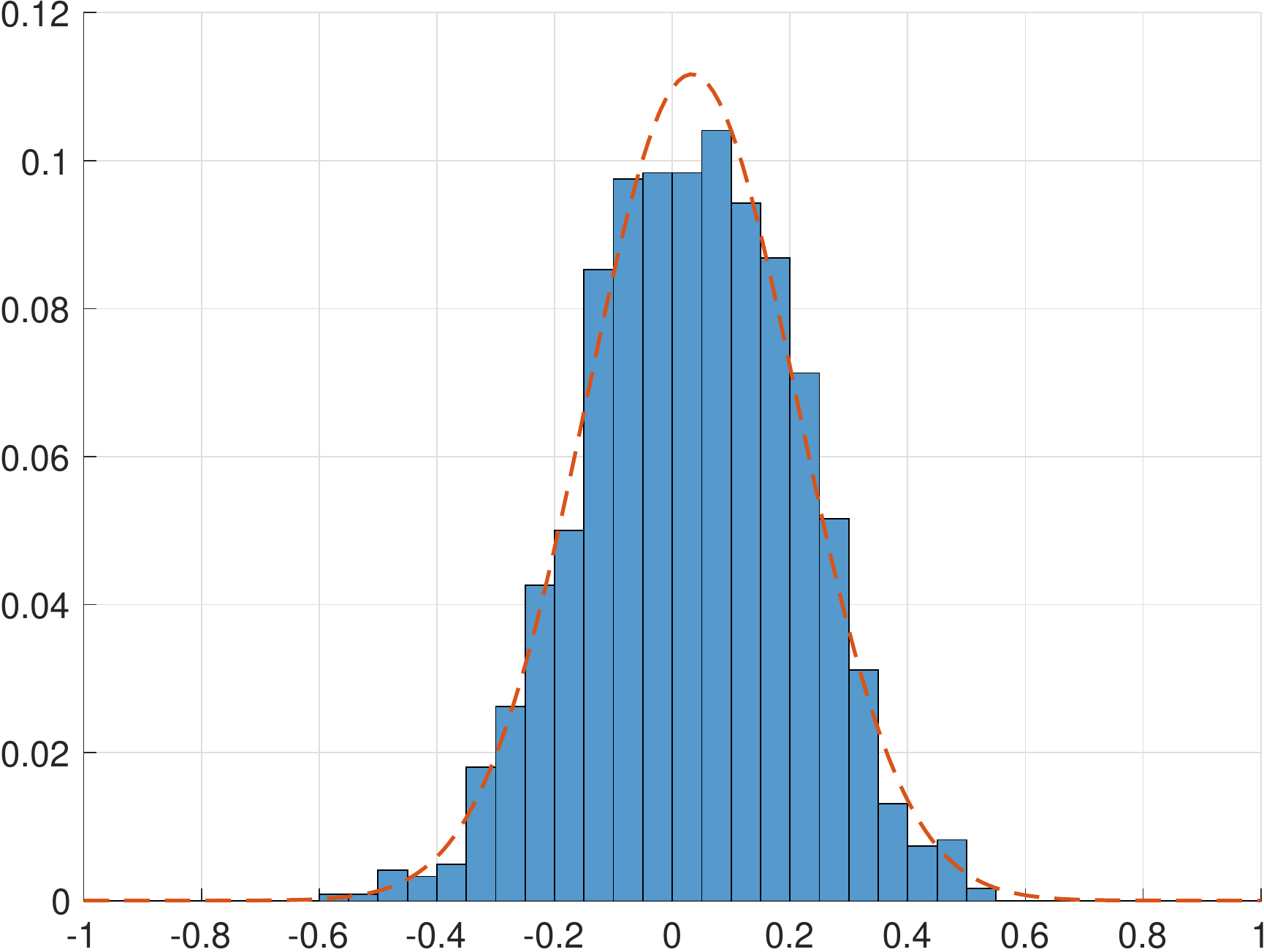}}
	\subcaptionbox{The box plots of different time-window $ \left \{ \widehat{\bm{\iota}}^{\T}\bm{x}_{t} \right \} $'s.\label{subfig:boxplot_6windows}}
	{\includegraphics[width=0.48\textwidth]{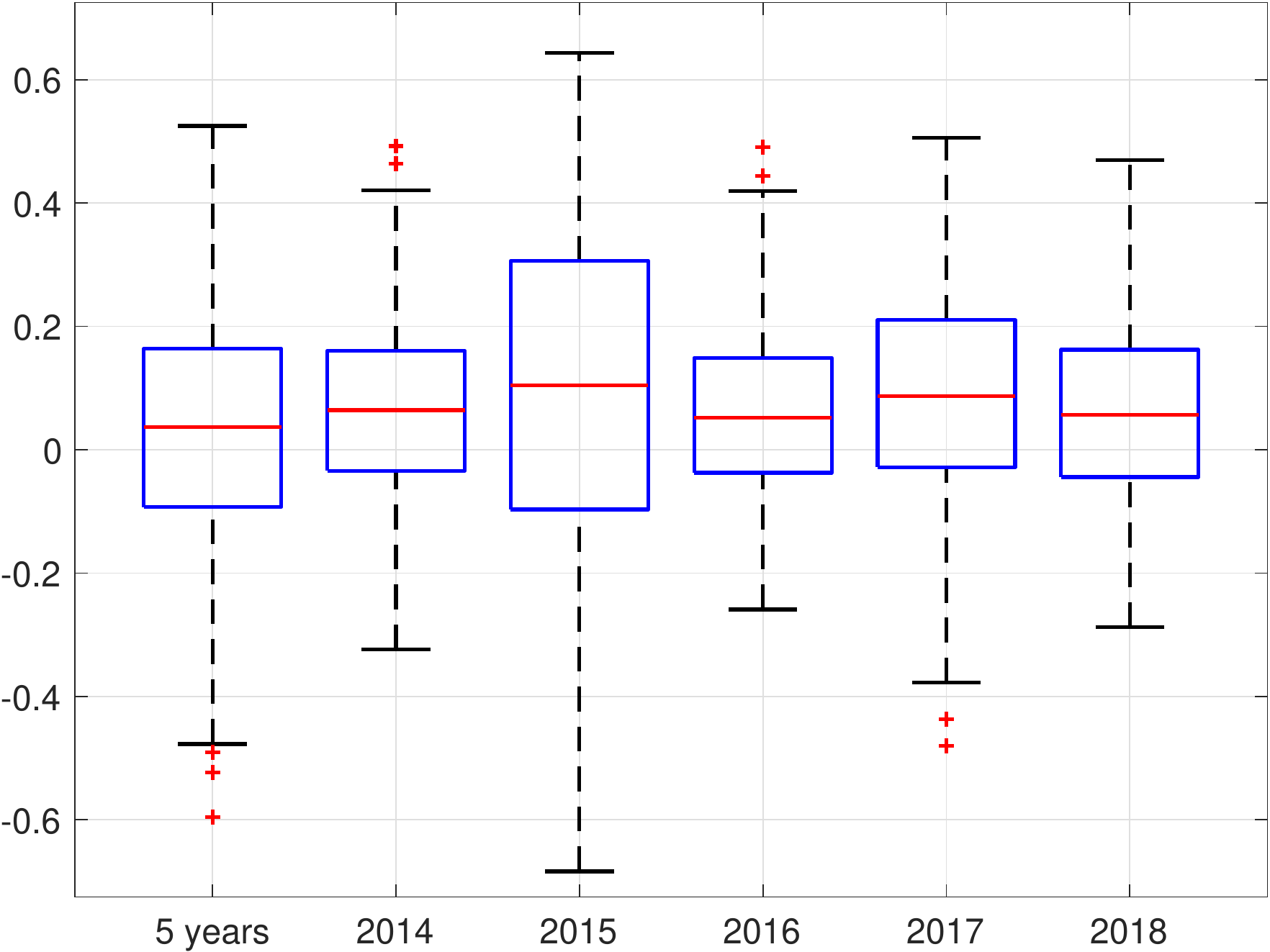}}
	
	\caption{Descriptive statistics of different time-window $ \left \{ \widehat{\bm{\iota}}^{\T}\bm{x}_{t} \right \} $'s.}
	\label{fig:descrpitve_statistics_of_muTx}
\end{figure}

To begin with, we interpret the sample of a $ 5 $-year time window $2014$-$2018$.
First, the mean of the $ 5 $-year sample is $ 0.0331 $, and its sample standard deviation is $ 0.1786 $, indicating
a non-concentrating distribution 
---
there are many sample points whose angles to the sample $ \mathrm{MD} $ are greater than $ 90^{\circ} $.
Its skewness is $ -0.1050 $, implying that the left tail is thicker than the right. 
Its kurtosis is $ 2.8450 $, 
so its distribution has slightly lighter tails than a normal distribution, 
consistent with the intrinsic bounded property of the inner product of two unit vectors.
Second,
in the central part of the subfigure (\subref{subfig:hist_5years}),
the histogram is lower than the p.d.f. of the normal distribution.
It means that the concentration around the mean of $ \bm{\iota}^{\T}\bm{X} $ is less than that of the normal distribution.
The phenomenon may mean that the samples are generated from a multimodal distribution.
Third, the most left box plot in (\subref{subfig:boxplot_6windows}) shows the quartiles and the extreme values of the $ 5 $-year sample.
The median (i.e., the solid red line) is $ 0.0366 $, and there are almost half of the sample points in the blue box,
which is between the two quartiles $ [-0.0930, 0.1637] $.
The number $ 0.0366 $ corresponds to IC in an active investment strategy.
The extreme values all lie on the left tail, as the red +'s lie on the bottom of the box.

Furthermore, we interpret the descriptive statistics of each $ 1 $-year sample $ \left \{ \widehat{\bm{\iota}}^{\T}\bm{x}_{t} \right \} $.
First, $ 2015 $ is a particular year.
$ 2015 $ has the highest sample mean and standard deviation, 
while the smallest sample skewness and kurtosis in Table~\ref{tab:descriptive_statistics_of_muTx}.
Also, the box plot of $ 2015 $ in Figure~\ref{subfig:boxplot_6windows} is the widest.
We suggest that one of the reasons is that in $ 2015 $ the notorious stock market crash broke out.
Second, $ 2014 $, $ 2016 $, and $ 2018 $ are quite similar. Their sample means are all about $ 0.06 $, and their standard deviations are around $ 0.15 $. 
They are positively skewed: in $ 2014 $, the skewness is about $ 0.0658 $; in $ 2016 $, $ 0.2436 $; and in $ 2018 $, $ 0.0306 $.
Their box plots also have similar shapes, no matter from the viewpoint of the quartiles, or of the range of each sample.
The extreme values in $ 2014 $ and $ 2016 $ are all lies on the right tail.
These similar years imply the possible stationarity of $ \bm{\iota}^{\T}\bm{X} $ in most time.
Third, we interpret the sample of $ 2017 $ year.
Its sample mean is $ 0.0804 $, one of the only two greater than $ 0.08 $ besides $ 2015 $.
Its skewness is $ -0.3272 $, which is also one of the only two that smaller than $ -0.30 $.
However, its standard deviation $ 0.1809 $ is lower than the special $ 2015 $ and very close to that in the $ 5 $-year sample.
The $ 2017 $ sample has the largest sample kurtosis $ 3.0134 $, almost identical to that of a normal distribution.
The $ 2017 $ box plot is most similar to the $ 5 $-year box plot --- their quartiles and extreme values are almost the same.

Last but least, 
the sample $ \mathrm{MRL} $ 
of high-dimensional $ \left \{ \bm{x}_{t} \right \} $
and the standard deviation of projected $1$-dim $ \left\{ \widehat{\bm{\iota}}^{\T}\bm{x}_{t} \right\} $ 
interpret $\bm{X}$'s concentration from different viewpoints.
The sample mean in Table~\ref{tab:descriptive_statistics_of_muTx} is equal to the sample $ \mathrm{MRL} $ in Table~\ref{tab:sample_MRL}, which is not a coincidence but a corollary of their definition.
The $ \mathrm{MRL} $
of $ \bm{X} $ plays the opposing role of variance, that the higher the $ \mathrm{MRL} $, the higher the concentration.
So the higher the sample mean in Table~\ref{tab:descriptive_statistics_of_muTx}, the more concentrated $ \bm{X} $ is.
However, $ \bm{\iota}^{\T}\bm{X} $'s sample standard deviation varies consistent with the sample mean in Table~\ref{tab:descriptive_statistics_of_muTx},
so the higher the sample mean, the more volatile the $ \bm{\iota}^{\T}\bm{X} $.

\subsubsection{The Statistics Analysis from \texorpdfstring{$ \bm{z}_{t} $}{$ z_{t} $} to \texorpdfstring{$ \bm{x}_{t} $}{$ x_{t} $}}
\label{subsec:statistics_analysis_z_to_x}

First,
we compare the similarity and difference between the representative components of $ x_{it} $ and their corresponding $ z_{it} $ with basic descriptive statistics. Table~\ref{tab:descriptive_statistics_of_z_y} is the descriptive statistics about two representative stocks:  One is 000001.SZ, Ping An Bank Co., Ltd.; The other is 600018.SH, Shanghai International Port (Group) Co., Ltd.
For simplicity, let $ z_{1t} $ an $ x_{1t} $ represent the returns and the standardized value of 000001.SZ, and $ z_{2t}, x_{2t} $ 600018.SH.

\begin{table}[htpb]
	\centering
	\begin{threeparttable}
		\caption{Descriptive statistics of the representative components.}
		\label{tab:descriptive_statistics_of_z_y}
		\begin{tabularx}{0.8\textwidth}{c||Z|Z||Z|Z}
			\hline\hline
			\multirow{2}{*}{statistics} &                \multicolumn{2}{c||}{000001.SZ}                 &                \multicolumn{2}{c}{600018.SH}                 \\ \cline{2-5}
			& \multicolumn{1}{c|}{$z_{1t}$} & \multicolumn{1}{c||}{$x_{1t}$} & \multicolumn{1}{c|}{$z_{2t}$} & \multicolumn{1}{c}{$x_{2t}$} \\ \hline\hline
			mean               & 0.0005                        & -0.0022                        & 0.0004                        & -0.0026                      \\ \hline
			standard deviation                 & 0.0205                        & 0.0508                         & 0.0235                        & 0.0649                       \\ \hline
			coefficient of variance                 & 41.3871                         & 22.9746			                          & 65.0917                         & 25.3970                        \\ \hline
			minimum$ ^{\text{(a)}} $              & -0.1002                       & -0.1810                        & -0.0999                       & -0.2509                      \\ \hline
			maximum$ ^{\text{(a)}} $                & 0.1000                        & 0.3274                         & 0.1007                        & 0.4661                       \\ \hline\hline
		\end{tabularx}
		\begin{tablenotes}{
				\footnotesize
				\item $ ^{\text{(a)}} $ The minimum and maximum of $ z_{1t} $ and $ z_{2t} $ are resulted from the price limit system, $ \pm 10\% $, in Chinese stock market, leading to the inefficacy of the information expression of daily simple returns.
			}
		\end{tablenotes}
	\end{threeparttable}
\end{table}%

We discuss Table~\ref{tab:descriptive_statistics_of_z_y} in two points.
First, the descriptive statistics of $ z_{i \cdot} $ and $ x_{i \cdot} $ are significantly different,
where the standardization increases the absolute value of sample means individually.
Also, the standard deviation of $ {x}_{i \cdot} $ are larger than that of $ {z}_{i \cdot} $ significantly.
The sample coefficient of variance decreases through the standardization.
Second, we interpret some facts.
The range of $ x_{i \cdot} $, worth emphasizing here, is broader than that of $ z_{i \cdot} $ restricted by the $ 10\% $ price limit.%\footnote{The range of $ z_{i\cdot} $ is defined as $ \sup_{0 < t, s < T}|z_{it} - z_{is}| $.}
 Thus, the standardization reveals the hidden dynamic information. 

Second, the standardization almost eliminates the cross-sectional correlation between each stock.
In detail, we plot two heat maps of the sample cross-sectional correlation coefficient matrix between each stock in Figure~\ref{fig:corr_xy}:
The left one (\subref{subfig:cov_z}) is of $ \bm{z}_{t} $, and the right (\subref{subfig:cov_x}) of $ \bm{x}_{t} = \bm{\chi}(\bm{z}_{t}) $.
\begin{figure}[htpb]
	\centering
	\begin{subfigure}[t]{0.48\textwidth}
		\centering
		\includegraphics[width=\textwidth]{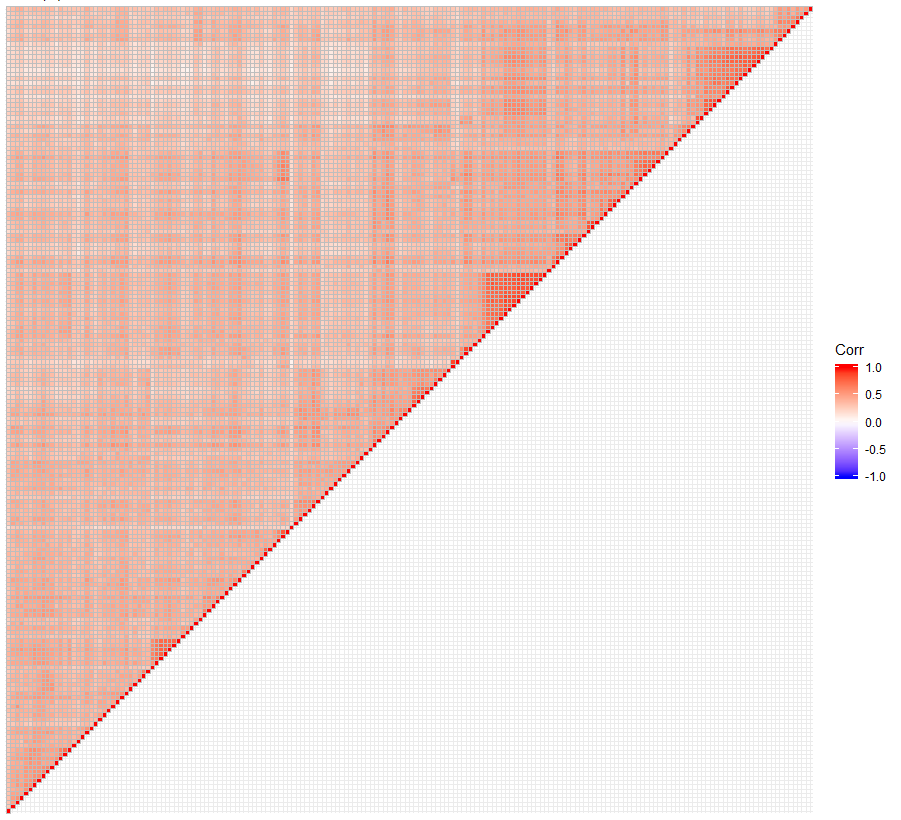}
		\caption{sample covariance matrix of $ \bm{z}_{t} $.}
		\label{subfig:cov_z}
	\end{subfigure}
	\begin{subfigure}[t]{0.48\textwidth}
		\centering
		\includegraphics[width=\textwidth]{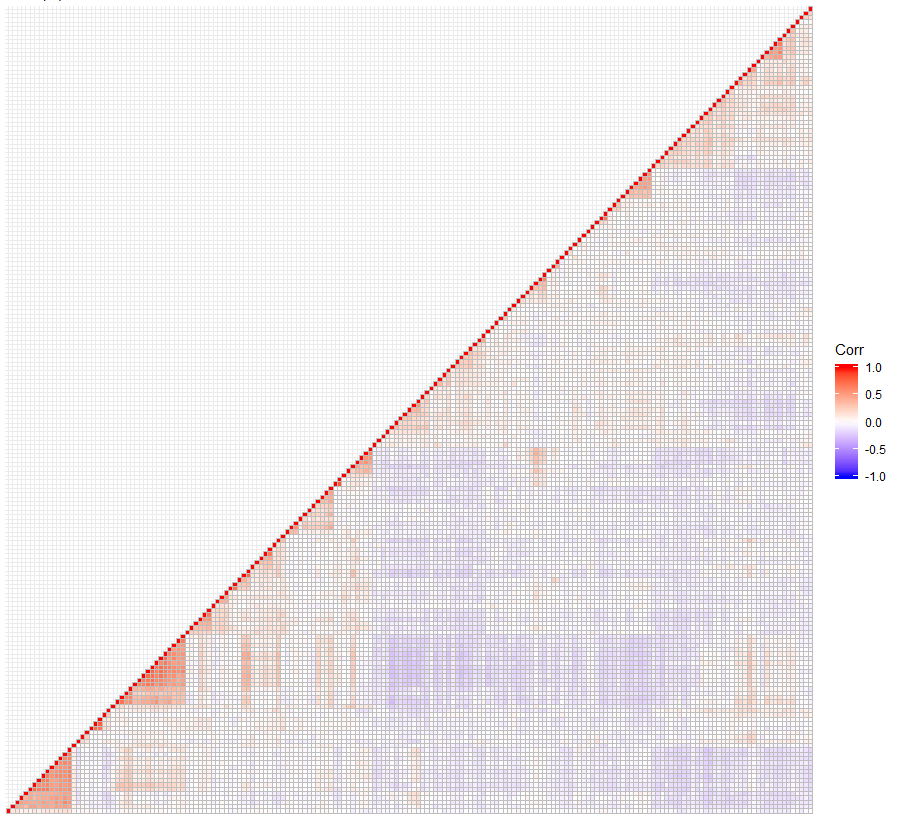}
		\caption{sample covariance matrix of $ \bm{x}_{t} $.}
		\label{subfig:cov_x}
	\end{subfigure}
	\caption[The sample cross-sectional correlation coefficient matrix of $ \bm{z}_{t} $ and $ \bm{x}_{t} $.]{
		The heat maps of the cross-sectional correlation coefficient matrix of $ \bm{z}_{t} $ and $ \bm{x}_{t} $: two 185-by-185 matrix.
		For the purpose of a clear demonstration, these two correlation matrix are processed after permutation along axes, resulting in that the $ (i, j) $-th element in (\subref{subfig:cov_z}) is not necessary that in (\subref{subfig:cov_x}).
		But the permutation does not affect the results.
		What's more, due to the symmetry of correlation matrix, (\subref{subfig:cov_z}) just shows the upper triangular one, while (\subref{subfig:cov_x}) lower.
	}
	\label{fig:corr_xy}
\end{figure}

We interpret Figure~\ref{fig:corr_xy}.
It is difficult to ignore the difference that $ \mathrm{corr}(\bm{z}) $ is greater than $ \mathrm{corr}(\bm{x}) $. In fact, by simple calculation, the mean of	 sample correlation coefficients of $ \bm{z}_{t} $ is $ \mathrm{mean}(\mathrm{corr}(\bm{z})) \approx 0.3855 $, while that of $ \bm{x}_{t} $ is $ \mathrm{mean}(\mathrm{corr}(\bm{x})) \approx -0.0034 $. In addition, the standard deviation of the sample correlation coefficients of $ \bm{z} $, $ \mathrm{std}(\mathrm{corr}(\bm{z})) \approx 0.1043 $, is close to that of $ \bm{x}_{t} $, $ \mathrm{std}(\mathrm{corr}(\bm{x})) \approx 0.1109 $. As a consequence, the standardization from $ \bm{z}_{t} $ to $ \bm{x}_{t} $ eliminates the cross-sectional correlation significantly.\footnote{
	$ \mathrm{mean}(\mathrm{corr}(\bm{z})) := \frac{1}{n(n-1)/2}\sum_{1 \leq i < j \leq n} \mathrm{corr}_{ij}(\bm{z}) $, 
	and $ \mathrm{std}(\mathrm{corr}(\bm{z})) := \frac{1}{\frac{n(n-1)}{2} - 1} \sum_{1 \leq i < j \leq n} \left(\big.\mathrm{corr}_{ij}(\bm{z}) - \mathrm{mean}(\mathrm{corr}(\bm{z}))\right)^{2} $.
}
In finance, it can be interpreted as that the standardization $ \bm{\chi} $ from $ \bm{z}_{t} $ to $ \bm{x}_{t} $ eliminates the \textit{beta} part, to a certain degree.%%

\subsection{The Time Series of $z_{it}$, $x_{it}$, and \texorpdfstring{$ \mathrm{MRL}\left (\big. \bm{\chi}(\bm{Z}) \right ) $}{$ MRL(X(Z)) $}}
\label{subsec:time_series}

In this subsection, 
we discuss the time-series properties.
Because $ \bm{x}_{t} $ is a $ 185 $-dim vector, 
then its time series could not be directly shown.
We first show the time series of two selected representative components
and interpret their characteristics in finance.
Then the time series of $ \mathrm{MRL}\left (\big. \bm{\chi}(\bm{Z}) \right ) $ are given with discussion.

\subsubsection{The Time Series of $ z_{it} $ and $ x_{it} $}

Corresponding to the above \S~\ref{subsec:statistics_analysis_z_to_x}, 
we take the time series analysis of the two representatives: $ {z}_{it} $ and $ {x}_{it} $, $ i = 1, 2 $, 
in Figure~\ref{fig:xy_time_series_histogram}.

\begin{figure}[htpb]
	\centering
	\begin{subfigure}[t]{0.48\textwidth}
		\centering
		\includegraphics[width=\textwidth]{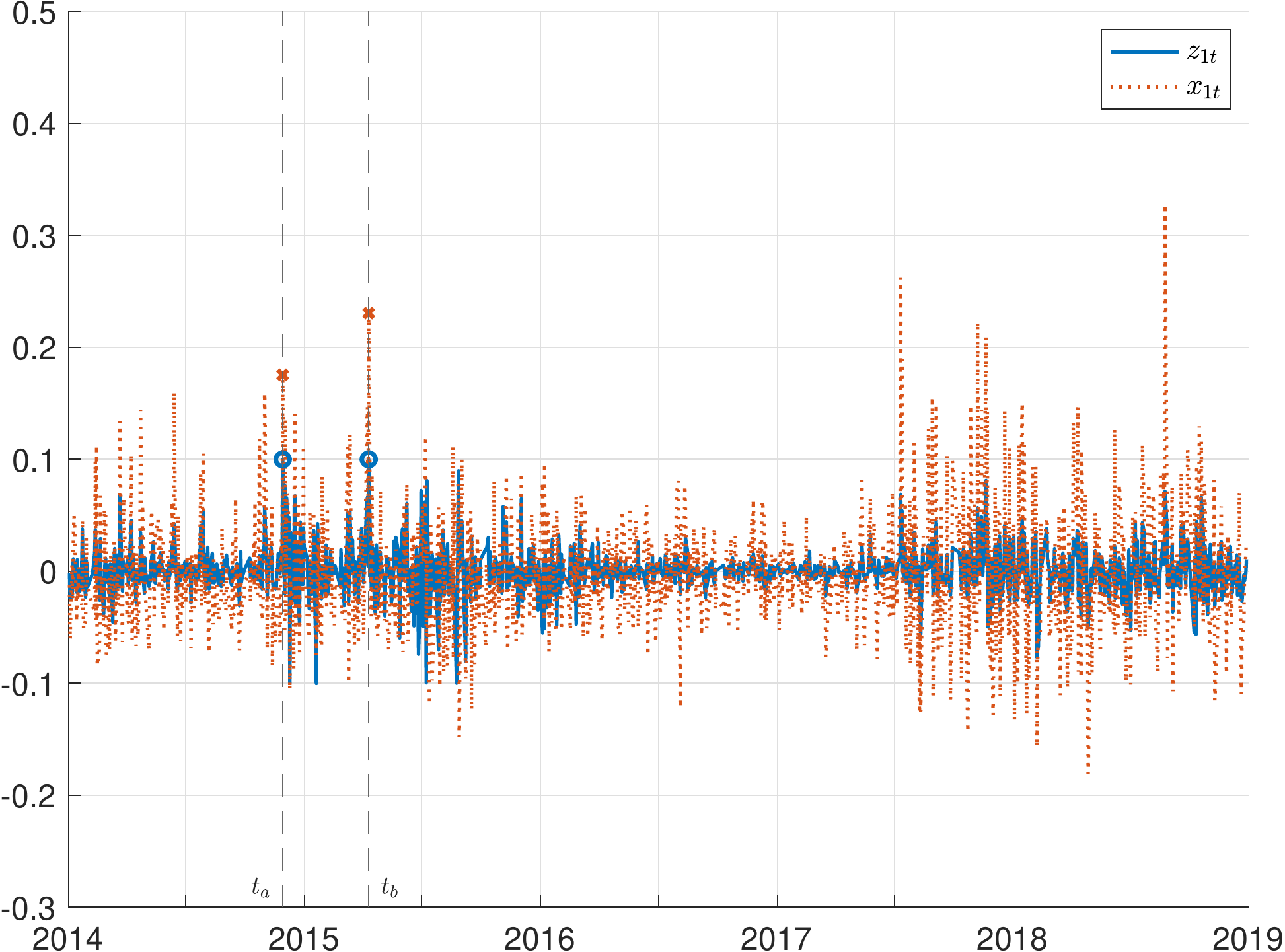}
		\caption{The time series of $ z_{1t} $ and $ x_{1t} $.}
		\label{subfig:xy_time_series_1}
	\end{subfigure}
	\begin{subfigure}[t]{0.48\textwidth}
		\centering
		\includegraphics[width=\textwidth]{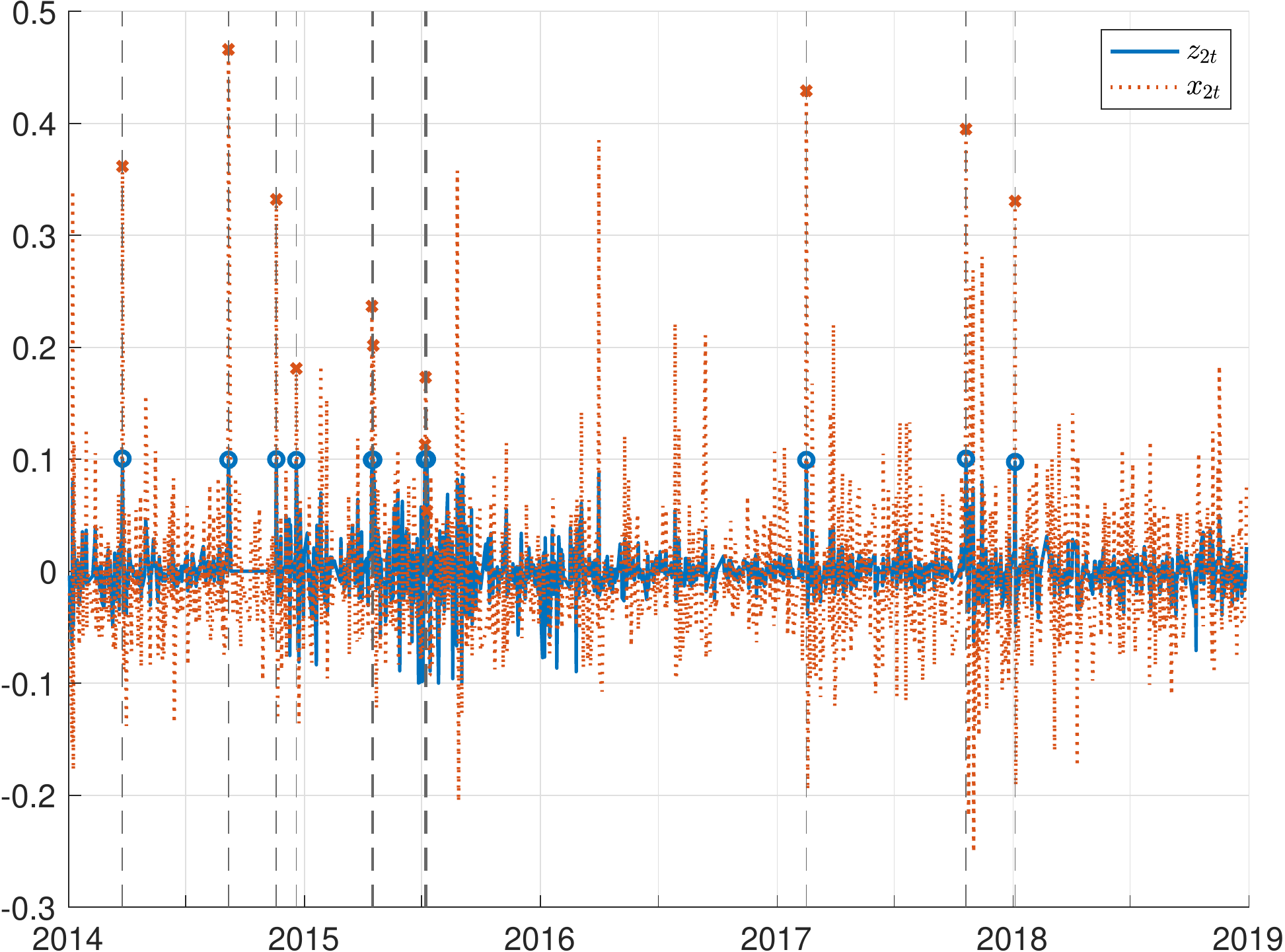}
		\caption{The time series of $ z_{2t} $ and $ x_{2t} $.}
		\label{subfig:xy_time_series_2}
	\end{subfigure}
	\caption[The time series of $ z_{it}, x_{it}, i = 1, 2 $.]{
		The time series of $ z_{it}, x_{it}, i = 1, 2 $.
		In detail, the solid blue lines represent $ z_{it} $ while the dotted orange ones $ x_{it} $. 
		The blue circles indicate the daily up limit of $ z_{it} $, while the orange crosses are the corresponding in  $ x_{it} $.
		The vertical black dashed lines are on the $ t $'s when $ x_{it} $ reaches the daily limits.
	}
	\label{fig:xy_time_series_histogram}
\end{figure}

Three points about the properties of the time series in Figure~\ref{fig:xy_time_series_histogram} are clarified here.
First, fundamental time series analysis shows that the classical hypothesis tests of autoregression and moving average are not significant, and the BIC chooses ARMA($ 0, 0 $).
Second, the time series of $ x_{it} $ has more information than that of $ z_{it} $. 
For instance, the solid blue lines of $ z_{it} $ are strictly bounded in $ [-0.1, 0.1] $, while the dotted orange lines of $ x_{it} $ are not.
Third, it is easy to see similar volatility clustering.
The standardization does not change the property of volatility clustering.

We interpret Figure~\ref{fig:xy_time_series_histogram} in finance.
First, $ z_{it} $'s and $ x_{it} $'s reveal different aspects of security returns.
In detail,
$ z_{it} $'s are absolute, while $ x_{it} $'s are relative.
$ z_{it} $'s faithfully represent the returns of security, but it requires additional information (such as the market index) to attribute the returns.
In contrast, one can not infer the exact returns from $ x_{it} $'s, but she can speculate its relative rank in the whole market.
Second, $ x_{it} $'s can break through some limitation implied $ z_{it} $'s in the financial market: such as the price limit.
The price limit sometimes suppresses the market liquidity, so that $ z_{it} $'s cannot adequately express market information.
However, $ x_{it} $'s can give different explanations for the same $ z_{it} $'s that reach the daily limit and reveal more market information to a certain degree.
For instance, in Figure~\ref{subfig:xy_time_series_1}, on 
$ t_{a} = $ 2014-11-28 and $ t_{b} = $ 2015-04-10, 
$ z_{1t_{a}} $ and $ z_{1t_{b}} $ reach the up limit resulting in $ z_{1t_{a}} = z_{1t_{b}} $, 
marked as dark circles. 
However, $ x_{1t_{a}} < x_{1t_{b}} $, as the orange crosses. 
Then one can conclude that the stock ranked higher on $ t_{b} $ than on $ t_{a} $. The number of stocks that reach the up limits on $ t_{a} $ is about twice that on $ t_{b} $, so on $ t_{b} $ the stock is more superior.

\subsubsection{The Time Series of \texorpdfstring{$ \mathrm{MRL}\left (\big. \bm{\chi}(\bm{Z}) \right ) $}{$ {MRL} ( {\chi}({Z})  ) $}}

The rolling calculation $ \left\Vert \overline{\bm{x}}(T) \right\Vert_{t} := \left\Vert \frac{1}{T}\sum_{s=t-T+1}^{t} \bm{x}_{s} \right\Vert $ estimates $ \mathrm{MRL}\left (\big. \bm{\chi}(\bm{Z}) \right ) $ at $ t $,
where $ T $ is the length of the rolling period and $ t $ is the time.
According to Proposition~\ref{prop:ET},
$ \mathrm{MRL}\left (\big. \bm{\chi}(\bm{Z}) \right ) = \mathbb{E}  T_{\bm{\chi}(\bm{Z})}\left (\Big.\mathrm{MD}\left (\big. \bm{\chi}(\bm{Z}) \right )\right ) $,
so we can estimate the maximum of $ \mathrm{IC} $'s expectation by the $ \mathrm{MRL} $.
To compare the time series properties, 
we also calculate the CSSD of the rolled stock returns.\footnote{$ \mathrm{CSSD}_{t} := \frac{1}{\sqrt{n}} \left\Vert P\widetilde{\bm{z}}_{t} \right\Vert = \sqrt{\frac{1}{n} \sum_{i=1}^{n} \left(\widetilde{z}_{it} - \overline{\widetilde{z}}_{t} \right)^{2}} $, where $ \widetilde{\bm{z}}_{t} := \frac{1}{T}\sum_{s=t-T+1}^{t} \bm{z}_{s} $ and $ \overline{\widetilde{z}}_{t} := \frac{1}{n} \sum_{i=1}^{n} \widetilde{z}_{it} $.}
Taking $ T = 20 $, the time series of 
$ \mathrm{MRL}\left (\big. \bm{\chi}(\bm{Z}) \right ) $ and CSSD are in Figure~\ref{fig:rolling20}.

\begin{figure}[htpb]
	\centering
	\includegraphics[width=0.5\textwidth]{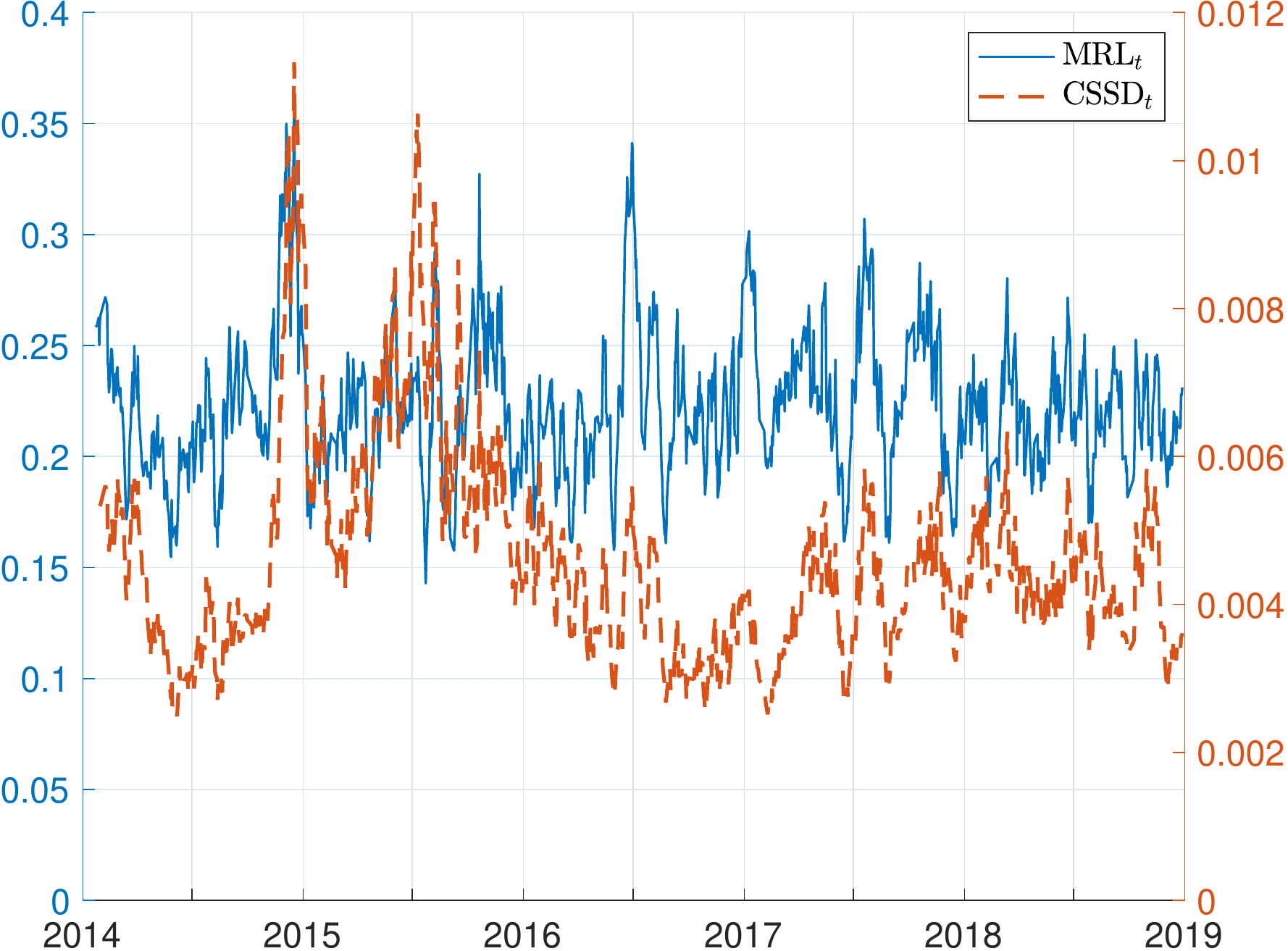}
	\caption[The time series of rolling $ \mathrm{MRL} $ and CSSD.]{
		The time series of rolling $ \mathrm{MRL} $ and CSSD.
		The solid blue line is the estimated $ \mathrm{MRL}_{t} $ with y-axis ticks on the left, while the dashed orange is the CSSD$ _{t} $ on the right.
		CSSD$ _{t} $ is also estimated from the mean returns of rolling period $ T = 20 $.
	}
	\label{fig:rolling20}
\end{figure}

We interpret Figure~\ref{fig:rolling20}.
First, we provide some descriptive statistics: For $ \mathrm{MRL}\left (\big. \bm{\chi}(\bm{Z}) \right ) $, the sample mean is $ 0.2224 $, and the standard deviation is $ 0.0323 $.
For CSSD, the sample mean is $ 0.0046 $, and the standard deviation is $ 0.0014 $. 
Furthermore, it is difficult to ignore the mean-reverting property of them.
Last, there is some co-movement between $ \mathrm{MRL}\left (\big. \bm{\chi}(\bm{Z}) \right ) $ and CSSD. The correlation coefficient of these two time series is about $ 0.4574 $. Moreover, the most correlated periods are about January 2015, when the Chinese stock market suffered the notorious meltdown crash.

\section{Conclusion}\label{sec:conclusion}

The paper focuses on some probabilistic properties of high-dimensional directional statistics and their applications in financial investment.

$ T_{\bm{\chi}(\bm{Z})}(\bm{\theta}) $ is a linear combination of the components of $ \bm{\chi}(\bm{Z}) $.
In definition, $ \bm{\chi}(\bm{Z}) $ is the standardized random vector generated by linear and directional projections, which represents the standardization of next-period cross-sectional returns.
Theoretically, we first give a representation theorem on $ \bm{X}(\bm{Z}) $ and its $ \mathrm{MD} $ and $ \mathrm{MRL} $.
Then, the covariance matrix of $ \bm{\chi}(\bm{Z}) $ is expressed and simplified by a projected particular random vector.
We derive the closed-form expression of the matrix by expanding the methodology in \citet{presnell2008MRL_PN}.
Last, we prove that the solutions to the maximization of $ \mathbb{E} T_{\bm{\chi}(\bm{Z})}(\bm{\theta}) $ could be expressed as the $ \mathrm{MD} $ and $ \mathrm{MRL} $ in directional statistics.
The solutions to the minimization of $ \mathrm{var} \left( T_{\bm{\chi}(\bm{Z})}(\bm{\theta}) \right) $ are the eigenvectors corresponding to the second smallest eigenvalue of $ \bm{\chi}(\bm{Z}) $.
According to the theoretical results, $ \mathrm{IC} $ is influenced by the dispersion and correlation of cross-sectional returns positively, while volatility negatively.
These conclusions are consistent with those in the finance literature.

Our simulation analysis supplements the theoretical results. 
We intuitively simulate $ \mathrm{MD}\left(\big.\bm{\chi}(\bm{Z})\right) $ by different parameters $ \bm{\mu} $ and $ \Sigma $.
We also illustrate $ T_{\bm{\chi}(\bm{Z})}(\bm{\theta}) $'s distributions, 
showing the impact of stock number $ n $ and heteroscedasticity: The more stocks, the less dispersive the $ T_{\bm{\chi}(\bm{Z})}(\bm{\theta}) $; In some heteroscedastic condition, the distribution of $ T_{\bm{\chi}(\bm{Z})}(\bm{\theta}) $ may be multimodal.
Then, we compare the theoretical and simulation results of the $ \mathrm{MRL}\left (\bm{\chi}(\bm{Z})\right ) $.

The empirical studies reveal that the standardization of $ 185 $-dim returns behaves significantly different statistical characteristics from the original data and excavates more information about the whole market. 
For one thing,
the sample $ \mathrm{MD} $'s of $ \bm{x}_{t} $ vary significantly during different time windows, whose components' descriptive statistics could hardly be categorized, and the smallest angle between them is $ 43^{\circ} $.
Notably, the components of the top- and bottom-ranked of the sample $ \mathrm{MD} $ also behave differently.
The sample $ \mathrm{MRL} $ and scatter matrix $ \overline{\bm{T}} $ of $ \bm{x}_{t} $ perform similarly during different time windows.
The $ \widehat{\bm{\iota}}^{\T}\bm{x}_{t} $ reveals not only the above difference but also the possible multimodality of the distribution of the sample. 
It provides a different view of the concentration of the sample rather than the $ \mathrm{MRL} $.
For another, $ x_{it} $ keeps most information of $ z_{it} $: $ \bm{x}_{t} $ and $ \bm{z}_{t} $ are highly correlated, and $ {x}_{it} $ maintains the similar volatility clustering property of $ {z}_{it} $.
The sample covariance matrix of $ \bm{z}_{t} $ and $ \bm{x}_{t} $ is significantly different, implying that the standardization eliminates the correlation.
Meanwhile, the time series of $ x_{it} $ is bounded, and mean-reverting, different from that of $ z_{it} $ and revealing more information. 
The time series about $ \mathrm{MRL}\left (\bm{\chi}(\bm{Z})\right ) $ also interpret the market condition.
These observations in the Chinese stock market were rarely found and revealed.

This paper is the first to apply high-dimensional directional statistics to financial investment strategy. 
We believe that this is a very potential research area. 
The cross-sectional standardization keeps the cross-rank while eliminating noise, having a higher signal-to-noise ratio for better prediction. 
High-dimensional directional statistics has become a significantly theoretical topic because of its complicated support set and few available tools.

\section*{Acknowledgments}
The authors report no conflicts of interest.
The authors alone are responsible for the content and writing of the paper.

\appendix

\section{Simulation Parameters}\label{apx:sim_para}

The parameters of the $ 50 $-dim $ \bm{\mu}_{50} $ and $ \Sigma_{50} $ are omitted due to the length. 
But it is worth mentioning, that the first $ 10 $ components of the $ \bm{\mu}_{50} $ is just $ \bm{\mu}_{10} $.
$ \Sigma_{10} $ is the $ 10 $th order leading principal submatrix of $ \Sigma_{50} $.
The parameters of the $ 10 $-dim $ \bm{\mu}_{10} $ and $ {\Sigma}_{10} $ in Section~\ref{sec:simulation} are as follows:
\begin{align*}
\bm{\mu}_{10}
=
10^{-4}\times
\begin{pmatrix}
4.60 \\7.83 \\14.78 \\-16.32 \\4.50 \\10.26 \\-3.22 \\0.39 \\-3.99 \\-4.41
\end{pmatrix}
,
\Sigma_{10}
=
10^{-4}\times
\begin{pmatrix}
3.67  & 2.26  & 0.98  & 0.75  & 1.54  & 0.72  & 0.16  & 1.48  & -0.05 & -0.08 \\
& 6.60  & 0.96  & 1.31  & 1.57  & 1.01  & 0.11  & 1.16  & -0.32 & -0.04 \\
&       & 5.72  & 1.29  & 0.97  & 1.60  & -0.20 & 0.41  & -0.38 & 0.19 \\
&       &       & 4.74  & 1.69  & 0.96  & 0.06  & 0.35  & 0.41  & 0.19 \\
&       &       &       & 5.11  & 0.97  & -0.04 & 0.72  & -0.11 & -0.06 \\
&       &       &       &       & 11.86 & 0.20  & 0.01  & 0.81  & 0.59 \\
&       &       &       &       &       & 1.47  & 0.07  & 0.50  & 0.23 \\
&       &       &       &       &       &       & 1.32  & -0.12 & -0.01 \\
&       &       &       &       &       &       &       & 6.10  & 0.64 \\
&       &       &       &       &       &       &       &       & 3.77 \\
\end{pmatrix}
.
\end{align*}
Particularly,
$ \Sigma'_{10} = \Lambda \Sigma_{10} \Lambda, \Lambda = \mathrm{diag}(1, 1, 1, 1, 1, 3, 1, 1, 1, 1) $.
It is a heteroscedastic case where a stock has high volatility.
The $ 3 $-dim $ \bm{\mu}_{3} $ is the first three components of $ \bm{\mu}_{10} $,
and $ \Sigma_{3} $ are the $ 3 $rd order leading principal submatrix of $ \Sigma_{10} $.

\newpage
\section{Proofs as Online Supplement}\label{apx:proof}

We list all necessary proofs.

\subsection{The Proof of Theorem~\ref{thm:chi_Z_UZ}}
\label{proof:chi_Z_UZ}

\begin{proof}[The proof of Theorem~\ref{thm:chi_Z_UZ}.]
We prove the theorem by three steps.

First,
one can show that
$ U, V,$ and $ W $ are all orthogonal matrix.
By the definition, we have
$ 
V^{\T}PV 
= 
\begin{psmallmatrix}
I_{n-1} & \\
& 0
\end{psmallmatrix}
$ 
and
$
U^{\T}PU
= 
\begin{psmallmatrix}
I_{n-1} & \\
& 0
\end{psmallmatrix},
$.

Second, we prove \eqref{equ:chi_Z_Uxi}.
\begin{align*}
\bm{\chi}(\bm{Z})
=&\ 
\frac{P\bm{Z}}{\|P\bm{Z}\|}
=
U\frac{U^{\T}P\bm{Z}}{\|U^{\T}P\bm{Z}\|}
.
\end{align*}
Note that $ 
U^{\T} P
= 
W^{\T}V^{\T} P
= 
W^{\T}
\begin{psmallmatrix}
I_{n-1} & \\
& 0
\end{psmallmatrix} 
V^{\T} 
=
\begin{psmallmatrix}
I_{n-1} & \\
& 0
\end{psmallmatrix} 
W^{\T}
V^{\T} 
=
\begin{psmallmatrix}
I_{n-1} & \\
& 0
\end{psmallmatrix} 
U^{\T}
$,
so
\begin{align*}
\bm{\chi}(\bm{Z})
=&\
U
\frac{\begin{pmatrix}
	I_{n-1} & \\
	& 0
	\end{pmatrix} 
	U^{\T} \bm{Z}}{\left \|\begin{pmatrix}
	I_{n-1} & \\
	& 0
	\end{pmatrix} 
	U^{\T} \bm{Z}\right \|}
=
U
\frac{\bm{\xi}}{\|\bm{\xi}\|}
.
\end{align*}
Because the last component of $ \bm{\xi} $ is $ 0 $, we have
$ \frac{\bm{\xi}}{\|\bm{\xi}\|} = \begin{psmallmatrix}
\frac{\bm{\xi}_{n-1}}{\|\bm{\xi}_{n-1}\|} \\
0
\end{psmallmatrix} $.

Third, we just need to show the diagonal of the covariance matrix of $ \bm{\xi}_{n-1} $.
\begin{align*}
\mathrm{cov}\left (
\bm{\xi}_{n-1}
\right )
=&\
\mathrm{cov}\left (
\begin{psmallmatrix}
I_{n-1} & \bm{0}_{n-1}
\end{psmallmatrix}
U^{\T}\bm{Z} 
\right )
\\
=&\
\begin{psmallmatrix}
I_{n-1} & \bm{0}_{n-1}
\end{psmallmatrix}
U^{\T}
\Sigma
U
\begin{psmallmatrix}
I_{n-1} & \bm{0}_{n-1}
\end{psmallmatrix}
\\
=&\
\begin{psmallmatrix}
I_{n-1} & \bm{0}_{n-1}
\end{psmallmatrix}
W^{\T}
V^{\T}
\Sigma
V
W
\begin{psmallmatrix}
I_{n-1} & \bm{0}_{n-1}
\end{psmallmatrix}
\\
=&\
\begin{psmallmatrix}
W_{n-1} & \bm{0}_{n-1}
\end{psmallmatrix}
V^{\T}
\Sigma
V
\begin{psmallmatrix}
W_{n-1} & \bm{0}_{n-1}
\end{psmallmatrix}
.
\end{align*}
By the definition of $ W_{n-1} $, we prove that $ \mathrm{cov}\left (
\bm{\xi}_{n-1}
\right ) $ is diagonal.
\end{proof}

\subsection{The Proof of Theorem~\ref{thm:MD_MRL_ndim}}
\label{proof:MD_MRL_ndim}

\begin{proof}[The proof of Theorem~\ref{thm:MD_MRL_ndim}]
	We first transform $ \bm{\chi}(\bm{Z}) $ into $ \bm{\xi} $ and show the distribution of $ \bm{\xi} $ by the construction of $ \Xi $.
	Then use Proposition~\ref{prop:presnell} to complete the proof.
	
	To begin with, by Theorem~\ref{thm:chi_Z_UZ} and Corollary~\ref{cor:numerical_characteristics_xi}, we take $ W_{n-1}=I_{n-1} $, so $ U = V $ in \eqref{equ:V_short}.
	Then \eqref{equ:MD_chi_Z} and \eqref{equ:MRL_chi_Z} show that
	\begin{align*}
	\mathrm{MD}\left(\big.\bm{\chi}(\bm{Z})\right)
	=&\ 
	U
	\begin{pmatrix}
	\mathrm{MD}
	\left(
	{
		\bm{\xi}_{n-1}
	}/{
		\left\|
		\bm{\xi}_{n-1}
		\right\|
	}
	\right)
	\\
	0
	\end{pmatrix}
	,
	\\
	\mathrm{MRL}\left(\big.\bm{\chi}(\bm{Z})\right)
	=&\ 
	\mathrm{MRL}
	\left(
	{
		\bm{\xi}_{n-1}
	}/{
		\left\|
		\bm{\xi}_{n-1}
		\right\|
	}
	\right)
	,
	\end{align*}
	where 
	\begin{align*} 
	\bm{\xi}_{n-1}
	=
	\begin{pmatrix}
	I_{n-1} & \bm{0}_{n-1}
	\end{pmatrix}
	U^{\T}\bm{Z} 
	\sim
	N
	\left(
	\begin{pmatrix}
	I_{n-1} & \bm{0}_{n-1}
	\end{pmatrix}
	U^{\T}\bm{\mu}
	,
	\sigma^{2}
	\begin{pmatrix}
	I_{n-1} & \bm{0}_{n-1}
	\end{pmatrix}
	U^{\T}
	\Xi
	U
	\begin{pmatrix}
	I_{n-1} & \bm{0}_{n-1}
	\end{pmatrix}^{\T}
	\right)
	.
	\end{align*}
	Note that
	$ 
	U^{\T}\Xi U = \begin{pmatrix}
	(1-\rho)I_{n-1} & \\
	& 1 + (n-1)\rho
	\end{pmatrix} 
	$,
	so
	\begin{align*}
	\bm{\xi}_{n-1} \sim N
	\left(
	\begin{pmatrix}
	I_{n-1} & \bm{0}_{n-1}
	\end{pmatrix}
	U^{\T}\bm{\mu}
	,
	\sigma^{2}
	(1-\rho)
	I_{n-1}
	\right)
	.
	\end{align*}
	
	Last, Proposition~\ref{prop:presnell} shows that
	\begin{align*}
	\mathrm{MD}\left({\bm{\xi}_{n-1}}/{\Vert \bm{\xi}_{n-1} \Vert}\right) 
	:=&\
	\frac{\begin{pmatrix}
		I_{n-1} & \bm{0}_{n-1}
		\end{pmatrix}
		U^{\T}\bm{\mu}}
	{\left \|
		\begin{pmatrix}
		I_{n-1} & \bm{0}_{n-1}
		\end{pmatrix}
		U^{\T}\bm{\mu}
		\right \|},
	\\
	\mathrm{MRL}\left({\bm{\xi}_{n-1}}/{\Vert \bm{\xi}_{n-1} \Vert}\right) 
	:=&\
	\varrho_{n-1}\left(\frac{\Vert \begin{pmatrix}
		I_{n-1} & \bm{0}_{n-1}
		\end{pmatrix}
		U^{\T}\bm{\mu} \Vert}{\sigma \sqrt{1-\rho}}\right)
	.
	\end{align*}
	So
	\begin{align*}
	\mathrm{MD}\left(\big.\bm{\chi}(\bm{Z})\right)
	=&\ 
	U
	\frac{\begin{pmatrix}
		I_{n-1} & \\
		& \bm{0}
		\end{pmatrix}
		U^{\T}\bm{\mu}}
	{\left \|
		\begin{pmatrix}
		I_{n-1} & \\
		& \bm{0}
		\end{pmatrix}
		U^{\T}\bm{\mu}
		\right \|}
	=
	\frac{P\bm{\mu}}{\|P\bm{\mu}\|}
	,
	\\
	\mathrm{MRL}\left(\big.\bm{\chi}(\bm{Z})\right)
	=&\ 
	\varrho_{n-1}\left(\frac{\Vert P\bm{\mu} \Vert}{\sigma \sqrt{1-\rho}}\right)
	.
	\end{align*}
\end{proof}

\subsection{The Complete Proof of Theorem~\ref{thm:covariance_matrix}}
\label{proof:covariance_matrix}

\begin{proof}[The details of the proof of Theorem~\ref{thm:covariance_matrix}]

We first show $ \mathrm{cov}\left(\frac{\bm{\xi}}{\|\bm{\xi}\|}\right) $ is a diagonal matrix like \eqref{equ:cov_xi}, and then prove its closed-form expressions.

$ \mathrm{cov}\left(\frac{\bm{\xi}}{\|\bm{\xi}\|}\right) = \begin{psmallmatrix}
\mathrm{var}\left(\frac{\xi_{1}}{\|\bm{\xi}\|}\right)  & \\
& \mathrm{var}\left(\frac{\xi_{2}}{\|\bm{\xi}\|}\right)  I_{n-1}
\end{psmallmatrix} $.
The diagonal is deducted by the symmetry and independence.
Specifically,
for one thing, $ \mathrm{cov}\left(\frac{\bm{\xi}}{\|\bm{\xi}\|}\right) $ is a diagonal matrix.
It is equivalent to $ \forall i \neq j, \mathrm{cov}\left(\frac{\xi_{i}}{\|\bm{\xi}\|}, \frac{\xi_{j}}{\|\bm{\xi}\|}\right) = 0 $.
For another, the component of the diagonal components of $ \mathrm{cov}\left(\frac{\bm{\xi}}{\|\bm{\xi}\|}\right) $ have the form like \eqref{equ:cov_xi}.
It is because $ \mathrm{var}\left(\frac{\xi_{2}}{\|\bm{\xi}\|}\right) = \mathrm{var}\left(\frac{\xi_{i}}{\|\bm{\xi}\|}\right), i\geq 2 $,
due to the i.i.d. of $ \xi_{i}, i\geq 2 $.

Furthermore, 
we show that $ \mathrm{var}\left(\frac{\xi_{1}}{\|\bm{\xi}\|}\right) $ and $ \mathrm{var}\left(\frac{\xi_{2}}{\|\bm{\xi}\|}\right) $ are all based on $ \mathbb{E} \frac{\xi_{1}^{2}}{\|\bm{\xi}\|^{2}} $.
By the definition of variance, we have
\begin{align}
\mathrm{var}\left(\frac{\xi_{1}}{\|\bm{\xi}\|}\right)
=&\  
\mathbb{E} \frac{\xi_{1}^{2}}{\|\bm{\xi}\|^{2}} - \left(\mathbb{E} \frac{\xi_{1}}{\|\bm{\xi}\|}\right)^{2}
,
\label{equ:var_xi_1_normed_append}
\end{align}
in which $ \mathbb{E} \frac{\xi_{1}}{\|\bm{\xi}\|} = \varrho_{n}\left(\frac{\|\bm{\nu}\|}{\lambda}\right) $.
Note that $ \sum_{i=1}^{n} \frac{\xi_{i}^{2}}{\|\bm{\xi}\|^{2}} = 1 $, so $ \mathbb{E} \frac{\xi_{2}^{2}}{\|\bm{\xi}\|^{2}} = \frac{1}{n-1}\left ( 1 - \mathbb{E} \frac{\xi_{1}^{2}}{\|\bm{\xi}\|^{2}} \right ) $.
Supplemented with $ \mathbb{E} \frac{\xi_{2}}{\|\bm{\xi}\|} = 0$, we have
\begin{align}
\mathrm{var}\left(\frac{\xi_{2}}{\|\bm{\xi}\|}\right)
=&\
\frac{1}{n-1}
\left(
1 - \mathbb{E} \frac{\xi_{1}^{2}}{\|\bm{\xi}\|^{2}}
\right)
.
\label{equ:var_xi_2_normed_append}
\end{align}

Below,
the closed-form expression of $ \mathbb{E} \frac{\xi_{1}^{2}}{\|\bm{\xi}\|^{2}} = 1
-
\frac{n-1}{n}
M\left(1, \frac{n}{2} + 1,  -\frac{\|\bm{\nu}\|^{2}}{2\lambda^{2}}\right)
$ is solved in three steps.

First, we give the simplified integral expression of $ \mathbb{E} \frac{\xi_{1}^{2}}{\|\bm{\xi}\|^{2}} $.
Define the two independent random variables:
$
X := \xi_{1}^{2},
Y := \xi_{2}^{2}+\cdots+\xi_{n}^{2},
$
where $ X \sim 
\chi^{2}(1, \frac{\|\bm{\nu}\|^{2}}{\lambda^{2}}) $ has the non-central chi-squared distribution of the degrees of freedom $ 1 $ and $ Y \sim \chi^{2}(n-1, 0) $ has the central one of $ n-1 $.
Their p.d.f.'s are
\begin{align*}
f_{X}(x)
=&\ 
\frac{1}{\sqrt{2\pi x}} 
\mathrm{e}^{-\frac{x}{2}}
\mathrm{e}^{-\frac{\|\bm{\nu}\|^{2}}{2\lambda^{2}}}
\cosh\left (\frac{\|\bm{\nu}\|}{\lambda}\sqrt{x}\right )
\mathbbm{1}\{x>0\}
,
\\
f_Y(y)
=&\
\frac{1}{2^{\frac{n-1}{2}} \Gamma\left(\frac{n-1}{2}\right)}
y^{\frac{n-3}{2} } e^{-\frac{y}{2}}
\mathbbm{1}\{y>0\}
.
\end{align*}
By the integral transformation 
$ x = rt $ and $ y = r - rt $, we have
\begin{align}
\mathbb{E} \frac{\xi_{1}^{2}}{\|\bm{\xi}\|^{2}}
=&\
\iint_{(0, \infty)^{2}}
\frac{x}{x+y}
f_{X}(x)f_Y(y)
\mathrm{d}x
\mathrm{d}y
= 
\int_{0}^{\infty}
r
\mathrm{d}
r
\int_{0}^{1} 
tf_{X}(rt)f_Y(r - rt)
\mathrm{d}
t,
\notag
\\
=&\ 
\frac{{e}^{-\frac{\|\bm{\nu}\|^{2}}{2\lambda^{2}}}}{\sqrt{\pi} 2^{\frac{n}{2}} \Gamma\left(\frac{n-1}{2}\right)}
\int_{0}^{\infty}
r^{\frac{n-2}{2}}
e^{-\frac{r}{2}}
{\rm d}r
\int_{0}^{1} 
t^{\frac{1}{2}}
(1-t)^{\frac{n-3}{2}}
\cosh\left (\frac{\|\bm{\nu}\|}{\lambda}\sqrt{rt}\right )
{\rm d}t
\label{equ:int_r_t}
.
\end{align}

Second, the integral can be simplified by modified Bessel function $ I_{\alpha}(\cdot) $.
In detail, from \citet[(9.6.18)]{AS1964handbook} we can derive that
\begin{align}
\int_{0}^{1} 
t^{\frac{1}{2}}
(1-t)^{\frac{n-3}{2}} 
\cosh\left (\frac{\|\bm{\nu}\|}{\lambda}\sqrt{rt}\right )
{\rm d}t
=&\ 
\frac{\sqrt{\pi}2^{\frac{n-2}{2}}}{\left(\frac{\|\bm{\nu}\|}{\lambda}\right)^{\frac{n}{2}}r^{\frac{n}{4}}}
\Gamma\left(\frac{n-1}{2}\right)
\left[
I_{\frac{n}{2}}\left (\frac{\|\bm{\nu}\|}{\lambda} \sqrt{r}\right )
+
\frac{\|\bm{\nu}\|}{\lambda} \sqrt{r}
I_{\frac{n+2}{2}}\left (\frac{\|\bm{\nu}\|}{\lambda} \sqrt{r}\right )
\right]
\label{equ:int_t_solved}
.
\end{align}
Substituting \eqref{equ:int_t_solved} into \eqref{equ:int_r_t}, we have
\begin{align}
\mathbb{E} \frac{\xi_{1}^{2}}{\|\bm{\xi}\|^{2}}
=&\ 
\frac{e^{-\frac{\|\bm{\nu}\|^{2}}{2\lambda^{2}}}}{2\left (\frac{\|\bm{\nu}\|}{\lambda}\right )^{\frac{n}{2}}}
\int_{0}^{\infty}
e^{-\frac{r}{2}}
r^{\frac{n}{4}-1}
\left[
I_{\frac{n}{2}}\left (\frac{\|\bm{\nu}\|}{\lambda} \sqrt{r}\right )
+
\frac{\|\bm{\nu}\|}{\lambda} \sqrt{r}
I_{\frac{n+2}{2}}\left (\frac{\|\bm{\nu}\|}{\lambda} \sqrt{r}\right )
\right]
\mathrm{d}r
.
\label{equ:int_r}
\end{align}

Third, it can be concisely expressed by confluent hypergeometric function. 
Specifically,
by the integral transformation $ x = \sqrt{r} $ and \citet[(9.6.3) and (11.4.28)]{AS1964handbook}, we have
\begin{align}
\int_{0}^{\infty}
e^{-\frac{x^{2}}{2}}
x^{\frac{n}{2} - 1}
I_{\frac{n}{2}}\left (\frac{\|\bm{\nu}\|}{\lambda} x\right )
\mathrm{d}x
=&\ 
\frac
{
	\left(\frac{\|\bm{\nu}\|}{\lambda}\right)^{\frac{n}{2}}
}
{n}
M
\left(
\frac{n}{2},
\frac{n}{2} + 1,
\frac{\|\bm{\nu}\|^{2}}{2\lambda^{2}}
\right)
,
\label{equ:int_I_to_M1}
\\
\int_{0}^{\infty}
e^{-\frac{x^{2}}{2}}
x^{\frac{n}{2}}
I_{\frac{n}{2}+1}\left (\frac{\|\bm{\nu}\|}{\lambda} x\right )
\mathrm{d}x
=&\ 
\frac
{
	\left (\frac{\|\bm{\nu}\|}{\lambda}\right )^{\frac{n}{2}+1}}
{n+2}
M
\left(
\frac{n}{2}+1,
\frac{n}{2}+2,
\frac{\|\bm{\nu}\|^{2}}{2\lambda^{2}}
\right)
.
\label{equ:int_I_to_M2}
\end{align}
Substituting \eqref{equ:int_I_to_M1} and \eqref{equ:int_I_to_M2} into \eqref{equ:int_r}, we have
\begin{align*}
\mathbb{E} \frac{\xi_{1}^{2}}{\|\bm{\xi}\|^{2}}
=&\ 
e^{-\frac{\|\bm{\nu}\|^{2}}{2\lambda^{2}}}
\left[
\frac
{
	1
}
{n}
M
\left(
\frac{n}{2},
\frac{n}{2} + 1,
\frac{\|\bm{\nu}\|^{2}}{2\lambda^{2}}
\right)
+
\frac
{
	\left ( \frac{\|\bm{\nu}\|}{\lambda} \right )^{2}}
{n+2}
M
\left(
\frac{n}{2}+1,
\frac{n}{2}+2,
\frac{\|\bm{\nu}\|^{2}}{2\lambda^{2}}
\right)
\right]
.
\end{align*}
By \citet[(13.1.27)]{AS1964handbook}, we can eliminate
$ e^{-\frac{\|\bm{\nu}\|^{2}}{2\lambda^{2}}} $:
\begin{align*}
\mathbb{E} \frac{\xi_{1}^{2}}{\|\bm{\xi}\|^{2}}
=&\ 
\frac
{1}
{n}
M
\left(
1,
\frac{n}{2} + 1,
-\frac{\|\bm{\nu}\|^{2}}{2\lambda^{2}}
\right)
+
\left ( \frac{\|\bm{\nu}\|}{\lambda} \right )^{2}
\frac
{1}
{n+2}
M
\left(
1,
\frac{n}{2}+2,
-\frac{\|\bm{\nu}\|^{2}}{2\lambda^{2}}
\right)
.
\end{align*}
By \citet[(13.4.4)]{AS1964handbook}, we can combine the two confluent hypergeometric function into a single one:
\begin{align*}
\mathbb{E} \frac{\xi_{1}^{2}}{\|\bm{\xi}\|^{2}}
=&\ 
1
-
\frac{n-1}{n}
M\left(1, \frac{n}{2} + 1,  -\frac{\|\bm{\nu}\|^{2}}{2\lambda^{2}}\right)
.
\end{align*}

Consequently, the closed-form expressions of
$ \mathrm{var}\left(\frac{\xi_{1}}{\|\bm{\xi}\|}\right) $
and 
$ \mathrm{var}\left(\frac{\xi_{2}}{\|\bm{\xi}\|}\right) $ in \eqref{equ:var_xi_1_normed_append} and \eqref{equ:var_xi_2_normed_append}
are solved:
\begin{align*}
\mathrm{var}\left(\frac{\xi_{1}}{\|\bm{\xi}\|}\right)
=&\  
1
-
\frac{n-1}{n}
M\left(1, \frac{n}{2} + 1,  -\frac{\|\bm{\nu}\|^{2}}{2\lambda^{2}}\right)
- 
\left[\varrho_{n}\left(\frac{\|\bm{\nu}\|}{\lambda}\right)\right ]^{2}
,
\\
\mathrm{var}\left(\frac{\xi_{2}}{\|\bm{\xi}\|}\right)
=&\
\frac{1}{n}
M\left(1, \frac{n}{2} + 1,  -\frac{\|\bm{\nu}\|^{2}}{2\lambda^{2}}\right)
,
\end{align*}
We regard $ \frac{\|\bm{\nu}\|}{\lambda} $ as $ x $ and define them as $ f_{n}(x) $ and $ g_{n}(x) $, so we can get \eqref{equ:f} and \eqref{equ:g}.
\end{proof}

\subsection{The Proof of Theorem~\ref{thm:cov_X_Z}}
\label{proof:cov_X_Z}

We first list and prove Lemma~\ref{lem:OTZ}, 
and then prove Theorem~\ref{thm:cov_X_Z}.

\begin{lemma}\label{lem:OTZ}
Given $ \bm{Z} \sim N(\bm{\mu}, \sigma^{2} \Xi), \sigma>0 $, 
where $ \Xi $ is defined in \eqref{equ:Xi},
there exists an orthogonal matrix $ O $ such that
\begin{align}
O^{\T}
P\bm{Z}
\sim
N
\left (
\begin{pmatrix}
\|P\bm{\mu}\| \\
\bm{0} _{n-1} \\
\end{pmatrix}
,
\begin{pmatrix}
(1-\rho)\sigma^{2}I_{n-1} & \\
& 0
\end{pmatrix}
\right )
.
\label{equ:OTZ}
\end{align}
$ O := VQ $. 
$ V $ is given in \eqref{equ:V_short}. 
$ Q := \begin{pmatrix}
Q_{n-1} & \\
& 1
\end{pmatrix} $ where $ Q_{n-1} $ is an $ (n-1) $-dim orthogonal matrix whose first column is the unitization of the first $ (n-1) $ components of $ V^{\T}\bm{\mu} $.
\end{lemma}

\begin{proof}[The proof of Lemma~\ref{lem:OTZ}]
By the definition of $ V $ and $ P $, 
$ V^{\T}P\bm{\mu} $ can be expressed as
\begin{align*}
V^{\T}P\bm{\mu} 
=&\
\begin{pmatrix}
I_{n-1} & \\
& 0
\end{pmatrix}
V^{\T}
\bm{\mu}
.
\end{align*}
We have the expectation in \eqref{equ:OTZ}.
\begin{align*}
\mathbb{E}
\left (
O^{\T}P\bm{Z}
\right )
=&\
\begin{pmatrix}
Q_{n-1}^{\T} & \\
& 0
\end{pmatrix}
V^{\T} \bm{\mu}
=
\begin{pmatrix}
\|P\bm{\mu}\| \\
\bm{0}_{n-1}
\end{pmatrix}
.
\end{align*}
The last equation is because the construction of the first column of $ Q_{n-1} $ is the unitization of the first $ (n-1) $ components of $ V^{\T}\bm{\mu} $.

Furthermore, we show the expression of the covariance matrix in \eqref{equ:OTZ}.
\begin{align*}
\mathrm{cov}
\left (
O^{\T}
P
\bm{Z}
\right )
=&\
Q^{\T}
\begin{pmatrix}
I_{n-1} & \\
& 0
\end{pmatrix}
V^{\T}
\Xi
V
\begin{pmatrix}
I_{n-1} & \\
& 0
\end{pmatrix}
Q
,
\\
=&\
\begin{pmatrix}
Q_{n-1}^{\T} & \\
& 0
\end{pmatrix}
\left [
\sigma^{2}
\begin{pmatrix}
(1-\rho)I_{n-1} & \\
& 1 + (n-1)\rho
\end{pmatrix}
\right ]
\begin{pmatrix}
Q_{n-1} & \\
& 0
\end{pmatrix}
,
\\
=&\
\begin{pmatrix}
(1-\rho)\sigma^{2}I_{n-1} & \\
& 0
\end{pmatrix}
.
\end{align*}

Therefore, the expectation and the covariance matrix of $ O^{\T}\bm{Z} $ are proved.
\end{proof}

\begin{proof}[The proof of Theorem~\ref{thm:cov_X_Z}]
	
	We first convert the expression of $ \mathrm{var}\left(T_{\bm{\chi}(\bm{Z})}(\bm{\theta})\right) $ by Lemma~\ref{lem:OTZ}, 
	and then prove the theorem by Theorem~\ref{thm:covariance_matrix}.
	
	First, 
	we show that 
	$ 	\mathrm{var}\left(T_{\bm{\chi}(\bm{Z})}(\bm{\theta})\right) $
	can be expressed by $ 	\mathrm{cov}
	\left(
	\frac{\bm{\eta}_{n-1}}{\|\bm{\eta}_{n-1}\|}
	\right) $, where $ \bm{\eta}_{n-1} $ is independent and has identical variance.
	In detail,
	by Lemma~\ref{lem:OTZ},
	there exists an orthogonal matrix $ O $ such that
	\begin{align*}
	\mathrm{var}\left(T_{\bm{\chi}(\bm{Z})}(\bm{\theta})\right)
	=&\ 
	\bm{\theta}^{\T}
	\mathrm{cov}
	\left(
	\frac{P\bm{Z}}{\|P\bm{Z}\|}
	\right)
	\bm{\theta}
	=
	\left (
	O^{\T}
	\bm{\theta}
	\right )^{\T}
	\mathrm{cov}
	\left(
	\frac{O^{\T}P\bm{Z}}{\|O^{\T}P\bm{Z}\|}
	\right)
	O^{\T}
	\bm{\theta}
	.
	\end{align*}
	We denote $ 
	\bm{\eta} := O^{\T}P\bm{Z},
	\bm{\delta} := O^{\T}\bm{\theta}.
	 $
	Then the last components of $ \bm{\eta} $ and $ \bm{\delta} $ are always $ 0 $, 
	so we have
	\begin{align*}
	\mathrm{var}\left(T_{\bm{\chi}(\bm{Z})}(\bm{\theta})\right)
	=&\
	\begin{pmatrix}
	\bm{\delta}_{n-1} \\
	0
	\end{pmatrix}^{\T}
	\begin{pmatrix}
	\mathrm{cov}
	\left(
	\frac{\bm{\eta}_{n-1}}{\|\bm{\eta}_{n-1}\|}
	\right) & \\
	& 0
	\end{pmatrix}
	\begin{pmatrix}
	\bm{\delta}_{n-1} \\
	0
	\end{pmatrix}
	,
	\\
	=&\
	\bm{\delta}_{n-1}^{\T}
	\mathrm{cov}
	\left(
	\frac{\bm{\eta}_{n-1}}{\|\bm{\eta}_{n-1}\|}
	\right)
	\bm{\delta}_{n-1}
	,
	\end{align*}
	where $ \bm{\eta}_{n-1} $ is the first $ (n-1) $ components of $ \bm{\eta} $ and $ \bm{\delta}_{n-1} $ is the first $ (n-1) $ components of $ \bm{\delta} $.
	The distribution of $ \bm{\eta}_{n-1} $ is
	\begin{align*}
	\bm{\eta}_{n-1}
	\sim
	N
	\left (
	\begin{pmatrix}
	\|P\bm{\mu}\| \\
	\bm{0} _{n-2} \\
	\end{pmatrix}
	,
	(1-\rho)\sigma^{2}I_{n-1}
	\right ).
	\end{align*}
	
	Second,
	we give the expression of 
	$ \mathrm{var}\left(T_{\bm{\chi}(\bm{Z})}(\bm{\theta})\right) $.
	By Theorem~\ref{thm:covariance_matrix},
	we have 
	\begin{align*}
	\mathrm{cov}\left(\frac{\bm{\eta}_{n-1}}{\|\bm{\eta}_{n-1}\|}\right)
	=&\ 
	\begin{pmatrix}
	f_{n-1}\left(\frac{\|P\bm{\mu}\|}{\sigma\sqrt{1-\rho}}\right)  &  \\
	& g_{n-1}\left(\frac{\|P\bm{\mu}\|}{\sigma\sqrt{1-\rho}}\right)I_{n-2}
	\end{pmatrix}
	,
	\end{align*}
	where $ f_{n-1} $ and $ g_{n-1} $ are in \eqref{equ:f} and \eqref{equ:g}.
	So
	\begin{align*}
	\mathrm{var}\left(T_{\bm{\chi}(\bm{Z})}(\bm{\theta})\right)
	=&\
	\bm{\delta}_{n-1}^{\T}
	\begin{pmatrix}
	f_{n-1}\left(\frac{\|P\bm{\mu}\|}{\sigma\sqrt{1-\rho}}\right)  & \\
	& g_{n-1}\left(\frac{\|P\bm{\mu}\|}{\sigma\sqrt{1-\rho}}\right)
	I_{n-2}
	\end{pmatrix}
	\bm{\delta}_{n-1}
	.
	\end{align*}
	Note that $ \bm{\delta}_{n-1} $ is a unit vector, so
	\begin{align*}
	\bm{\delta}_{n-1}^{\T}
	\begin{pmatrix}
	1 &  \\
	& 0_{n-2}   \\
	\end{pmatrix}
	\bm{\delta}_{n-1}
	+
	\bm{\delta}_{n-1}^{\T}
	\begin{pmatrix}
	0 &  \\
	& I_{n-2}   \\
	\end{pmatrix}
	\bm{\delta}_{n-1}
	=
	1
	.
	\end{align*}
	Therefore,
	\begin{align*}
	\mathrm{var}\left(T_{\bm{\chi}(\bm{Z})}(\bm{\theta})\right)
	=&\ 
	f_{n-1}\left(\frac{\|P\bm{\mu}\|}{\sigma\sqrt{1-\rho}}\right)\cdot
	\bm{\delta}_{n-1}^{\T}
	\begin{pmatrix}
	1 &  \\
	& 0_{n-2}   \\
	\end{pmatrix}
	\bm{\delta}_{n-1}
	+
	g_{n-1}\left(\frac{\|P\bm{\mu}\|}{\sigma\sqrt{1-\rho}}\right)\cdot
	\left[
	1
	-
	\bm{\delta}_{n-1}^{\T}
	\begin{pmatrix}
	1 &  \\
	& 0_{n-2}   \\
	\end{pmatrix}
	\bm{\delta}_{n-1}
	\right]
	,
	\\
	=&\
	\left[
	f_{n-1}\left(\frac{\|P\bm{\mu}\|}{\sigma\sqrt{1-\rho}}\right)
	-
	g_{n-1}\left(\frac{\|P\bm{\mu}\|}{\sigma\sqrt{1-\rho}}\right)
	\right]\cdot
	\bm{\delta}_{n-1}^{\T}
	\begin{pmatrix}
	1 &  \\
	& 0_{n-2}   \\
	\end{pmatrix}
	\bm{\delta}_{n-1}
	+
	g_{n-1}\left(\frac{\|P\bm{\mu}\|}{\sigma\sqrt{1-\rho}}\right)
	,
	\\
	=&\
	\left[
	f_{n-1}\left(\frac{\|P\bm{\mu}\|}{\sigma\sqrt{1-\rho}}\right)
	-
	g_{n-1}\left(\frac{\|P\bm{\mu}\|}{\sigma\sqrt{1-\rho}}\right)
	\right]\cdot
	\left(
	\begin{pmatrix}
	1 \\
	\bm{0}_{n-2}
	\end{pmatrix}^{\T}
	\bm{\delta}_{n-1}
	\right)^{2}
	+
	g_{n-1}\left(\frac{\|P\bm{\mu}\|}{\sigma\sqrt{1-\rho}}\right)
	.
	\end{align*}
	
	Last, we convert $ 	\begin{psmallmatrix}
	1 \\
	\bm{0}_{n-2}
	\end{psmallmatrix}^{\T}
	\bm{\delta}_{n-1} $ into
	$ \bm{\chi}(\bm{\mu})^{\T}\bm{\theta} $.
	To be more specific,
	\begin{align*}
	\begin{pmatrix}
	1 \\
	\bm{0}_{n-2}
	\end{pmatrix}^{\T}
	\bm{\delta}_{n-1}
	=&\
	\begin{pmatrix}
	1 \\
	\bm{0}_{n-1}
	\end{pmatrix}^{\T}
	\bm{\delta}
	=
	\begin{pmatrix}
	1 \\
	\bm{0}_{n-1}
	\end{pmatrix}^{\T}
	O^{\T}
	\bm{\theta}
	=
	\left (
	Q
	\begin{pmatrix}
	1 \\
	\bm{0}_{n-1}
	\end{pmatrix}
	\right )^{\T}
	V^{\T}
	\bm{\theta}
	.
	\end{align*}
	Note the construction of $ Q = \begin{pmatrix}
	Q_{n-1} & \\
	& 1
	\end{pmatrix} $, and the first column of $ Q_{n-1} $ is the unitization of the first $ (n-1) $ components of $ V^{\T}\bm{\mu} $,
	i.e.
	$
	\frac{1}{\left \|
		\begin{pmatrix}
		I_{n-1} & \bm{0}_{n-1}
		\end{pmatrix}V^{\T}\bm{\mu}\right \|} 
	\begin{pmatrix}
	I_{n-1} & \bm{0}_{n-1}
	\end{pmatrix}V^{\T}\bm{\mu}
	$,
	so
	\begin{align*}
	Q
	\begin{pmatrix}
	1 \\
	\bm{0}_{n-1}
	\end{pmatrix}
	=
	\frac{V^{\T}P\bm{\mu}}{\|V^{\T}P\bm{\mu}\|}
	.
	\end{align*}
	Therefore,
	\begin{align*}
	\left (
	Q
	\begin{pmatrix}
	1 \\
	\bm{0}_{n-1}
	\end{pmatrix}
	\right )^{\T}
	V^{\T}
	\bm{\theta}
	=&\
	\left (
	\frac{P\bm{\mu}}{\|P\bm{\mu}\|}
	\right )^{\T}
	\bm{\theta}
	=
	\bm{\chi}(\bm{\mu})^{\T}
	\bm{\theta}
	.
	\end{align*}
	
	Consequently,
	\begin{align*}
	\mathrm{var}\left(T_{\bm{\chi}(\bm{Z})}(\bm{\theta})\right)
	=&\
	\left[
	f_{n-1}\left(\frac{\|P\bm{\mu}\|}{\sigma\sqrt{1-\rho}}\right)
	-
	g_{n-1}\left(\frac{\|P\bm{\mu}\|}{\sigma\sqrt{1-\rho}}\right)
	\right]\cdot
	\left(
	\bm{\chi}(\bm{\mu})^{\T}\bm{\theta}
	\right)^{2}
	+
	g_{n-1}\left(\frac{\|P\bm{\mu}\|}{\sigma\sqrt{1-\rho}}\right)
	.
	\end{align*}
\end{proof}

\end{document}